\PassOptionsToPackage{unicode}{hyperref}
\PassOptionsToPackage{hyphens}{url}
\PassOptionsToPackage{dvipsnames,svgnames,x11names}{xcolor}
\documentclass[
  12pt]{article}

\usepackage{amsmath,amssymb}
\usepackage{iftex}
\ifPDFTeX
  \usepackage[T1]{fontenc}
  \usepackage[utf8]{inputenc}
  \usepackage{textcomp} 
\else 
  \usepackage{unicode-math}
  \defaultfontfeatures{Scale=MatchLowercase}
  \defaultfontfeatures[\rmfamily]{Ligatures=TeX,Scale=1}
\fi
\usepackage{lmodern}
\ifPDFTeX\else  
\fi
\IfFileExists{upquote.sty}{\usepackage{upquote}}{}
\IfFileExists{microtype.sty}{
  \usepackage[]{microtype}
  \UseMicrotypeSet[protrusion]{basicmath} 
}{}
\makeatletter
\@ifundefined{KOMAClassName}{
  \IfFileExists{parskip.sty}{%
    \usepackage{parskip}
  }{
    \setlength{\parindent}{0pt}
    \setlength{\parskip}{6pt plus 2pt minus 1pt}}
}{
  \KOMAoptions{parskip=half}}
\makeatother
\usepackage{xcolor}
\makeatletter
\ifx\paragraph\undefined\else
  \let\oldparagraph\paragraph
  \renewcommand{\paragraph}{
    \@ifstar
      \xxxParagraphStar
      \xxxParagraphNoStar
  }
  \newcommand{\xxxParagraphStar}[1]{\oldparagraph*{#1}\mbox{}}
  \newcommand{\xxxParagraphNoStar}[1]{\oldparagraph{#1}\mbox{}}
\fi
\ifx\subparagraph\undefined\else
  \let\oldsubparagraph\subparagraph
  \renewcommand{\subparagraph}{
    \@ifstar
      \xxxSubParagraphStar
      \xxxSubParagraphNoStar
  }
  \newcommand{\xxxSubParagraphStar}[1]{\oldsubparagraph*{#1}\mbox{}}
  \newcommand{\xxxSubParagraphNoStar}[1]{\oldsubparagraph{#1}\mbox{}}
\fi
\makeatother

\usepackage{longtable,booktabs,array}
\usepackage{calc} 
\usepackage{etoolbox}
\makeatletter
\patchcmd\longtable{\par}{\if@noskipsec\mbox{}\fi\par}{}{}
\makeatother
\IfFileExists{footnotehyper.sty}{\usepackage{footnotehyper}}{\usepackage{footnote}}
\makesavenoteenv{longtable}
\usepackage{graphicx}
\makeatletter
\def\maxwidth{\ifdim\Gin@nat@width>\linewidth\linewidth\else\Gin@nat@width\fi}
\def\maxheight{\ifdim\Gin@nat@height>\textheight\textheight\else\Gin@nat@height\fi}
\makeatother
\setkeys{Gin}{width=\maxwidth,height=\maxheight,keepaspectratio}
\makeatletter
\def\fps@figure{htbp}
\makeatother

\makeatletter
\@ifpackageloaded{caption}{}{\usepackage{caption}}
\AtBeginDocument{%
\ifdefined\contentsname
  \renewcommand*\contentsname{Table of contents}
\else
  \newcommand\contentsname{Table of contents}
\fi
\ifdefined\listfigurename
  \renewcommand*\listfigurename{List of Figures}
\else
  \newcommand\listfigurename{List of Figures}
\fi
\ifdefined\listtablename
  \renewcommand*\listtablename{List of Tables}
\else
  \newcommand\listtablename{List of Tables}
\fi
\ifdefined\figurename
  \renewcommand*\figurename{Figure}
\else
  \newcommand\figurename{Figure}
\fi
\ifdefined\tablename
  \renewcommand*\tablename{Table}
\else
  \newcommand\tablename{Table}
\fi
}
\@ifpackageloaded{float}{}{\usepackage{float}}
\floatstyle{ruled}
\@ifundefined{c@chapter}{\newfloat{codelisting}{h}{lop}}{\newfloat{codelisting}{h}{lop}[chapter]}
\floatname{codelisting}{Listing}

\makeatother
\makeatletter
\@ifpackageloaded{caption}{}{\usepackage{caption}}
\@ifpackageloaded{subcaption}{}{\usepackage{subcaption}}
\makeatother

\ifLuaTeX
  \usepackage{selnolig}  
\fi
\usepackage[]{natbib}
\usepackage{bookmark}

\IfFileExists{xurl.sty}{\usepackage{xurl}}{} 
\hypersetup{
  pdftitle={Title},
  pdfauthor={Author 1; Author 2},
  pdfkeywords={3 to 6 keywords, that do not appear in the title},
  colorlinks=true,
  linkcolor={blue},
  filecolor={Maroon},
  citecolor={Blue},
  urlcolor={Blue},
  pdfcreator={LaTeX via pandoc}}

\RequirePackage{bm}
\usepackage{amsthm}
\usepackage{ulem}
\usepackage[table]{xcolor}
\usepackage{multirow}
\usepackage{arydshln}
\usepackage{makecell}

\theoremstyle{plain}
\newtheorem{theorem}{Theorem}
\newtheorem{corollary}{Corollary}
\newtheorem{lemma}{Lemma}
\newtheorem{proposition}{Proposition}

\theoremstyle{remark}

\newtheorem{definition}{Definition}
\newtheorem{assumption}{Assumption}

\newcommand{\bbE}{{\mathbb{E}}}
\newcommand{\bbV}{{\text{var}}}
\newcommand{\cov}{{\text{cov}}}

\newcommand{\F}{{\text{F}}}
\newcommand{\FE}{{\tilde{\F}}}
\newcommand{\C}{{\text{C}}}
\newcommand{\CF}{{\text{CF}}}
\newcommand{\diag}{{\text{diag}}}

\newcommand{\PRIASV}{{\text{PRIASV}}}
\newcommand{\M}{{\text{M}}}
\newcommand{\TSCRFE}{{\text{S-CRFE}}}

\newcommand{\bomega}{{\bm \omega}}
\newcommand{\bx}{{\bf x}}

\newcommand{\bg}{{\bm g}}
\newcommand{\bv}{{\bf v}}
\newcommand{\bu}{{\bf u}}
\newcommand{\bw}{{\bf w}}
\newcommand{\bc}{{\bm c}}
\newcommand{\bq}{{\bm q}}
\newcommand{\be}{{\bm e}}
\newcommand{\bC}{{\bf C}}

\newcommand{\bZero}{{\boldsymbol 0}}
\newcommand{\bOne}{{\boldsymbol 1}}

\newcommand{\bA}{{\boldsymbol A}}
\newcommand{\bB}{{\boldsymbol B}}
\newcommand{\bD}{{\boldsymbol D}}

\newcommand{\bG}{{\boldsymbol G}}
\newcommand{\bM}{{\boldsymbol M}}
\newcommand{\bI}{{\boldsymbol I}}
\newcommand{\bQ}{{\boldsymbol Q}}

\newcommand{\bms}{{\boldsymbol s}}
\newcommand{\bS}{{\boldsymbol S}}

\newcommand{\bX}{{\boldsymbol X}}
\newcommand{\bY}{{\boldsymbol Y}}
\newcommand{\bV}{{\boldsymbol V}}
\newcommand{\bW}{{\boldsymbol W}}
\newcommand{\bU}{{\boldsymbol U}}
\newcommand{\bz}{{\boldsymbol z}}
\newcommand{\bZ}{{\boldsymbol Z}}

\newcommand{\bbeta}{{\boldsymbol \beta}}
\newcommand{\bmu}{{\boldsymbol \mu}}
\newcommand{\bepsilon}{{\boldsymbol \epsilon}}
\newcommand{\bvarepsilon}{{\boldsymbol \varepsilon}}

\newcommand{\bLambda}{{\boldsymbol \Lambda}}
\newcommand{\bGamma}{{\boldsymbol \Gamma}}
\newcommand{\btau}{{\boldsymbol \tau}}

\newcommand{\bzeta}{{\boldsymbol \zeta}}
\newcommand{\bEta}{{\boldsymbol \eta}}
\newcommand{\bxi}{{\boldsymbol \xi}}
\newcommand{\bvarphi}{{\boldsymbol \varphi}}
\newcommand{\bDelta}{{\boldsymbol \Delta}}

\newcommand{\cA}{{\mathcal A}}

\newcommand{\cD}{{\mathcal D}}

\newcommand{\cI}{{\mathcal I}}

\newcommand{\cQ}{{\mathcal Q}}

\newcommand{\cS}{{\mathcal S}}

\newcommand{\cN}{{\mathcal N}}

\newcommand{\bbP}{{\mathbb P}}
\newcommand{\bbR}{{\mathbb R}}

\newcommand{\anon}{1}

\begin{document}

\def\spacingset#1{\renewcommand{\baselinestretch}%
{#1}\small\normalsize} \spacingset{1}


\if1\anon
{
  \title{\bf Data-Adaptive Rerandomization for $2^K$ Factorial Designs}
  \author{Tingxuan Han\hspace{.2cm}\\
    Department of Statistics and Data Science, Tsinghua University\\
    and \\
    Ke Deng\thanks{Corresponding author: kdeng@tsinghua.edu.cn}\\
    Department of Statistics and Data Science, Tsinghua University}
  \maketitle
} \fi

\if0\anon
{
  \bigskip
  \bigskip
  \bigskip
  \begin{center}
    {\LARGE\bf Data-Adaptive  Rerandomization for $2^K$ Factorial Designs}
\end{center}
  \medskip
} \fi

\bigskip
\begin{abstract}
Factorial designs allow simultaneous estimation of multiple main effects and interactions, but covariate imbalance can substantially reduce estimation precision.
Existing rerandomization methods improve covariate balance yet do not fully exploit heterogeneous priorities across factorial effects or effect-specific covariate importance.
To address these limitations, this paper proposes a data-adaptive rerandomization framework for $2^K$ factorial designs. 
We first develop an oracle criterion that jointly incorporates researchers' priorities over factorial effects and effect-specific covariate importance, enabling precision gains with guaranteed lower bounds.
To make the oracle criterion implementable, we develop a data-adaptive procedure that learns effect-specific covariate importance from a random subset of units and applies an estimated oracle criterion to the remaining units. 
Unlike existing two-stage rerandomization methods for treatment-control experiments, our procedure accommodates multiple factorial effects and requires no auxiliary dataset.
Under a finite-population framework, we establish design-based asymptotic theory and show that the proposed procedure preserves the oracle design's precision-prioritization property and, under suitable conditions, achieves the same asymptotic precision as the oracle design.
Numerical studies demonstrate substantial efficiency gains over existing rerandomization methods.

\end{abstract}

\noindent%
{\it Keywords:} 
Causal inference;
Covariate adjustment;
Randomization inference; Adaptive design
\vfill

\newpage
\spacingset{1.8} 

\section{Introduction}\label{sec:intro}

Randomized experiments are widely regarded as the gold standard for causal inference \citep{Imbens2015}. 
When multiple treatment factors are of interest, factorial designs provide an efficient framework for simultaneously estimating main effects and interactions within a single experiment \citep{fisher1935,yates1937factorial,cochran1950experimental,box2005statistics,dasgupta2015causal,Branson2016,montgomery2017design}.
Although complete randomization balances covariates in expectation, substantial imbalance may arise in a realized assignment, reducing the estimation precision \citep{fisher1925,Holschuh1980,Wu1981,Urbach1985,Imai2008,cox1982randomization,Cox2009,Imbens2015,morgan2012rerandomization,li2018asymptotic}. 
This challenge is particularly acute in completely randomized factorial experiments (CRFE), where covariate balance must be achieved across many treatment combinations \citep{Branson2016,li2020rerandomization}.

A natural remedy is \textit{rerandomization}, which repeatedly randomizes until a prespecified balance criterion is satisfied. 
Formalized by \citet{morgan2012rerandomization} using the Mahalanobis distance and later extended to $2^K$ factorial designs as ReFM \citep{Branson2016,li2020rerandomization}, rerandomization improves precision by balancing covariates across treatment combinations.
For a recent and comprehensive review of rerandomization methods and their theoretical and practical developments, see \cite{junior2025does}.
However, ReFM treats all factorial effects and covariates equally, despite the fact that researchers may prioritize some effects over others and that covariates may vary in their relevance to different effects. 
Ignoring such heterogeneity can lead to substantial efficiency loss, particularly when the numbers of effects and covariates grow.

To accommodate heterogeneous priorities, \citet{li2020rerandomization} proposed tiered extensions of ReFM that organize both covariates and factorial effects into prespecified priority tiers.
However, these methods require the priority structure to be specified a priori and do not account for effect-specific covariate importance.
In the simpler two-arm setting, \citet{liu2025bayesian} moved beyond prespecified covariate priorities by developing the oracle rerandomization criterion (ReO), which quantifies covariate importance and achieves optimal asymptotic efficiency, together with a two-stage procedure that learns the required information from an auxiliary dataset.
Extending this idea to factorial designs is nontrivial because multiple factorial effects and their associated covariate importance must be accommodated simultaneously, while suitable auxiliary data may be unavailable in practice. 

To address these challenges, this paper makes three contributions.
First, we develop an oracle rerandomization criterion for $2^K$ factorial designs that jointly incorporates the importance of factorial effects and covariates, enabling precision gains to be allocated according to researchers' priorities.
Second, because the required covariate importance information is typically unknown, we propose a data-adaptive procedure that learns it from a random subset of units and constructs an estimated oracle criterion for the remaining units.
The procedure requires no auxiliary dataset and is therefore more broadly applicable in practice.
Third, we establish design-based asymptotic theory for both the oracle criterion and the data-adaptive procedure.
In particular, we characterize the asymptotic behavior of the data-adaptive procedure, establish its efficiency gains over existing methods, and show that, under suitable conditions, it achieves the same asymptotic precision as the infeasible oracle design.
Simulations further demonstrate its superior finite-sample performance.

The remainder of this paper is organized as follows. 
Section \ref{sec:notation} introduces the notation and reviews related methods.
Section \ref{sec:ReOMA} develops the oracle rerandomization criterion for $2^K$ factorial designs, and
Section \ref{sec:twostage} presents a data-adaptive procedure for approximating the oracle criterion using a single dataset. 
Numerical simulations and real-data applications are provided in Section \ref{sec:simu}, and Section \ref{sec:summary} concludes the paper.

\section{Notations and preliminaries}\label{sec:notation}

\subsection{Framework of potential outcomes}
Consider a finite population with $n$ units.
In a $2^K$ factorial experiment, there are $K$ binary treatment factors with two levels $-1$ and $+1$, resulting in a set of $Q=2^K$ treatment combinations $\cQ=\{-1,+1\}^K$.
For the $q$th treatment combination in $\cQ$, let $\bm\iota(q)=(\iota_1(q),\ldots,\iota_K(q))'\in\cQ$ denote its treatment combination vector, where $\iota_k(q)\in\{-1,+1\}$ is the level of treatment factor $k$.
Under the \emph{stable unit treatment value assumption} (SUTVA) \citep{Rubin1980},
let $Y_i(q)$ denote the potential outcome of unit $i$ under treatment combination $q$, and let $\bY_i=(Y_i(1),\ldots,Y_i(Q))'$ collect all its potential outcomes.
The finite-population means of potential outcomes under treatment combinations in $\cQ$ are defined as
$$\Bar{Y}(q) = \frac{1}{n}\sum_{i=1}^n Y_i(q),\ 1\leq q\leq Q,$$
which compose a mean vector $\Bar{\bY} = (\Bar{Y}(1),\ldots,\Bar{Y}(Q))'$.

{Following \citet{dasgupta2015causal}, we define the factorial effects as linear contrasts of the treatment-combination means.
To construct these contrasts, we first define the generating vector for the main effect of treatment factor $k$ as $\bg_k = (g_{k1},\ldots,g_{kQ})' = (\iota_k(1),\ldots,\iota_k(Q))'$.
The generating vector for an interaction effect is obtained by element-wise multiplication of the generating vectors for the corresponding main effects. 
For example, the generating vector for the interaction between factors $k$ and $k'$ is $\bg_k\circ\bg_{k'}$, where $\circ$ denotes element-wise multiplication.
This yields $F=2^K-1$ generating vectors, denoted by
$\bg_1,\ldots,\bg_F$.
Collecting these generating vectors as the columns of a matrix, we obtain a $Q\times F$ generating matrix $\bG=(\bg_1,\ldots,\bg_F)$, with $\bm b_q=(g_{1q},\ldots,g_{Fq})'$ as its $q$th row corresponding to the treatment combination $q$.
Thus, the columns of $\bG$ correspond to factorial effects, while the rows correspond to treatment combinations.
Table~\ref{tab:generate} illustrates this structure for $K=3$.

Given the generating matrix $\bG$,
the $f$th factorial effect  in the finite population is defined as
\[\tau_f=\frac{1}{2^{K-1}}\bg_f'\bar{\bY}=\frac{1}{2^{K-1}}\sum_{q=1}^Q g_{fq}\bar{Y}(q),\qquad f=1,\ldots,F,\]
which compose the vector of factorial effects $\btau=(\tau_1,\ldots,\tau_F)'$.
Similarly, for unit $i$, we define the individual factorial effect corresponding to the $f$th factorial effect as
$$\tau_{if}=\frac{1}{2^{K-1}}\bg_f'\bY_i=\frac{1}{2^{K-1}}\sum_{q=1}^Q g_{fq}Y_i(q),\qquad f=1,\ldots,F,$$
and the vector of individual factorial effects as
$\btau_i=(\tau_{i1},\ldots,\tau_{iF})'$. 
Moreover, define
$\bA_q=\frac{1}{2^{K-1}}\bm b_q=(a_{1q},\ldots,a_{Fq})'$, we have
\[\btau=\frac{1}{2^{K-1}}\bG'\bar{\bY}
=\frac{1}{2^{K-1}}\sum_{q=1}^Q \bm b_q\bar{Y}(q)
=\sum_{q=1}^Q\bA_q\bar{Y}(q),\]
\[\btau_i=\frac{1}{2^{K-1}}\bG'\bY_i
=\frac{1}{2^{K-1}}\sum_{q=1}^Q \bm b_qY_i(q)
=\sum_{q=1}^Q\bA_qY_i(q).\]
}

{
\spacingset{1}
\small
\begin{longtable}[!t]{@{}ccccc ccccccc@{}}
\caption{The $Q\times F$ generating matrix $\bG$ when $K=3$ (with $Q=8$ and $F=7$).
}\label{tab:generate}\tabularnewline
\toprule\noalign{}
  & \multicolumn{3}{c}{Main effects} && \multicolumn{4}{c}{Interaction effects} & Treatment & Average\\ \cmidrule(lr){2-4}\cmidrule(lr){6-9}

Order in $\cQ$ & 1 & 2 & 3 && 12 & 13 & 23 & 123 & combinations& outcomes\\
\midrule\noalign{}
\endfirsthead
\toprule\noalign{}
Effects & 1 & 3 & 2 && 13 & 12 & 23 & 123 & \\
\midrule\noalign{}
\endhead
\bottomrule\noalign{}
\endlastfoot

1 & $+1$ & $+1$ & $+1$ && $+1$ & $+1$ & $+1$ & $+1$ & $\bm b_1'$ & $\bar{Y}(1)$\\
2 & $+1$ & $+1$ & $-1$ && $+1$ & $-1$ & $-1$ & $-1$ & $\bm b_2'$ & $\bar{Y}(2)$\\
3 & $+1$ & $-1$ & $+1$ && $-1$ & $+1$ & $-1$ & $-1$ & $\bm b_3'$& $\bar{Y}(3)$\\
4 & $+1$ & $-1$ & $-1$ && $-1$ & $-1$ & $+1$ & $+1$ & $\bm b_4'$ & $\bar{Y}(4)$\\
5 & $-1$ & $+1$ & $+1$ && $-1$ & $-1$ & $+1$ & $-1$ & $\bm b_5'$ & $\bar{Y}(5)$\\
6 & $-1$ & $+1$ & $-1$ && $-1$ & $+1$ & $-1$ & $+1$ & $\bm b_6'$ & $\bar{Y}(6)$\\
7 & $-1$ & $-1$ & $+1$ && $+1$ & $-1$ & $-1$ & $+1$ & $\bm b_7'$ & $\bar{Y}(7)$\\
8 & $-1$ & $-1$ & $-1$ && $+1$ & $+1$ & $+1$ & $-1$ & $\bm b_8'$ & $\bar{Y}(8)$\\
\addlinespace[2pt]
\midrule
Generating vectors
& \raisebox{0.2ex}{$\bg_1$} & \raisebox{0.2ex}{$\bg_2$} & \raisebox{0.2ex}{$\bg_3$}
&& \raisebox{0.2ex}{$\bg_4$} & \raisebox{0.2ex}{$\bg_5$} & \raisebox{0.2ex}{$\bg_6$} & \raisebox{0.2ex}{$\bg_7$} & \raisebox{-0.3ex}{$\bG$} & \raisebox{-0.3ex}{$\bar{\bY}$}\\
\end{longtable}
}

\subsection{Factorial effect estimator and covariate imbalance}  

Let $\bZ=(Z_1,\ldots,Z_n)'\in \{1,\ldots,Q\}^n$ denote the treatment assignment vector, where $Z_i=q$ indicates that unit $i$ is assigned to treatment combination $q$. 
Under a \emph{completely randomized factorial experiment} (CRFE), exactly $n_q$ units are assigned to treatment combination $q$, with $\sum_{q=1}^Q n_q=n$, and we write $r_q=n_q/n$.
The assignment vector is drawn uniformly from
$$\mathcal{Z}=\left\{\bz\in \{1,\ldots,Q\}^n: \sum_{i=1}^n 1\{z_i=q\} = n_q, q=1,\ldots,Q\right\},$$
so that $\bbP(\bZ = \bz) = n_1!\ldots n_Q!/n!$ for $\bz\in\mathcal{Z}$.
The difference-in-means estimator for $\btau$ is
\begin{equation}\label{eq:estimator_fac}
\hat{\btau} = \sum_{q=1}^Q \bA_q \hat{\Bar{Y}}(q) = (\hat{\tau}_1,\ldots,\hat{\tau}_F)',
\end{equation}
where $\hat{\Bar{Y}}(q) = n_q^{-1}\sum_{i:Z_i=q} Y_i(q)$ is the observed mean outcome under  combination $q$.

Suppose that each unit $i$ is associated with a $p$-dimensional covariate vector $\bX_i\in\mathbb{R}^p$. 
For $q=1,\ldots,Q$, let $\hat{\Bar{\bX}}(q) = n_q^{-1}\sum_{i:Z_i=q} \bX_i$ denote the sample mean of the covariates among units assigned to treatment combination $q$.
Define the covariate mean-difference vector as
\begin{equation}\label{eq:CovariateMeanDifference}
\hat{\btau}_{\bx} = \sum_{q=1}^Q \bA_q \otimes \hat{\Bar{\bX}}(q) = (\hat{\btau}_{\bx,1}',\ldots,\hat{\btau}_{\bx,F}')' \in \bbR^{Fp},
\end{equation}
where $\otimes$ denotes the Kronecker product, and
$\hat{\btau}_{\bx,f} = \sum_{q=1}^Q a_{fq}\hat{\bar{\bX}}(q) \in \bbR^{p}$
is the covariate mean-difference vector associated with the $f$th factorial effect.
Let $\bar{\bX} = n^{-1}\sum_{i=1}^n \bX_i$ denote the finite-population mean of the covariates. 
It follows that
$\sum_{q=1}^Q \bA_q \otimes \bar{\bX} =\bZero$.

\subsection{Additional notations}

For the finite population $\{\bX_i,\bY_i\}_{i=1}^n$ and $q \in \{1,\ldots,Q\}$, define 
$S_{qq} = (n-1)^{-1}\sum_{i=1}^n (Y_i(q) - \bar{Y}(q))^2$,
$\bS_{\btau\btau} = (n-1)^{-1}\sum_{i=1}^n (\btau_i - \btau)(\btau_i - \btau)'$, 
$\bS_{\bx\bx} = (n-1)^{-1}\sum_{i=1}^n (\bX_i - \bar\bX)(\bX_i - \bar\bX)'$, and 
$\bS_{q,\bx} = (n-1)^{-1}\sum_{i=1}^n (Y_i(q) - \bar{Y}(q))(\bX_i - \bar\bX)'$.
The factorial effects are said to be additive if all individual factorial effect vectors are identical, i.e., $\btau_i \equiv \btau$, or equivalently,
$\bS_{\btau\btau}=\bZero$ \citep{li2020rerandomization}. 
Under additivity,
$S_{qq} \equiv S_{11}$, and $\bS_{q,\bx} \equiv \bS_{1,\bx}$.

In what follows, $\bA_n \overset{\cdot}{\sim} \bB_n$ denotes equality in asymptotic distribution, and $\bbV_a$ denotes the variance of an asymptotic distribution. 
For matrices $\bM_1$ and $\bM_2$, $\bM_1\le \bM_2$ means that $\bM_2-\bM_1$ is positive semidefinite. For a positive semidefinite matrix $\bM\in\mathbb{R}^{m\times m}$ of rank $l_0$, let $\bM_l^{1/2}\in\mathbb{R}^{m\times l}$ $(l \ge l_0)$ denote any matrix satisfying $(\bM^{1/2}_l)(\bM^{1/2}_l)' = \bM$.
When $l=m$, $\bM^{1/2}$ denotes the unique symmetric positive semidefinite square root of $\bM$.

\subsection{Asymptotic properties of CRFE}
Under the completely randomized factorial experiment (CRFE), by Theorem 3 of \citet{li2017general}, the random vector $\sqrt{n}(\hat{\btau}'-\btau',\hat{\btau}_{\bx}')'$ has mean $\bZero$ and covariance matrix given by 
{\setlength{\arraycolsep}{3pt}
\begin{equation}\label{eq:V_fac}
\begin{split}
\bV &\equiv \left(\begin{array}{cc}
\bV_{\btau\btau} & \bV_{\btau\bx} \\
\bV_{\bx\btau} & \bV_{\bx\bx}
\end{array}\right)
= \sum_{q=1}^Q \frac{1}{r_q}\left(\begin{array}{cc}
\bA_q \bA_q' S_{qq} & (\bA_q \bA_q')\otimes \bS_{q,\bx} \\
(\bA_q \bA_q')\otimes \bS_{\bx,q} & (\bA_q \bA_q')\otimes \bS_{\bx\bx}
\end{array}\right) - \left(\begin{array}{cc}
\bS_{\btau\btau} & \bZero_{F \times Fp} \\
\bZero_{Fp \times F} & 
\bZero_{Fp \times Fp}
\end{array}\right).
\end{split}
\end{equation} }
When the effects are additive, then 
\[\begin{split}
\bV &\equiv \left(\begin{array}{cc}
\bV_{\btau\btau} & \bV_{\btau\bx} \\
\bV_{\bx\btau} & \bV_{\bx\bx}
\end{array}\right)
= \left(\begin{array}{cc}
\left(\sum_{q=1}^Q r_q^{-1} \bA_q \bA_q'\right) S_{11} & \left(\sum_{q=1}^Q r_q^{-1} \bA_q \bA_q'\right) \otimes \bS_{1,\bx} \\
\left(\sum_{q=1}^Q r_q^{-1} \bA_q \bA_q'\right) \otimes \bS_{\bx,1} & \left(\sum_{q=1}^Q r_q^{-1} \bA_q \bA_q'\right) \otimes \bS_{\bx\bx}
\end{array}\right).
\end{split}\] 
Theorem 5 of \citet{li2017general} shows that, under CRFE, when Assumption~\ref{assump:regularity_fac} below holds, $\sqrt{n}(\hat{\btau}'-\btau',\hat{\btau}_{\bx}')'$ has the following asymptotic sampling distribution:
\begin{equation}\label{eq:asymp_fac}
\sqrt{n}\begin{pmatrix}
\hat{\btau}-\btau \\
\hat{\btau}_\bx
\end{pmatrix}
\overset{\cdot}{\sim}
\cN(\bZero,\bV).
\end{equation}

\begin{assumption}\label{assump:regularity_fac}
As $n\to\infty$, for each $1\le q\le Q$, the following conditions hold:
(1) the sample proportion $r_q=n_q/n$ converges to a positive limit;
(2) the finite-population variances and covariance matrices $S_{qq}$, $\bS_{\btau\btau}$, $\bS_{\bx\bx}$, and $\bS_{q,\bx}$ have finite limits, $\bS_{\bx\bx}$ and its limit $\bS_{\bx\bx,\infty}$ are nonsingular;
(3) $\max_{1\leq i\leq n} |Y_i(q)-\Bar{Y}(q)|^2/n \rightarrow 0$ and $\max_{1\leq i\leq n}||\bX_i - \Bar{\bX}||^2_2/n\rightarrow0$.
\end{assumption}

Define $\bV_{\btau\btau}^{||} = \bV_{\btau\bx}\bV_{\bx\bx}^{-1}\bV_{\bx\btau}$, and $\bV_{\btau\btau}^{\perp} = \bV_{\btau\btau} - \bV_{\btau\btau}^{||}$.
Theorem 1 in \cite{li2020rerandomization} implies that,
under CRFE, the linear projection of $\hat{\btau}-\btau$ onto $\hat{\btau}_{\bx}$ is
\begin{equation}\label{eq:project}
\sqrt{n}(\hat{\btau}-\btau) = \sqrt{n}\bV_{\btau\bx}\bV_{\bx\bx}^{-1} \hat{\btau}_\bx + \bm r,
\end{equation}
where $\bm r = (r_1,\ldots,r_F)'$ is the residual vector,
$\cov(\sqrt{n}\bV_{\btau\bx}\bV_{\bx\bx}^{-1} \hat{\btau}_\bx) = \bV_{\btau\btau}^{||}$,
$\cov(\bm r) = \bV_{\btau\btau}^{\perp}$, and 
$\cov(\hat{\btau}_\bx,\bm r) = \bZero$.
Let $V_{\tau_f\tau_f}$ and $V_{\tau_f\tau_f}^{\parallel}$ denote the $f$th diagonal elements of $\bV_{\btau\btau}$ and $\bV_{\btau\btau}^{\parallel}$. 
By Corollary 1 of \citet{li2020rerandomization}, 
the squared multiple correlation between $\hat{\tau}_f$ and $\hat{\btau}_{\bx}$ is
\begin{equation}\label{eq:R2f}
R^2_f = \text{corr}^2(\hat{\tau}_f,\hat{\btau}_\bx) = \frac{V_{\tau_f\tau_f}^{||}}{V_{\tau_f\tau_f}}.
\end{equation}

\subsection{Rerandomization based on Mahalanobis distance}

To improve covariate balance,
\citet{li2020rerandomization} proposed ReFM, which accepts a treatment assignment if
\begin{equation}\label{eq:ReFM}
\phi_M(\sqrt{n}\hat{\btau}_{\bx},\bV_{\bx\bx}) = 
1\{(\sqrt{n}\hat{\btau}_{\bx})'\bV_{\bx\bx}^{-1}(\sqrt{n}\hat{\btau}_{\bx}) \leq \xi_{Fp,\alpha}\}=1,
\end{equation}
where $\xi_{Fp,\alpha}$ denotes the $\alpha$-quantile of $\chi^2_{Fp}$ distribution.
Under ReFM, 
\begin{equation}\label{eq:DisReFM}
\sqrt{n}(\hat{\btau}-\btau) \overset{\cdot}{\sim}
\big(\bV_{\btau\btau}^{\perp}\big)^{1/2}\bvarepsilon + \big(\bV_{\btau\btau}^{||}\big)_{Fp}^{1/2}\bzeta_{Fp,\alpha},
\end{equation}
where $\bvarepsilon \sim \mathcal{N}(\bm 0, \bm I_{F})$, $\bzeta_{Fp,\alpha} \sim \bm\eta|\bm\eta'\bm\eta \leq \xi_{Fp,\alpha}$, and $\bm\eta \sim \mathcal{N}(\bm 0,\bm I_{Fp})$ is independent of $\bvarepsilon$. 
Consequently, the asymptotic sampling covariance matrix of $\sqrt{n}\hat{\btau}$ is
$\lim_{n\to\infty} 
\left(\bV_{\btau\btau}^{\perp} + v_{Fp,\alpha}\bV_{\btau\btau}^{||}\right),$
where $v_{Fp,\alpha} = \bbP(\chi^2_{p+2} \leq \xi_{p,\alpha})/\bbP(\chi^2_{p} \leq \xi_{p,\alpha})$.
For the $f$th factorial effect,
\begin{equation}\label{eq:DisReFM_f}
\sqrt{n}(\hat{\tau}_f-\tau_f)
\overset{\cdot}{\sim}
\sqrt{V_{\tau_f\tau_f}}
\left(\sqrt{1-R^2_f}\cdot \varepsilon_0 + \sqrt{R^2_f}\cdot L_{Fp,\alpha}\right),
\end{equation}
where $\varepsilon_0\sim\mathcal{N}(0,1)$,
$L_{Fp,\alpha} \sim \eta_1 \mid \bm\eta'\bm\eta \le \xi_{Fp,\alpha}$, and $\bm\eta\sim\mathcal{N}(\bm 0,\bm I_{Fp})$ is independent of $\varepsilon_0$.
The percentage reduction in asymptotic sampling variance (PRIASV) of $\hat{\tau}_f$ is
\[\begin{split}
\PRIASV_{\M,f} = 100\times \left[1-\frac{\bbV_a(\sqrt{n}(\hat{\tau}_f - \tau_f) \mid  \phi_M = 1)}{\bbV_a(\sqrt{n}(\hat{\tau}_f - \tau_f))}\right]
= \lim_{n\to\infty} 100\times (1-v_{Fp,\alpha})R^2_f.
\end{split}\]
Because the Mahalanobis distance weights all coordinates equally after standardization, ReFM implicitly treats all factorial effects and covariates as equally important.

To accommodate heterogeneous effect importance, \citet{li2020rerandomization} proposed ReFMT$_\F$, which partitions factorial effects into importance tiers and constructs the rerandomization criterion from the Mahalanobis distances of the corresponding subvectors of $\hat{\btau}_{\bx}$.
ReFMT$_\CF$ further partitions covariates into tiers to account for heterogeneous covariate importance.
However, ReFMT$_\CF$ assumes a common covariate-tier structure across all factorial effects, although covariate relevance may vary by effect.
Moreover, the precision gain for a given effect depends on its correlations with effects in other tiers, making effect-specific improvements difficult to control. 
Finally, the covariate tiers must be specified a priori, which may be difficult without reliable prior information on covariate importance.

These limitations motivate a rerandomization framework that accommodates heterogeneous effect priority and effect-specific covariate importance without requiring prior knowledge of the latter. 
We first develop an oracle criterion that fully exploits both sources of information and then propose a practical data-adaptive procedure that approximates it by learning effect-specific covariate importance from a random subset of units.


\section{Oracle rerandomization for $2^K$ factorial designs}\label{sec:ReOMA}

\subsection{Effect priorities and covariate importance}

We begin by formalizing the two sources of information: researchers' priorities across factorial effects, and the importance of covariates for estimating each effect.

In practice, different factorial effects often have different scientific or practical importance.
To reflect researchers' priorities, let $\omega_f \geq 0$ denote the prespecified importance weight assigned to $\tau_f$, with larger values indicating a higher priority for precision improvement.
Without loss of generality, we relabel the effects so that $\omega_1 \geq \omega_2 \geq \cdots \geq \omega_F$, and hence $\btau=(\tau_1,\ldots,\tau_F)'$ is ordered by decreasing importance.

To characterize the importance of covariates for each factorial effect, consider the projection decomposition in \eqref{eq:project}. 
Under CRFE,
$\bB = \bV_{\bx\bx}^{-1}\bV_{\bx\btau} = (\bB_1,\ldots,\bB_F) \in \bbR^{Fp \times F}$ is the projection coefficient matrix of $\hat{\btau}-\btau$ onto $\hat{\btau}_{\bx}$.
Accordingly,
$\sqrt{n}(\hat{\btau}-\btau) = \sqrt{n}\bB'\hat{\btau}_\bx + \bm r$,
i.e.,
\begin{equation}\label{eq:projection}
\sqrt{n}(\hat{\tau}_f-\tau_f) = \sqrt{n}\bB_f'\hat{\btau}_\bx + r_f,\ f=1,\ldots,F.
\end{equation}
The vector $\bB_f$ characterizes the relevance of the covariates to the estimation of $\tau_f$, while $\bB$ summarizes the overall effect-specific covariate importance structure.
Under the oracle setting, $\bB$ is assumed to be known and is used to construct the rerandomization criterion.

\subsection{Admissible rerandomization criteria}

To search for an optimal rerandomization criterion, we define a family of admissible covariate balance criteria, denoted by $\Phi_\alpha$. 
Each criterion $\phi(\sqrt{n}\hat{\btau}_{\bx},\bV_{\bx\bx},\bB)$
is a function of the covariate mean-difference vector $\sqrt{n}\hat{\btau}_{\bx}$, its covariance matrix $\bV_{\bx\bx}$, and the importance matrix $\bB$, and satisfies the regularity conditions in Assumption~\ref{assump:criterion}.
For any $\phi\in\Phi_\alpha$, let
$\mathcal{B}_{\phi}
=
\{\bm\mu:\phi(\bm\mu,\bV_{\bx\bx},\bB)=1\}$
denote the acceptance region induced by the criterion. 
By Assumption~\ref{assump:criterion} (2), under CRFE and Assumption~\ref{assump:regularity_fac},
$\lim_{n\to\infty}
\bbP
(
\sqrt n\hat{\btau}_{\bx}
\in
\mathcal B_\phi
)
=
\alpha,$
so every criterion in $\Phi_\alpha$ has the same asymptotic acceptance probability, allowing us to select among them based on estimation efficiency.

\begin{assumption}\label{assump:criterion}
Criterion $\phi(\cdot,\cdot,\cdot)$ satisfies:
(1) $\phi(\cdot,\cdot,\cdot)$ is almost surely continuous.
(2) For $\bD\sim \mathcal{N}(\bm 0,\bV_{\bx\bx})$, $\mathbb{P}\{\phi(\bD,\bV_{\bx\bx},\bB)=1\}=\alpha>0$, and $\cov(\bD\mid \phi(\bD,\bV_{\bx\bx},\bB)=1)$ is a continuous function of $\bV_{\bx\bx}$ and $\bB$.
(3) For any $\bm\mu$, $\phi(\bm\mu,\bV_{\bx\bx},\bB)
=
\phi(-\bm\mu,\bV_{\bx\bx},\bB)$.
\end{assumption}

\subsection{A weighted variance criterion}
We now seek an optimal criterion within $\Phi_\alpha$ based on estimation efficiency.
When $K=1$, the $2^K$ factorial experiment reduces to a treatment--control experiment. In this setting, \citet{liu2025bayesian} derived the oracle rerandomization criterion ReO by minimizing the asymptotic variance of $\sqrt n(\hat\tau-\tau)$.
For $K>1$, multiple factorial effects must be estimated simultaneously. A natural extension is therefore to choose a rerandomization criterion that minimizes a weighted sum of their asymptotic sampling variances:
\[
\arg\min_{\phi\in\Phi_\alpha}
\sum_{f=1}^F
\omega_f\,\bbV_a(\hat{\tau}_f\mid \phi=1),
\]
where $\omega_f$ reflects the prespecified priority of the $f$th factorial effect.
Combining the projection decomposition in Eq.~\eqref{eq:projection} with the asymptotic independence of $\hat{\btau}_{\bx}$ and $\bm r$ from Eq.~\eqref{eq:asymp_fac}, the above objective is equivalent to
\begin{equation}\label{eq:weighted_sum}
\arg\min_{\phi\in\Phi_\alpha}
\sum_{f=1}^{F}
\omega_f\,
\bbV_a(\bB_f'\hat{\btau}_{\bx}\mid \phi=1).
\end{equation}
By Lemma A1 of \citet{liu2025bayesian}, the solution to Eq. \eqref{eq:weighted_sum} is
$1\left\{(\sqrt{n}\hat{\btau}_\bx)'\bB\bomega\bB'(\sqrt{n}\hat{\btau}_\bx) \leq a\right\}$,
where $\bm\omega=\operatorname{diag}\{\omega_1,\ldots,\omega_F\}$ is the diagonal matrix of effect-priority weights, and $a$ is the $\alpha$ quantile of $\bD'\bB\bomega\bB'\bD$,  $\bD\sim\mathcal{N}(\bZero,\bV_{\bx\bx})$. 

Though straightforward, this objective has an undesirable feature. 
Because $\bB_1'\hat{\btau}_{\bx},\ldots,\bB_F'\hat{\btau}_{\bx}$ are generally correlated, minimizing their weighted sum of variances may overweight directions shared across multiple factorial effects.
Moreover, effects with larger baseline variances tend to receive greater variance reduction. 
As a result, the allocation of precision gains is influenced not only by the prespecified weights $\omega_f$, but also by the covariance structure of the estimators.
When the goal is to improve the precision of individual factorial effect estimators while placing greater emphasis on more important effects, it is desirable to disentangle these two sources of influence.


\subsection{A refined optimization criterion}

To better align precision gains with researchers' priorities, we refine the objective in two steps. First, we orthogonalize
$\bB_1'\hat{\btau}_{\bx},
\ldots,
\bB_F'\hat{\btau}_{\bx}$
to remove the influence of their correlation structure. Second, we replace variances by relative variance reductions, thereby eliminating the implicit weighting induced by their original scales. Consequently, the contribution of each component to the optimization is determined solely by the prespecified weights \(\omega_f\).

We begin with the orthogonalization.
Let $\tilde{\tau}_{\bx,1} = \bB_1'\hat{\btau}_\bx$. 
For $f=2,\ldots,F$, define $\tilde{\tau}_{\bx,f}$ as the residual from the linear projection of $\bB_f'\hat{\btau}_{\bx}$ onto
$\left(\bB_1'\hat{\btau}_\bx,\ldots,\bB_{f-1}'\hat{\btau}_\bx\right)'$.
Thus, $\tilde{\tau}_{\bx,f}$ captures the information in $\bB_f'\hat{\btau}_{\bx}$ that is not contained in the preceding components.
Equivalently, 
\[\begin{split}
\tilde{\btau}_\bx = (\tilde{\tau}_{\bx,1},\ldots,\tilde{\tau}_{\bx,F})' = \bQ_\bB \bB'(\hat{\btau}_\bx-\btau_\bx),
\end{split}\]
where $\bQ_{\bB}$ denotes the transformation matrix from the Gram--Schmidt orthogonalization:
\begin{equation}\label{eq:QB}
\bQ_\bB = \left(\begin{array}{ccccc}
    1 & 0 & \cdots & 0 & 0 \\
    \Tilde{V}_{2,1} & 1 &  \cdots & 0 & 0\\
    \vdots & \vdots & \ddots & \vdots & \vdots \\
    \Tilde{V}_{F-1,1} & \Tilde{V}_{F-1,2} & \cdots & 1 & 0\\
    \Tilde{V}_{F,1} & \Tilde{V}_{F,2} & \cdots & \Tilde{V}_{F,F-1} & 1
\end{array}\right).
\end{equation}
For $f = 2,\ldots,F$, 
$(\tilde{V}_{f,1},\ldots,\tilde{V}_{f,f-1}) = - \bV_{\btau\btau}^{||}[f,\overline{f-1}]\left\{\bV_{\btau\btau}^{||}[\overline{f-1},\overline{f-1}]\right\}^{-1}$, where
$\bV_{\btau\btau}^{\parallel}[f,\overline{f-1}]$ is the submatrix of $\bV_{\btau\btau}^{\parallel}$ formed by the $f$th row and the first $f-1$ columns, and
$\bV_{\btau\btau}^{\parallel}[\overline{f-1},\overline{f-1}]$ denotes the leading $(f-1)\times(f-1)$ principal submatrix of $\bV_{\btau\btau}^{\parallel}$.
The above construction requires 
the following assumption.

\begin{assumption}\label{assump:V_tautau}
The matrices $\bV_{\btau\btau}^{\parallel}$ and its limit $\bV_{\btau\btau,\infty}^{\parallel}$ are nonsingular.
\end{assumption}

Under Assumption~\ref{assump:regularity_fac}, since $\bV_{\bx\bx}$ and its limit are positive definite, and
$\bV_{\btau\btau}^{\parallel}
=
\bB'\bV_{\bx\bx}\bB$,
Assumption~\ref{assump:V_tautau} is equivalent to requiring \(\bB\) and its limit \(\bB_\infty\) to have full column rank.
More generally, it requires the projected components
$\bB_1'\hat{\btau}_{\bx},
\ldots,
\bB_F'\hat{\btau}_{\bx}$
to be linearly independent in the sense that no nontrivial linear combination of them has zero variance.
This condition is mild
and is assumed throughout.


By construction, the components of
$\sqrt{n}\tilde{\btau}_{\bx}$ are mutually uncorrelated, and therefore,
\begin{equation}\label{eq:Lambda_fac}
\cov(\sqrt{n}\tilde{\btau}_\bx) = \diag\{\lambda_1,\ldots,\lambda_F\} = \bLambda
= \bQ_\bB\bV_{\btau\btau}^{||}\bQ_\bB',
\end{equation} 
where 
$\lambda_f = \bbV(\sqrt{n}\tilde{\tau}_{\bx,f})
= V_{\tau_f\tau_f}^{\parallel}-\bV_{\btau\btau}^{||}[f,\overline{f-1}]\left\{\bV_{\btau\btau}^{||}[\overline{f-1},\overline{f-1}]\right\}^{-1}\bV_{\btau\btau}^{||}[\overline{f-1},f]$.



Based on the orthogonalized components, we redefine the objective as maximizing the weighted sum of their percentage reductions in asymptotic sampling variance (PRIASV):
\begin{equation}\label{eq:new_target}
\arg\max_{\phi\in\Phi_\alpha}
\sum_{f=1}^F
\omega_f
\left(
1-
\frac{\bbV_a(\tilde{\tau}_{\bx,f}\mid \phi=1)}
{\bbV_a(\tilde{\tau}_{\bx,f})}
\right),
\end{equation}
where $\omega_1\geq\cdots\geq \omega_F$ are the prespecified effect-priority weights. 
The following theorem provides an analytic solution to this optimization problem.

\begin{theorem}\label{thm:ReOMA}
Under Assumptions~\ref{assump:regularity_fac} and~\ref{assump:V_tautau}, and assuming that the importance matrix $\bB$ is known, the rerandomization criterion 
\begin{equation}\label{eq:ReOMA}
\begin{split}
\phi_{\bomega}(\sqrt{n}\hat\btau_\bx,\bV_{\bx\bx},\bB) &= 1\{(\sqrt{n}\tilde\btau_\bx)'\bomega\bLambda^{-1}(\sqrt{n}\tilde\btau_\bx) \leq \xi_{\bomega,\alpha}\}\\
&= 1\{(\sqrt{n}\hat\btau_\bx)'\bB\bQ_\bB'\bomega\bLambda^{-1}\bQ_\bB\bB'(\sqrt{n}\hat\btau_\bx) \leq \xi_{\bomega,\alpha}\}
\end{split}
\end{equation}
belongs to $\Phi_\alpha$ and solves the optimization problem in Eq.~\eqref{eq:new_target}, where $\bomega=\diag\{\omega_1,\ldots,\omega_F\}$, $\bQ_{\bB}$ is given in Eq.~\eqref{eq:QB}, and $\bLambda$ is defined in Eq.~\eqref{eq:Lambda_fac}. 
The threshold $\xi_{\bomega,\alpha}$ is the $\alpha$ quantile of $\sum_{f=1}^F \omega_f\eta_f^2$, where $\eta_1,\ldots,\eta_F$ are independent standard normal random variables.
\end{theorem}

The criterion in Theorem~\ref{thm:ReOMA} constrains a weighted Mahalanobis distance of the orthogonalized components $\sqrt{n}\tilde{\btau}_{\bx}$. 
Because it relies on the oracle knowledge of \(\bB\), we refer to it as the oracle rerandomization criterion for \(2^K\) factorial designs (ReO$_\F$).
When all factorial effects are assigned equal priority, $\omega_1=\cdots=\omega_F$, the criterion $\phi_{\bomega}$ reduces to
\begin{equation}\label{eq:ReOMAE}
\begin{split}
\tilde\phi_{\bomega}(\sqrt{n}\hat\btau_\bx,\bV_{\bx\bx},\bB) 
&= 1\left\{(\sqrt{n}\hat\btau_\bx)'\bB\big(\bV_{\btau\btau}^{||}\big)^{-1}\bB'(\sqrt{n}\hat\btau_\bx) \leq \xi_{F,\alpha}\right\},
\end{split}
\end{equation}
which constrains the Mahalanobis distance of the projected covariate imbalance
$\bB'(\sqrt{n}\hat{\btau}_{\bx})$ and therefore treats all factorial effects equally. 
We refer to this special case as ReO$_{\FE}$.

The next subsection investigates the theoretical properties of ReO$_\F$ and ReO$_\FE$, showing how the prespecified effect priorities are translated into effect-specific precision gains and comparing their performance with existing rerandomization methods.

\subsection{Asymptotic properties of the oracle criteria}

The asymptotic analysis is conducted under the finite-population framework, where the covariates and potential outcomes are fixed and randomness arises solely from treatment assignment.
Similar to ReFM,
the asymptotic distribution of $\hat{\btau}$ under $\phi_{\bomega}$ consists of a Gaussian component and a rerandomization-induced truncated Gaussian component.
The key difference is that the truncation is imposed on weighted orthogonalized components, reflecting the prespecified effect priorities.
When all weights are equal, ReO$_\FE$ reduces the dimension of the truncated part from $Fp$ to $F$.
Theorem~\ref{thm:AsymptoticDistri} formalizes these results.

\begin{theorem}\label{thm:AsymptoticDistri}
Under Assumptions~\ref{assump:regularity_fac} and~\ref{assump:V_tautau}, under ReO$_\F$, 
\begin{equation}\label{eq:DisReOMA}
\sqrt{n}(\hat{\btau}-\btau)
\overset{\cdot}{\sim}(\bV_{\btau\btau}^{\perp})^{1/2}\bvarepsilon +  \bQ_\bB^{-1}\bLambda^{1/2}\bm\eta \mid \bm\eta'\bomega\bm\eta \leq \xi_{\bomega,\alpha},
\end{equation}
where $\bLambda$ is defined in Eq.~\eqref{eq:Lambda_fac}, $\bomega=\diag\{\omega_1,\ldots,\omega_F\}$, $\bm\eta=(\eta_1,\ldots,\eta_F)'\sim\mathcal{N}(\bZero,\bI_F)$, 
$\bvarepsilon\sim\mathcal{N}(\bZero,\bI_F)$, 
with $\bm\eta$ and $\bvarepsilon$ independent,
and $\xi_{\bomega,\alpha}$ is the $\alpha$-quantile of $\sum_{f=1}^F \omega_f\eta_f^2$. 
Moreover, when all weights are equal,
$\omega_1=\cdots=\omega_F$ (ReO$_\FE$),
the distribution reduces to
\begin{equation}\label{eq:DisReOMAE}
\sqrt{n}(\hat{\btau}-\btau)\overset{\cdot}{\sim}
\big(\bV_{\btau\btau}^{\perp}\big)^{1/2}\bvarepsilon + \big(\bV_{\btau\btau}^{||}\big)^{1/2}\bzeta_{F,\alpha},
\end{equation}
where 
$\bzeta_{F,\alpha}\sim \bm\eta\mid \bm\eta'\bm\eta\leq \xi_{F,\alpha}$, and $\xi_{F,\alpha}$ is the $\alpha$-quantile of $\chi^2_F$ distribution.
\end{theorem}

Theorem~\ref{thm:AsymptoticDistri} implies that rerandomization preserves the consistency of $\hat{\btau}$. We next quantify the resulting precision gains.
To this end, we decompose the rerandomization-affected covariance component $\bV_{\btau\btau}^{||}$ according to the orthogonalized components $\tilde{\tau}_{\bx,1},\ldots,\tilde{\tau}_{\bx,f}$.
For $f=1,\ldots,F$, define the covariance component of $\sqrt{n}(\hat{\btau}-\btau)$ associated with $\sqrt{n}\tilde{\tau}_{\bx,f}$ as
\[\begin{split}
\bV_{\btau\btau}^{||}[f] &= 
\cov\left(\cov(\hat{\btau}-\btau,\tilde{\tau}_{\bx,f})(\bbV(\tilde{\tau}_{\bx,f}))^{-1}\sqrt{n}\tilde{\tau}_{\bx,f}\right)
= \bV_{\btau\btau}^{||}\bQ_\bB'\left(\frac{\be_f\be_f'}{\lambda_f}\right)\bQ_\bB\bV_{\btau\btau}^{||},
\end{split}\]
where $\be_f$ is the $f$th canonical basis vector in $\bbR^F$.
Summing over all orthogonalized components yields
$\sum_{f=1}^F \bV_{\btau\btau}^{||}[f]= \bV_{\btau\btau}^{||}\bQ_\bB'\bLambda^{-1}\bQ_\bB\bV_{\btau\btau}^{||}
= \bV_{\btau\btau}^{||}.$
In addition, for any $1\leq f,j\leq F$, define the squared correlation between $\hat{\tau}_f$ and $\tilde{\tau}_{\bx,j}$ as
\[\begin{split}
R^2_f[j] = 
\frac{(\cov(\hat{\tau}_f,\tilde{\tau}_{\bx,j}))^2}{\bbV(\hat{\tau}_f)\bbV(\tilde{\tau}_{\bx,j})}
= \frac{(\be_f'\bV_{\btau\btau}^{||}\bQ_\bB'\be_j)^2}{ V_{\tau_f\tau_f}\lambda_j}.
\end{split}\]
The following lemma characterizes the structure of $R_f^2[j]$.

\begin{lemma}\label{lem:R2}
For $f=1,\ldots,F-1$, 
$R^2_f[j]=0$ for all $j > f$.
Moreover,
$\sum_{j=1}^f R^2_f[j] = R^2_f$.
\end{lemma}

Lemma~\ref{lem:R2} shows that $\hat{\tau}_f$ is correlated only with the first $f$ orthogonalized components. 
Consequently, any precision gain for $\hat{\tau}_f$ must arise from these components.
This property is a direct consequence of the sequential orthogonalization: the components are constructed according to the prespecified ordering of the factorial effects, so that each later component is orthogonal to all preceding effect estimators.

To quantify the impact of the rerandomization constraint on each orthogonalized component, define
\begin{equation}\label{eq:cf}
c_f = \bbE\left[\eta_f^2 \mid \sum_{i=1}^F \omega_i \cdot \eta_i^2 \leq \xi_{\bomega,\alpha}\right],\ f=1,\ldots,F,
\end{equation}
where $\eta_1,\ldots,\eta_F$ are independent standard normal random variables. 
The constants $c_1,\ldots,c_F$ have the following properties.
\begin{lemma}\label{lem:cf}
For $\omega_1\geq \ldots \geq \omega_F \geq 0$,
we have
$c_1 \leq \ldots\leq c_F \leq 1$.
Moreover, for $\omega_1\geq \ldots \geq \omega_F > 0$, as $\alpha \rightarrow 0$, we have $c_f/v_{F,\alpha} \rightarrow
\omega_f^{-1} \left(\prod_{j=1}^F \omega_j\right)^{1/F}.$
\end{lemma}

We are now ready to characterize the resulting covariance reduction and the corresponding percentage reduction in asymptotic sampling variance (PRIASV).

\begin{theorem}\label{thm:PRIASV}
Suppose that Assumptions~\ref{assump:regularity_fac} and~\ref{assump:V_tautau} hold. 
Under ReO$_\F$, the reduction in the asymptotic covariance of $\sqrt{n}(\hat{\btau}-\btau)$ relative to CRFE is
$\lim_{n\to\infty}\sum_{f=1}^F (1-c_f) \cdot\bV_{\btau\btau}^{||}[f]$, and the PRIASV of $\hat{\tau}_f$ is
$\lim_{n\to\infty} 100 \times \sum_{j=1}^f (1-c_j)R^2_f[j]$, $f=1,\ldots,F$.
When all effect weights are equal,
$\omega_1=\ldots=\omega_F$ (ReO$_\FE$), then $c_1=\ldots=c_F=v_{F,\alpha}$,
the covariance reduction simplifies to
$\lim_{n\to\infty} (1-v_{F,\alpha})\bV_{\btau\btau}^{||}$,
and the PRIASV of each $\hat{\tau}_f$ is 
$\lim_{n\to\infty} 100 \times (1-v_{F,\alpha})R^2_f$.
\end{theorem}

Lemma~\ref{lem:cf} and Theorem~\ref{thm:PRIASV} jointly imply that ReO$_\F$ improves the asymptotic precision of $\hat{\btau}$ relative to CRFE, and the rerandomization constraint acts more strongly on components with higher priorities.
Moreover, the following corollary shows that each factorial effect is guaranteed a minimum PRIASV determined by its prescribed priority.

\begin{corollary}\label{cor:LowerBound}
Suppose that Assumptions~\ref{assump:regularity_fac} and~\ref{assump:V_tautau} hold. 
If $\omega_1\geq \cdots \geq \omega_F$, then, under ReO$_\F$, the PRIASV of $\hat{\tau}_f$ has the lower bound
$\lim_{n\to\infty} 100 \times (1-c_f)R_f^2.$
\end{corollary}

Corollary~\ref{cor:LowerBound} and  Lemma~\ref{lem:cf} show how prespecified effect priorities translate into effect-specific guarantees on precision improvement.
For small $\alpha$, the lower bound on the PRIASV of $\hat{\tau}_f$ is determined by $R_f^2$ and $c_f
\approx
v_{F,\alpha}
\omega_f^{-1}\left(\prod_{j=1}^F \omega_j\right)^{1/F}.$
Here, $R_f^2$ measures the variance in $\hat{\tau}_f$ explainable by covariates, while a larger $\omega_f$ yields a smaller $c_f$ and hence a larger lower bound.
When the projected components $\bB_1'\hat{\btau}_\bx,\ldots,\bB_F'\hat{\btau}_\bx$ are uncorrelated, this lower bound is attained exactly. 
Otherwise, correlation among the projected components may generate additional precision gains through higher-priority effects, so the actual PRIASV can exceed the lower bound.
Thus, ReO$_\F$ provides a direct mechanism for allocating precision gains according to the prespecified priorities.
By contrast, ReFMT$_\F$ and ReFMT$_\CF$ do not guarantee effect-specific lower bounds, as their gains depend on the complex correlation structure among effects and covariates.



We now consider the equal-priority case, where ReO$_\F$ reduces to ReO$_\FE$.
Recall that under ReFM \citep{li2020rerandomization}, the reduction in the asymptotic sampling covariance matrix of $\sqrt{n}(\hat{\btau}-\btau)$ relative to CRFE is $\lim_{n\rightarrow\infty}(1-v_{Fp,\alpha})\bV_{\btau\btau}^{||}$.
Since $v_{Fp,\alpha}\geq v_{F,\alpha}$, ReO$_\FE$ achieves a larger covariance reduction and correspondingly a larger PRIASV for every factorial effect.
Therefore, ReO$_\FE$ uniformly improves upon ReFM in terms of estimation precision.

The advantage of ReO$_\FE$ can also be characterized through peakedness, which measures the concentration of the sampling distribution around $\bZero$.
Intuitively, a more peaked distribution is more concentrated near the origin and thus corresponds to a more precise estimator.

\begin{definition}
For symmetric random vectors $\bm\phi,\bm\psi \in \mathbb{R}^m$,
we say that $\bm\phi$ is more peaked than $\bm\psi$, denoted by $\bm\phi\succ \bm\psi$, if $\bbP(\bm\phi \in \mathcal{K}) \geq \bbP(\bm\psi \in \mathcal{K})$ for any symmetric convex set $\mathcal{K}\subset\mathbb{R}^m$.
\end{definition}

As shown by \citet{li2020rerandomization}, $\bm\phi\succ\bm\psi$ implies $\cov(\bm\phi)\leq \cov(\bm\psi)$, so peakedness is a stronger measure of concentration than covariance reduction.
The following theorem shows that the asymptotic distribution of $\sqrt{n}(\hat{\btau}-\btau)$ under ReO$_\FE$ is more peaked than that under ReFM.

\begin{theorem}\label{thm:peakedness}
Suppose that Assumptions~\ref{assump:regularity_fac} and~\ref{assump:V_tautau} hold. 
Then the asymptotic sampling distribution of $\sqrt{n}(\hat{\btau}-\btau)$ under ReO$_\FE$ 
is more peaked than its counterpart under ReFM.
\end{theorem}

Finally, we compare with a tiered covariate scheme.
Recall that ReFMT$_\CF$ \citep{li2020rerandomization} incorporates tiers of both covariates and factorial effects. 
When all factorial effects are assigned equal priority, it reduces to a version in which only the covariates are tiered; we denote this special case by ReFMT$_\C$.
The following theorem shows that ReO$_{\FE}$ dominates ReFMT$_{\C}$ in terms of asymptotic covariance reduction.

\begin{theorem}\label{thm:comparison}
Suppose that Assumptions~\ref{assump:regularity_fac} and~\ref{assump:V_tautau} hold. 
Then the asymptotic sampling covariance matrix of $\sqrt{n}(\hat{\btau}-\btau)$ under ReO$_\FE$ is no larger than that under ReFMT$_\C$.
\end{theorem}

\section{Data-adaptive procedure }\label{sec:twostage}

The results in the previous section establish the theoretical advantages of the oracle criteria. 
When the covariate importance matrix \(\bB\) is known, ReO$_\F$ provides effect-specific guarantees on precision improvement, while its equal-priority special case ReO$_\FE$ uniformly improves upon Mahalanobis-distance-based criteria.
In practice, however, the importance matrix \(\bB\) depends on the unknown relationship between covariates and potential outcomes and is therefore unavailable, motivating a data-adaptive implementation of the proposed criteria.

Our approach is inspired by the two-stage rerandomization framework developed for treatment-control experiments \citep{liu2025bayesian}. 
Unlike that procedure, which relies on an auxiliary data set to estimate the covariate importance structure, we learn $\bB$ from a random subset of units in the experiment itself.
This avoids a strong practical requirement but introduces a new theoretical challenge: the rerandomization criterion is learned from one subset and then applied to the other, so the resulting estimators are no longer independent.
The following subsections develop the procedure and establish its theoretical properties.

\subsection{CRFE with sample splitting}

Before introducing the data-adaptive rerandomization procedure, we first consider a sample-splitting implementation of complete randomization for $2^K$ factorial designs. 

Specifically, the dataset is randomly split into two subsamples, and a completely randomized factorial experiment (CRFE) is conducted separately within each subsample. 
We refer to this procedure as CRFE with sample splitting (S-CRFE). 
Its theoretical analysis serves as a useful benchmark for the subsequent development of the data-adaptive rerandomization procedure. The procedure can be summarized as follows:
\begin{enumerate}
\item \textbf{Sample splitting.} Randomly split the finite population into two subsets of sizes $n^{(1)}$ and $n^{(2)}=n-n^{(1)}$, with proportions $\rho_n=n^{(1)}/n\in(0,1)$ and $1-\rho_n$, respectively. Let $\bS=(S_1,\ldots,S_n)'\in\{0,1\}^n$ denote the splitting indicator, where $S_i=1$ indicates assignment to the first subset. 
Thus, $\sum_{i=1}^n S_i=n^{(1)}$.
Given $\bS=\bms$, define $\mathcal{I}_1(\bms)=\{i:s_i=1\}$ and
$\mathcal{I}_2(\bms)=\{i:s_i=0\}$
as the index sets of the two subsets, ordered increasingly and with cardinalities $n^{(1)}$ and $n^{(2)}$.

\item \textbf{CRFE in the first subset.} Conduct a CRFE within $\mathcal{I}_1(\bms)$.
Let $\bZ^{(1)}=(Z_1^{(1)},\ldots,Z_{n^{(1)}}^{(1)})'\in \{1,\ldots,Q\}^{n^{(1)}}$ denote the treatment assignment vector, where $Z_i^{(1)}$ corresponds to the $i$th unit in  $\mathcal{I}_1(\bms)$. 
For $q=1,\ldots,Q$, denote the number of units assigned to treatment combination $q$ as $n_{1q}=\sum_{i=1}^{n^{(1)}}1\{Z_i^{(1)}=q\}=n^{(1)}r_q$.

\item \textbf{CRFE in the second subset.} Conduct a CRFE within $\mathcal{I}_2(\bms)$. 
Let $\bZ^{(2)}=(Z_1^{(2)},\ldots,Z_{n^{(2)}}^{(2)})'\in \{1,\ldots,Q\}^{n^{(2)}}$ denote the treatment assignment vector, where $Z_i^{(2)}$ corresponds to the $i$th unit in  $\mathcal{I}_2(\bms)$.  
For $q=1,\ldots,Q$, denote the number of units assigned to treatment combination $q$ as $n_{2q}=\sum_{i=1}^{n^{(2)}}1\{Z_i^{(2)}=q\}=n^{(2)}r_q$.
\end{enumerate}

Following the above procedure, the overall treatment assignment vector $\bZ$ is uniquely determined by
$(\bS,\bZ^{(1)},\bZ^{(2)})$ and is denoted by
$\bZ(\bS,\bZ^{(1)},\bZ^{(2)})$.
For any index set $\mathcal{S}$, let $\bZ_{\mathcal{S}}$ denote the subvector of $\bZ$ indexed by $\mathcal{S}$. 
Then $\bZ_{\cI_1(\bS)} = \bZ^{(1)}$, $\bZ_{\cI_2(\bS)} = \bZ^{(2)}$.
Define the extended assignment vectors
$\tilde{\bZ}^{(1)},\tilde{\bZ}^{(2)}\in\{0,1,\ldots,Q\}^n$ by
$\tilde{\bZ}^{(1)}_{\cI_1(\bS)} = \bZ^{(1)}$, $\tilde{\bZ}^{(1)}_{\cI_2(\bS)} = \bZero$; and 
$\tilde{\bZ}^{(2)}_{\cI_2(\bS)} = \bZ^{(2)}$, $\tilde{\bZ}^{(2)}_{\cI_1(\bS)} = \bZero$.
The resulting estimators of $\btau$ based on the two subsets are
\[\begin{split}
\hat{\btau}_1(\bS,\bZ^{(1)}) &= 
\sum_{q=1}^Q \bA_q \cdot n_{1q}^{-1}\sum_{i:\tilde{Z}_i^{(1)}=q} S_i Y_i(q),\\
\hat{\btau}_2(\bS,\bZ^{(2)}) &= 
\sum_{q=1}^Q \bA_q \cdot n_{2q}^{-1}\sum_{i:\tilde{Z}_i^{(2)}=q} (1-S_i) Y_i(q).
\end{split}\]
Then the estimator $\hat{\btau}$ in Eq.~\eqref{eq:estimator_fac} is a weighted average of the estimators from the two subsets:
\begin{equation}\label{eq:estimator_TSCRFE}
\hat{\btau} = \hat{\btau}(\bS,\bZ^{(1)},\bZ^{(2)}) = 
\rho_n \cdot \hat{\btau}_1(\bS,\bZ^{(1)}) + (1-\rho_n) \cdot \hat{\btau}_2(\bS,\bZ^{(2)}),
\end{equation}
where $\rho_n=n^{(1)}/n$. 
Similarly, define the corresponding covariate mean-difference vectors as
\[\begin{split}
\hat{\btau}_{\bx,1}(\bS,\bZ^{(1)}) &= \sum_{q=1}^Q \bA_q \otimes \left(n_{1q}^{-1}\sum_{i:\tilde{Z}_i^{(1)}=q} S_i \bX_i\right),
\\
\hat{\btau}_{\bx,2}(\bS,\bZ^{(2)}) &= \sum_{q=1}^Q \bA_q \otimes \left(n_{2q}^{-1}\sum_{i:\tilde{Z}_i^{(2)}=q} (1-S_i) \bX_i\right).
\end{split}\]

\subsection{Asymptotic properties of CRFE with sample splitting}

Under CRFE with sample splitting (S-CRFE), the treatment assignment mechanism is jointly determined by
$\bS$, $\bZ^{(1)}$, and $\bZ^{(2)}$.
For any $\bms\in\{0,1\}^n$, $\bz^{(1)}\in\{1,\ldots,Q\}^{n^{(1)}}$, and $\bz^{(2)}\in\{1,\ldots,Q\}^{n^{(2)}}$ satisfying
$\sum_{i=1}^n s_i = n^{(1)}$,
$\sum_{i=1}^{n^{(1)}} 1\{z_i^{(1)}=q\} = n_{1q}$, and 
$\sum_{i=1}^{n^{(2)}} 1\{z_i^{(2)}=q\} = n_{2q}$ for $q\in \{1,\ldots,Q\}$,  
\[\begin{split}
\bbP(\bS=\bms,\bZ^{(1)}= \bz^{(1)},\bZ^{(2)} = \bz^{(2)})
&= \binom{n}{n^{(1)}}^{-1}\frac{n_{11}!\ldots n_{1Q}!}{n^{(1)}!}
\frac{n_{21}!\ldots n_{2Q}!}{n^{(2)}!}\\
&= \frac{n_{11}!\ldots n_{1Q}!n_{21}!\ldots n_{2Q}!}{n!}.
\end{split}\]
Therefore, CRFE with sample splitting is equivalent to randomly partitioning the \(n\) units into \(2Q\) groups of sizes
\(n_{11},\ldots,n_{1Q},n_{21},\ldots,n_{2Q}\).
Applying Theorem~3 of \citet{li2017general} then yields Theorem~\ref{thm:DistTSMCRE}, which characterizes the expectations and covariance structures of the mean-difference statistics for the potential outcomes and covariates.

\begin{theorem}\label{thm:DistTSMCRE}
Under S-CRFE, the random vector
\[\begin{split}
\big(\sqrt{n^{(1)}}(\hat{\btau}_1(\bS,\bZ^{(1)})-\btau)',
\sqrt{n^{(1)}}\hat{\btau}_{\bx,1}'(\bS,\bZ^{(1)}),
\sqrt{n^{(2)}}(\hat{\btau}_2(\bS,\bZ^{(2)})-\btau)',
\sqrt{n^{(2)}}\hat{\btau}_{\bx,2}'(\bS,\bZ^{(2)})\big)'
\end{split}\]
has mean zero and covariance matrix
{\setlength{\arraycolsep}{3pt}
\[\begin{split}
&\bV_\TSCRFE  
=
\begin{pmatrix}
\bV_{\btau\btau} + (1-\rho_n)\bS_{\btau\btau} & \bV_{\btau\bx} & -\sqrt{\rho_n(1-\rho_n)} \bS_{\btau\btau} & \bZero_{F\times Fp}\\
\bV_{\bx\btau} & \bV_{\bx\bx} & \bZero_{Fp\times F} & \bZero_{Fp \times Fp}\\
-\sqrt{\rho_n(1-\rho_n)} \bS_{\btau\btau} & \bZero_{F\times Fp} & \bV_{\btau\btau} + \rho_n\bS_{\btau\btau} & \bV_{\btau\bx}\\
\bZero_{Fp\times F} & \bZero_{Fp\times Fp} & \bV_{\bx\btau} & \bV_{\bx\bx}
\end{pmatrix}.
\end{split}\]}
\end{theorem}

Under Assumption~\ref{assump:regularity_fac_full} below, the random vector in Theorem \ref{thm:DistTSMCRE} converges in distribution to a multivariate normal distribution, as formalized in Theorem~\ref{thm:TSCRFE_asymp}.

\begin{assumption}\label{assump:regularity_fac_full}
As $n\to\infty$:
(1) The first-subset proportion $\rho_n=n^{(1)}/n$ converges to a limit $\rho\in[0,1)$.
(2) For each $q\in\{1,\ldots,Q\}$, $r_q=n_q/n$ converges to a positive limit.
(3) The finite-population quantities $S_{qq}$, $\bS_{q,\bx}$, $\bS_{\btau\btau}$, and $\bS_{\bx\bx}$ converge to finite limits. Moreover, $\bS_{\bx\bx}$ and its limit $\bS_{\bx\bx,\infty}$ are nonsingular.
(4) For each $q\in\{1,\ldots,Q\}$, $\max_{1\leq i\leq n} |Y_i(q)-\bar{Y}(q)|^2/n^{(1)} \rightarrow 0$, and $\max_{1\leq i\leq n} ||\bX_i-\bar{\bX}||_2^2/n^{(1)} \rightarrow 0$.
\end{assumption}

\begin{theorem}\label{thm:TSCRFE_asymp}
Under S-CRFE, if Assumption~\ref{assump:regularity_fac_full} holds, then
$\big(\sqrt{n^{(1)}}(\hat{\btau}_1(\bS,\bZ^{(1)})-\btau)',
\sqrt{n^{(1)}}\hat{\btau}_{\bx,1}'(\bS,\bZ^{(1)}),
\sqrt{n^{(2)}}(\hat{\btau}_2(\bS,\bZ^{(2)})-\btau)',
\sqrt{n^{(2)}}\hat{\btau}_{\bx,2}'(\bS,\bZ^{(2)})\big)'
\overset{\cdot}{\sim}
\cN(\bZero, \bV_{\TSCRFE}).$
\end{theorem}

Moreover, under S-CRFE, for
$\bz\in\{1,\ldots,Q\}^n$ satisfying
$\sum_{i=1}^n 1\{z_i=q\} = n_q$, $q=1,\ldots,Q$,
\[\begin{split}
\bbP(\bZ = \bz) &= \sum_{\bZ(\bms,\bz^{(1)},\bz^{(2)}) = \bz} \bbP(\bS = \bms, \bZ^{(1)} = \bz^{(1)}, \bZ^{(2)} = \bz^{(2)})\\
&= \frac{n_{11}!\ldots n_{1Q}!n_{21}!\ldots n_{2Q}!}{n!} \times 
\prod_{q=1}^Q \binom{n_q}{n_{1q}}
= \frac{n_1!\ldots n_Q!}{n!}.
\end{split}\]
Hence, S-CRFE induces exactly the same assignment mechanism as CRFE.
It follows that
\[\begin{split}
\sqrt{n}(\hat{\btau}(\bS,\bZ^{(1)},\bZ^{(2)})-\btau) \overset{\cdot}{\sim} \cN(\bZero,\bV_{\btau\btau}).
\end{split}\]

\subsection{Data-adaptive ReO$_\F$ or  ReO$_\FE$ (DA-ReO$_\F$ or DA-ReO$_\FE$)}

Building on S-CRFE, we propose data-adaptive versions of ReO$_\F$ and  ReO$_\FE$.
The key idea is to use one subsample to estimate the covariate importance matrix \(\bB\), and then use the resulting estimator \(\hat{\bB}\) to construct the rerandomization criterion for the other subsample.
Consequently, the treatment assignments in the two subsamples are no longer independent.

Specifically, we first randomly split the sample into two subsets, $\mathcal{I}_1(\bS)$ and $\mathcal{I}_2(\bS)$. 
Within $\mathcal{I}_1(\bS)$, we implement ReFM to obtain the treatment assignment $\bW^{(1)}$. 
The resulting data are used to construct a difference-in-means estimator $\hat{\btau}_1$ of $\btau$ and an estimator $\hat{\bB}$ of $\bB$.
Treating $\hat{\bB}$ as a proxy for $\bB$, we implement ReO$_\F$ or ReO$_\FE$ within $\mathcal{I}_2(\bS)$ to obtain $\bW^{(2)}$, and construct the second difference-in-means estimator $\hat{\btau}_2$ of $\hat{\btau}$.
We refer to the resulting designs as data-adaptive rerandomization for $2^K$ factorial designs, denoted by DA-ReO$_\F$ when ReO$_\F$ is used in $\mathcal{I}_2(\bS)$ and by DA-ReO$_\FE$ when ReO$_\FE$ is used instead. 
The design is illustrated in Figure~\ref{fig:TSReO}.

\begin{figure}
\centering{
\includegraphics[width=0.8\linewidth,height=\textheight]{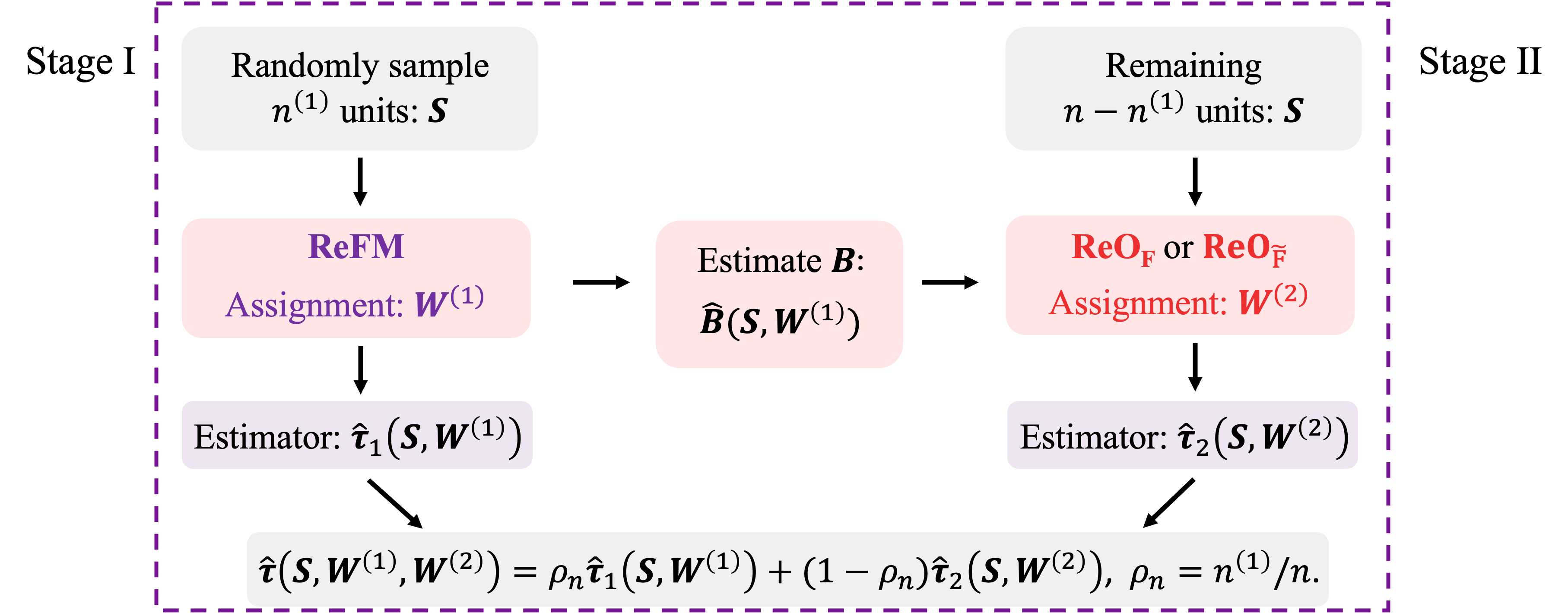}
}
\caption{\label{fig:TSReO}Data-adaptive rerandomization for $2^K$ factorial designs.}
\end{figure}

We now formally describe the procedure. 
Given the sample split \(\bS=\bms\), the assignment \(\bW^{(1)}\) on $\mathcal I_1(\bm s)$ is generated via ReFM, with conditional distribution
\begin{equation}\label{eq:assign1_fac}
(\bW^{(1)} \mid \bS = \bms) \sim \bZ^{(1)} \mid 
\phi_M\big(\sqrt{n^{(1)}}\hat{\btau}_{\bx,1}(\bms,\bZ^{(1)}),
\bV_{\bx\bx,1}(\bms)\big)=1,
\end{equation}
where $\bZ^{(1)}$ denotes the CRFE assignment vector on $\mathcal{I}_1(\bms)$, and $\phi_M$ is defined in Eq.~\eqref{eq:ReFM} with asymptotic acceptance probability $\alpha$. 
The covariance matrix 
$\bV_{\bx\bx,1}(\bms) 
=
\sum_{q=1}^Q r_{q}^{-1}(\bA_q\bA_q')\otimes \bms_{\bx\bx}(\cI_1(\bms)),$
where $\bms_{\bx\bx}(\mathcal{I}_1(\bms))= (n^{(1)}-1)^{-1}\sum_{i=1}^{n} s_i (\bX_i - \bar{\bX}_1(\bms))(\bX_i - \bar{\bX}_1(\bms))'$ is the sample covariance matrix of the covariates in $\cI_1(\bms)$, with $\bar{\bX}_1(\bms) = \frac{1}{n^{(1)}}\sum_{i=1}^{n} s_i \bX_i$.
If $\bms_{\bx\bx}(\mathcal{I}_1(\bms))$ is singular, we define
$\bV_{\bx\bx,1}^{-1}(\bms)=\bZero$, and the assignment mechanism coincides with CRFE.

Given $\bS=\bms$ and $\bW^{(1)}=\bw$, we construct a consistent estimator $\hat{\bB}(\bms,\bw)$ of $\bB$ using the observed outcomes in $\mathcal{I}_1(\bms)$.
Its explicit form is given later.

Conditional on \(\bS=\bms\) and \(\bW^{(1)}=\bw\), the assignment \(\bW^{(2)}\) is generated by applying ReO\(_\F\) or ReO\(_\FE\) with \(\hat{\bB}(\bms,\bw)\) in place of \(\bB\). 
Its conditional distribution is
\[\begin{split} 
(\bW^{(2)} \mid \bS = \bms, \bW^{(1)} = \bw) \sim \bZ^{(2)} \mid 
\phi_{\bomega}\big(\sqrt{n^{(2)}}\hat{\btau}_{\bx,2}(\bms, \bZ^{(2)}), \bV_{\bx\bx,2}(\bms), \hat{\bB}(\bms,\bw)\big) = 1,
\end{split}\] 
where $\bZ^{(2)}$ denotes the CRFE assignment vector on \(\mathcal I_2(\bS)\), and \(\phi_{\bW}\) is defined in Eq.~\eqref{eq:ReOMA} with asymptotic acceptance probability \(\alpha\). 
The covariance matrix
$\bV_{\bx\bx,2}(\bms) = \sum_{q=1}^Q r_q^{-1}(\bA_q\bA_q')\otimes \bms_{\bx\bx}(\cI_2(\bms))$,
where $\bms_{\bx\bx}(\mathcal{I}_2(\bms))=(n^{(2)}-1)^{-1}\sum_{i=1}^{n} (1-s_i) (\bX_i - \bar{\bX}_2(\bms))(\bX_i - \bar{\bX}_2(\bms))'$, with
$\bar{\bX}_2(\bms) = \frac{1}{n^{(2)}}\sum_{i=1}^{n} (1-s_i) \bX_i$.
If \((\hat{\bB}(\bms,\bw))'\bV_{\bx\bx,2}(\bms)\hat{\bB}(\bms,\bw)\) is singular, the acceptance probability is set to one, and the assignment mechanism coincides with CRFE.

Finally, the difference-in-means estimator of $\btau$ is 
\[\begin{split} 
\hat{\btau} = \hat{\btau}(\bS, \bW^{(1)}, \bW^{(2)}) = \rho_n \cdot \hat{\btau}_1(\bS, \bW^{(1)}) + (1-\rho_n) \cdot \hat{\btau}_2(\bS, \bW^{(2)}).
\end{split}\] 
where \(\hat{\btau}_1(\bS,\bW^{(1)})\) and
\(\hat{\btau}_2(\bS,\bW^{(2)})\) are estimators based on
\(\mathcal I_1(\bS)\) and \(\mathcal I_2(\bS)\), respectively:
$\hat{\btau}_1(\bS,\bW^{(1)}) = 
\sum_{q=1}^Q \bA_q \cdot n_{1q}^{-1}\sum_{i:\tilde{W}_i^{(1)}=q} S_i Y_i(q)$,
$\hat{\btau}_2(\bS,\bW^{(2)}) = 
\sum_{q=1}^Q \bA_q \cdot n_{2q}^{-1}\sum_{i:\tilde{W}_i^{(2)}=q} (1-S_i) Y_i(q)$.
Here, the extended assignment vectors $\tilde{\bW}^{(1)}$ and $\tilde{\bW}^{(2)}$ are defined from $(\bS,\bW^{(1)})$ and $(\bS,\bW^{(2)})$, respectively, analogously to $\tilde{\bZ}^{(1)}$ and $\tilde{\bZ}^{(2)}$.
In particular, $\tilde{\bW}^{(1)}$ agrees with $\bW^{(1)}$ on $\cI_1(\bS)$ and is zero elsewhere, while $\tilde{\bW}^{(2)}$ agrees with $\bW^{(2)}$ on $\cI_2(\bS)$ and is zero elsewhere.

\subsection{Construction of a consistent estimator of covariate importance}

We now propose a consistent estimator of $\bB$ based on the data from $\cI_1(\bS)$. 
Given $\bS=\bms$ and $\bW^{(1)}=\bw$, let $\tilde{\bW}^{(1)}=\tilde{\bw}$ denote the corresponding extended assignment vector. 
For each treatment combination $q=1,\ldots,Q$, define the sample means of the covariates and observed outcomes among units in $\cI_1(\bS)$ assigned to treatment combination $q$ as
\[\begin{split}
\Bar{\bX}_{1q}(\bms,\bw) = n_{1q}^{-1}\sum_{i=1}^{n} s_i1\{\tilde{w}_i=q\}\bX_{i},\ 
\Bar{Y}_{1q}(\bms,\bw) = n_{1q}^{-1}\sum_{i=1}^{n} s_i1\{\tilde{w}_i=q\}Y_{i}(q).
\end{split}\]
The corresponding sample variance and covariance matrices are
$s_{qq}(\bms,\bw) = (n_{1q}-1)^{-1}\sum_{i=1}^{n} s_i 1\{\tilde{w}_i=q\}(Y_{i}(q)-\Bar{Y}_{1q}(\bms,\bw))^2$,
$\bms_{\bx,q}(\bms,\bw) = (n_{1q}-1)^{-1}\sum_{i=1}^{n} s_i1\{\tilde{w}_i=q\}(\bX_{i}-\Bar{\bX}_{1q}(\bms,\bw))(Y_{i}(q)-\Bar{Y}_{1q}(\bms,\bw))$,
and 
$\bms_{\bx\bx,q}(\bms,\bw) = (n_{1q}-1)^{-1}\sum_{i=1}^{n} s_i1\{\tilde{w}_i=q\}(\bX_{i}-\Bar{\bX}_{1q}(\bms,\bw))(\bX_{i}-\Bar{\bX}_{1q}(\bms,\bw))'$.
The finite-population projection coefficient of \(Y_i(q)\) on \(\bX_i\) is
$\bbeta_q = \bS_{\bx\bx}^{-1}\bS_{\bx,q}$, $q =1,\ldots,Q$.
By definition, 
$\bB$ can be written as
\[\begin{split}
\bB = \bV_{\bx\bx}^{-1}\bV_{\bx\btau} = \sum_{q=1}^Q\left[\left(\sum_{g=1}^Q \frac{\bA_{g}\bA_g'}{r_g}\right)^{-1}\frac{\bA_q\bA_q'}{r_q}
\right]\otimes \bbeta_q.
\end{split}\]
Motivated by this representation, given \(\bS=\bms\) and \(\bW^{(1)}=\bw\), we estimate \(\bbeta_q\) by $\hat{\bbeta}_q(\bms,\bw) = \bms_{\bx\bx,q}^{-1}(\bms,\bw)\bms_{\bx,q}(\bms,\bw)$, 
$q=1\ldots,Q$.
If $\bms_{\bx\bx,q}(\bms,\bw)$ is singular, we define
$\hat{\bbeta}_q(\bms,\bw)=\bZero$.
We then estimate \(\bB\) by
\begin{equation}\label{eq:prior_fac}
\hat{\bB}(\bms,\bw) = 
\sum_{q=1}^Q\left[\left(\sum_{g=1}^Q \frac{\bA_{g}\bA_g'}{r_g}\right)^{-1}\frac{\bA_q\bA_q'}{r_q}
\right]\otimes \hat{\bbeta}_q(\bms,\bw).
\end{equation}
The following proposition shows that \(\hat{\bB}(\bS,\bW^{(1)})\) consistently estimates \(\bB\).


\begin{proposition}\label{pro:consistent_fac}
Under DA-ReO$_\F$ or DA-ReO$_\FE$, if Assumption~\ref{assump:regularity_fac_full} holds, then 
$\hat{\bB}(\bS,\bW^{(1)})-\bB=o_p(1)$.
\end{proposition}

\subsection{Asymptotic properties of the data-adaptive procedure}

We next study the asymptotic properties of the factorial effect estimator under DA-ReO$_\F$ and DA-ReO$_\FE$ with $\bB$ estimated by Eq.~\eqref{eq:prior_fac}.
We first derive the joint asymptotic distribution of $\hat{\btau}_1(\bS,\bW^{(1)})$ and $\hat{\btau}_2(\bS,\bW^{(2)})$, and then obtain the asymptotic distribution of their weighted combination $\hat{\btau}(\bS,\bW^{(1)},\bW^{(2)})$. 
The result is stated in Theorem~\ref{thm:distribution_fac}.

\begin{theorem}\label{thm:distribution_fac}
Suppose that Assumptions~\ref{assump:V_tautau} and~\ref{assump:regularity_fac_full} hold. 
As $n\to\infty$, under DA-ReO$_\F$ with effect weights
$\bomega=\diag\{\omega_1,\ldots,\omega_F\}$, for 
$\hat{\btau}_1=\hat{\btau}_1(\bS,\bW^{(1)})$ and $\hat{\btau}_2 = \hat{\btau}_2(\bS,\bW^{(2)})$,
we have
\[\begin{split}
&(\sqrt{n^{(1)}}(\hat{\btau}_1-\btau),
\sqrt{n^{(2)}}(\hat{\btau}_2-\btau))
\overset{\cdot}{\sim}
((\bV_{\btau\btau}^{||})^{1/2}_{Fp}\cdot \bzeta_{Fp,\alpha}
+\bepsilon_1,\bQ_\bB^{-1}\bLambda^{1/2}\bm\eta \mid \bm\eta'\bomega\bm\eta \leq \xi_{\bomega,\alpha}
+ \bepsilon_2),
\end{split}\]
where 
$\bzeta_{Fp,\alpha} \sim \tilde{\bEta}|\tilde{\bEta}'\tilde{\bEta} \leq \xi_{Fp,\alpha}$, with 
$\tilde{\bm\eta}\sim\mathcal{N}(\bZero,\bI_{Fp})$,
$\bm\eta\sim\mathcal{N}(\bZero,\bI_F)$,
and $\tilde{\bEta}$ and $\bm\eta$ are independent.
Moreover, 
\[\begin{split}
\begin{pmatrix}
\bepsilon_1\\
\bepsilon_2
\end{pmatrix}\sim\cN\left(\begin{pmatrix}
\bZero \\ \bZero
\end{pmatrix},\begin{pmatrix}
\bV_{\btau\btau}^{\perp}+(1-\rho_n)\bS_{\btau\btau} & 
-\sqrt{\rho_n(1-\rho_n)}\cdot \bS_{\btau\btau}\\
-\sqrt{\rho_n(1-\rho_n)}\cdot \bS_{\btau\btau} & 
\bV_{\btau\btau}^{\perp}+\rho_n \bS_{\btau\btau}
\end{pmatrix}\right)
\end{split}\]
with $(\bepsilon_1',\bepsilon_2')'$ independent of $\tilde{\bm\eta}$ and $\bm\eta$.
Under DA-ReO$_\FE$ (i.e., $\omega_1=\cdots=\omega_F$),
\[\begin{split}
&(\sqrt{n^{(1)}}(\hat{\btau}_1-\btau),
\sqrt{n^{(2)}}(\hat{\btau}_2-\btau))
\overset{\cdot}{\sim}
((\bV_{\btau\btau}^{||})^{1/2}_{Fp}\cdot \bzeta_{Fp,\alpha}
+\bepsilon_1,(\bV_{\btau\btau}^{||})^{1/2}\cdot \bm\zeta_{F,\alpha}
+ \bepsilon_2),
\end{split}\]
where $\bm\zeta_{F,\alpha} \sim \bEta \mid \bEta'\bEta \leq \xi_{F,\alpha}$.
\end{theorem}

Theorem~\ref{thm:distribution_fac} shows that 
$\hat{\btau}_1(\bS,\bW^{(1)})$ and
$\hat{\btau}_2(\bS,\bW^{(2)})$ are not asymptotically independent. 
Their asymptotic correlation is negative and depends on both the splitting proportion \(\rho_n\) and the finite-population covariance matrix of the individual factorial effects, \(\bS_{\btau\btau}\).
The correlation vanishes when \(\rho_n\in\{0,1\}\), corresponding to a degenerate split, or when the factorial effects are additive, that is, \(\bS_{\btau\btau}=\bZero\).
Moreover, since the correlated components are jointly Gaussian, the asymptotic distribution of the combined estimator
$\hat{\btau} = \hat{\btau}(\bS,\bW^{(1)},\bW^{(2)})$ can be derived explicitly, as stated in the following corollary.

\begin{corollary}\label{cor:distribution2_fac}
Suppose that Assumptions~\ref{assump:V_tautau} and~\ref{assump:regularity_fac_full} hold. 
Under DA-ReO$_\F$,
\[\begin{split}
\sqrt{n}(\hat{\btau} - \btau) \overset{\cdot}{\sim}\sqrt{\rho_n}\big(\bV_{\btau\btau}^{||}\big)^{1/2}_{Fp}\cdot \bm\zeta_{Fp,\alpha} + \sqrt{1-\rho_n}(\bQ_\bB^{-1}\bLambda^{1/2}\bm\eta \mid \bm\eta'\bomega\bm\eta \leq \xi_{\bomega,\alpha}) + \big(\bV_{\btau\btau}^{\perp}\big)^{1/2}\cdot \bvarepsilon,
\end{split}\]
where $\bvarepsilon\sim\mathcal{N}(\bZero,\bI_F)$, $\bzeta_{Fp,\alpha} \sim \tilde{\bEta}|\tilde{\bEta}'\tilde{\bEta} \leq \xi_{Fp,\alpha}$,  $\tilde\bEta \sim \cN(\bZero,\bI_{Fp})$, and $\bm\eta\sim\mathcal{N}(\bZero,\bI_F)$. 
Moreover,
$\tilde{\bm\eta}$, $\bm\eta$, and $\bvarepsilon$ are mutually independent.
Under DA-ReO$_\FE$,
\[\begin{split}
\sqrt{n}(\hat{\btau} - \btau) \overset{\cdot}{\sim} \sqrt{\rho_n}(\bV_{\btau\btau}^{||})_{Fp}^{1/2}\cdot \bm\zeta_{Fp,\alpha} +
\sqrt{1-\rho_n}(\bV_{\btau\btau}^{||})^{1/2}\cdot \bm\zeta_{p,\alpha} +
(\bV_{\btau\btau}^{\perp})^{1/2}\cdot \bvarepsilon,
\end{split}\]
where $\bzeta_{F,\alpha}
\sim
\bm\eta
\mid
\bm\eta'\bm\eta\leq \xi_{F,\alpha}$.
\end{corollary}

Comparing Corollary~\ref{cor:distribution2_fac} with Eq.'s~\eqref{eq:DisReFM}, \eqref{eq:DisReOMA}, and \eqref{eq:DisReOMAE}, 
the asymptotic distribution under DA-ReO$_\F$ (or DA-ReO$_\FE$) can be viewed as a weighted combination of those under ReFM and ReO$_\F$ (or ReO$_\FE$), with weights determined by $\rho_n$.
If $\rho_n\to\rho=0$, the asymptotic distribution under DA-ReO$_\F$ (or DA-ReO$_\FE$) reduces to that under ReO$_\F$ (or ReO$_\FE$).

Based on the asymptotic distribution, the following result characterizes the reduction in asymptotic covariance and the corresponding PRIASV for each factorial effect estimator.

\begin{corollary}\label{cor:PRIASV2_fac}
Suppose that Assumptions~\ref{assump:V_tautau} and~\ref{assump:regularity_fac_full} hold. 
Under DA-ReO$_\F$,
the reduction in the asymptotic covariance matrix of \(\sqrt n(\hat{\btau}-\btau)\) is
$\lim_{n\to\infty} \sum_{f=1}^F \left[1-\rho_n v_{Fp,\alpha} - (1-\rho_n)c_f\right]\cdot\bV_{\btau\btau}^{||}[f]$, whereas under DA-ReO$_\FE$ it is
$\lim_{n\to\infty} \left[1-\rho_n v_{Fp,\alpha} - (1-\rho_n)v_{F,\alpha}\right]\cdot\bV_{\btau\btau}^{||}$.
For $f=1,\ldots,F$, the PRIASV of $\hat{\tau}_f$ under DA-ReO$_\F$ is
$\lim_{n\to\infty} 100 \times \sum_{j=1}^f (1-\rho_n v_{Fp,\alpha} - (1-\rho_n) c_j)R^2_f[j]$, 
whereas under DA-ReO$_\FE$ it is $\lim_{n\to\infty} 100 \times (1-\rho_n v_{Fp,\alpha} - (1-\rho_n) v_{F,\alpha})R^2_f$.
Furthermore, if $\omega_1\geq\cdots\geq \omega_F$, then under DA-ReO$_\F$, the PRIASV of $\hat{\tau}_f$ is at least $\lim_{n\to\infty} 100 \times (1-\rho_n v_{Fp,\alpha} - (1-\rho_n) c_f)R^2_f$.
\end{corollary}

Corollaries~\ref{cor:distribution2_fac} and~\ref{cor:PRIASV2_fac} yield several important implications.
First, when all factorial effects are equally important, DA-ReO$_\FE$ asymptotically dominates ReFM. 
Specifically, the asymptotic covariance matrix of \(\hat{\btau}\) under DA-ReO$_\FE$ is no larger than that under ReFM, and the PRIASV of each \(\hat{\tau}_f\) is no smaller.
Second, under DA-ReO$_\F$ (with DA-ReO$_\FE$ as a special case), the precision gain of each \(\hat{\tau}_f\), measured by PRIASV, admits a lower bound determined by \(\rho_n\) and \(c_f\). 
This lower bound increases as \(\rho_n\) decreases. 
In particular, if \(\rho_n\to0\),  DA-ReO$_\F$ and DA-ReO$_\FE$ reduce to ReO$_\F$ and ReO$_\FE$, respectively. 
Thus, the proposed procedures asymptotically recover the corresponding oracle criteria.
Moreover, as in ReO$_\F$, larger effect weights yield larger lower bounds on the precision gains of the corresponding factorial effect estimators. 
Consequently, precision gains can be allocated across factorial effects according to their relative importance through the choice of \(\omega_1,\ldots,\omega_F\).

\subsection{Conservative covariance estimator and confidence set}\label{sec:CI_fac}
By Corollaries~\ref{cor:distribution2_fac} and~\ref{cor:PRIASV2_fac}, 
the asymptotic covariance matrix of $\sqrt{n}(\hat{\btau}-\btau)$ under  DA-ReO$_\F$ or DA-ReO$_\FE$ depends on 
$\bV_{\btau\btau}^{\perp}$, $\bV_{\btau\btau}^{\parallel}[f]$  ($f=1,\ldots,F$), and $\bV_{\btau\btau}^{\parallel}$.
Let $\bW\in\{1,\ldots,Q\}^n$ denote the overall treatment assignment vector, defined by
$\bW_{\cI_1(\bS)} = \bW^{(1)}$ and $\bW_{\cI_2(\bS)} = \bW^{(2)}$.
For each $q\in\{1,\ldots,Q\}$, define \(s_{qq}\), \(\bms_{q,\bx}\), and \(\bms_{\bx\bx,q}\) as the sample variance of the observed outcomes, the sample covariance between the observed outcomes and covariates, and the sample covariance matrix of the covariates among units assigned to treatment combination \(q\) under \(\bW\):
$s_{qq} = (n_q-1)^{-1}\sum_{i=1}^n 1\{W_i=q\} (Y_i-\hat{\bar{Y}}_q)^2$,
$s_{q,\bx} = (n_q-1)^{-1}\sum_{i=1}^n 1\{W_i=q\} (Y_i-\hat{\bar{Y}}_q)(\bX_i - \hat{\bar{\bX}}_q)$,
and 
$\bms_{\bx\bx,q} = (n_q-1)^{-1}\sum_{i=1}^n 1\{W_i=q\} (\bX_i - \hat{\bar{\bX}}_q)(\bX_i - \hat{\bar{\bX}}_q)'$,
where $\hat{\bar{Y}}_q = n_q^{-1}\sum_{i=1}^n 1\{W_i=q\} Y_i(q)$, and $\hat{\bar{\bX}}_q = n_q^{-1}\sum_{i=1}^n 1\{W_i=q\} \bX_i$.
The following lemma shows that these sample quantities consistently estimate their finite-population counterparts.

\begin{lemma}\label{lem:consistent_fac}
Under DA-ReO$_\F$ (including DA-ReO$_\FE$), if Assumptions~\ref{assump:V_tautau} and~\ref{assump:regularity_fac_full} hold, then, for each $q\in\{1,\ldots,Q\}$,
$\bms_{q,\bx} - \bS_{q,\bx} = o_p(1)$, $s_{qq} - S_{qq} = o_p(1)$, and $\bms_{\bx\bx,q} - \bS_{\bx\bx,q} = o_p(1)$.
\end{lemma}

Lemma~\ref{lem:consistent_fac} yields the following asymptotically conservative estimator of $\bV_{\btau\btau}^{\perp}=\bV_{\btau\btau}-\bV_{\btau\btau}^{||}$:
$$\hat\bV_{\btau\btau}^{\perp} = \sum_{q=1}^Q r_q^{-1}\bA_q \bA_q' (s_{qq} - \bms_{q,\bx} \bms_{\bx\bx,q}^{-1} \bms_{\bx,q}),$$
where $\bms_{\bx\bx,q}^{-1}=\bZero$ if $\bms_{\bx\bx,q}$ is singular.
For the explained component, define $\hat{\bV}_{\btau\bx} = \sum_{q=1}^Q \frac{1}{r_q}(\bA_q \bA_q')\otimes \bms_{q,\bx}$, $\hat\bB = \bV_{\bx\bx}^{-1}\hat{\bV}_{\bx\btau}$, and $\hat{\bV}_{\btau\btau}^{||} = \hat\bB'\bV_{\bx\bx}\hat\bB$.
Let $\bQ_{\hat\bB}$ denote the estimated orthogonal matrix obtained from
$\hat{\bV}_{\btau\btau}^{||}$, and define $\hat{\bLambda} = \diag\{\hat{\lambda}_1,\ldots,\hat{\lambda}_F\} = \bQ_{\hat\bB}\hat{\bV}_{\btau\btau}^{||}\bQ_{\hat\bB}'$.
Then the covariance component corresponding to the $f$th factorial effect can
be estimated by
\begin{equation}\label{eq:V_tautau}
\hat\bV_{\btau\btau}^{||}[f]
= \hat{\bV}_{\btau\btau}^{||}\bQ_{\hat\bB}'\left(\frac{\be_f\be_f'}{\hat\lambda_f}\right)\bQ_{\hat\bB}\hat{\bV}_{\btau\btau}^{||},\ f=1,\ldots,F,
\end{equation}
where $\be_f$ is the $f$th canonical basis vector in $\mathbb{R}^F$.
The above quantities are well defined when $\hat\bB$ has full column rank.

For the estimand $\sqrt{n}\bC\btau$, where $\bC$ has full row rank, the asymptotic covariance matrix under DA-ReO$_\F$ is estimated by
\begin{equation}\label{eq:variance1}
\bC\hat\bV_{\btau\btau}^{\perp}\bC' + 
\sum_{f=1}^F \left(\rho_nv_{Fp,\alpha} + (1-\rho_n)c_f\right)\bC\hat\bV_{\btau\btau}^{||}[f]\bC'.
\end{equation}
If $\hat{\bB}$ does not have full column rank, we replace the second term by
$(\rho_nv_{Fp,\alpha} + (1-\rho_n)\max\{c_1,\ldots,c_F\})\bC\hat\bV_{\btau\btau}^{||}\bC'$.
Under DA-ReO$_\FE$, the corresponding estimator is
\begin{equation}\label{eq:variance2}
\bC\hat\bV_{\btau\btau}^{\perp}\bC' + \left(\rho_nv_{Fp,\alpha} + (1-\rho_n)v_{F,\alpha}\right)
\bC\hat\bV_{\btau\btau}^{||}\bC'.
\end{equation}
For small acceptance probability $\alpha$, the distribution of
$\sqrt{n}\bC(\hat{\btau}-\btau)$ can be approximated by the normal
distribution with covariance $\bC\bV_{\btau\btau}^{\perp}\bC'$.
Hence, we construct the confidence set
\begin{equation}\label{eq:CI2}
\sqrt{n}\bC\hat{\btau}
+
\mathcal{O}(\bC\hat\bV_{\btau\btau}^{\perp}\bC',\hat c_{1-\delta}),
\end{equation}
where $\mathcal{O}(\bM, c) = \{\bmu: \bmu'\bM^{-1} \bmu \leq c\}$.
The critical value $\hat c_{1-\delta}$ is obtained from the estimated
asymptotic distribution of $\sqrt{n}(\hat{\btau}-\btau)$.
Specifically, it is the $(1-\delta)$ quantile of
$(\bC\bm\phi)'
(\bC\hat\bV_{\btau\btau}^{\perp}\bC')^{-1}
(\bC\bm\phi)$,
where $\bm\phi$ follows the estimated asymptotic distribution.
When $\hat{\bB}$ is rank deficient, 
$\big(\hat\bV_{\btau\btau}^{||}\big)^{1/2}\bm\eta$
is used in place of
$(\bQ_{\hat\bB}^{-1}\hat\bLambda^{1/2}\bm\eta
\mid
\bm\eta'\bomega\bm\eta\leq\xi_{\bomega,\alpha})$.
Let $\bS_{\btau\backslash\bx}  = \bS_{\btau\btau} - \bS_{\btau,\bx} \bS_{\bx\bx}^{-1} \bS_{\bx,\btau}$.
The following theorem establishes the validity of the proposed covariance
estimators and confidence sets.

\begin{theorem}\label{thm:CI2}
Suppose that Assumptions~\ref{assump:V_tautau} and~\ref{assump:regularity_fac_full} hold.
For $\bC$ with full row rank, the covariance estimators in
Eq.~\eqref{eq:variance1} and Eq.~\eqref{eq:variance2} are asymptotically
conservative for $\sqrt{n}\bC\btau$, in the sense that their probability
limits dominate the true asymptotic covariance matrices in the Loewner order.
The confidence set in Eq.~\eqref{eq:CI2} has asymptotic coverage probability
at least $1-\delta$.
Moreover, if $\bS_{\btau\backslash\bx}\to\bZero$, the confidence set is
asymptotically exact.
\end{theorem}

\section{Simulations}\label{sec:simu}
This section evaluates the finite-sample performance of the proposed data-adaptive rerandomization procedures, DA-ReO$_\F$ and DA-ReO$_\FE$, using model-generated and real datasets.

\subsection{Simulation studies with model-generated data}

We first evaluate the finite-sample performance of rerandomization methods using model-generated datasets. 
We consider a $2^K$ factorial experiment with $K=3$, yielding $Q=8$ treatment combinations and $F=7$ factorial effects. 
The covariates are generated as
$\bX_i\sim\mathcal{N}(\bZero_p,\bI_p)$, $p=5$.
To allow for nonlinear effects, we additionally include all pairwise interactions of the covariates, denoted by $\tilde{\bX}_i$, with dimension $\tilde p=p(p-1)/2$.

For each treatment combination $q\in\{1,\ldots,Q\}$, the potential outcomes are generated from
\[
Y_i(q)=a_q+\bm b_q'\bX_i+\tilde{\bm b}_q'\tilde{\bX}_i,\ i=1,\ldots,n,
\]
where $a_q \sim \cN(0,\sigma_a^2)$, $\bm b_q \sim \cN(2\cdot\bOne_p,\bI_p)$, and $\tilde{\bm b}_q \sim \cN(\bZero_{\tilde{p}}, \sigma_b^2 \bI_{\tilde{p}})$.
We set $n=800$ and $r_q=1/8$ for all $q$. 
Two outcome models are considered: a linear model with $\sigma_b^2=0$ and a nonlinear model with $\sigma_b^2=1$. 
The corresponding values of $\sigma_a^2$ are set to $20$ and $8$, respectively, so that the factorial effects have approximately the same level of explained variation, with $R_f^2\approx0.6$. 
Here, $R_f^2$ is computed using only the linear covariates $\bX$, which are also the covariates used in the rerandomization criteria.


We assume that the priority structure of factorial effects is known, whereas the covariate importance structure is unknown and may differ across effects. 
Thus, ReFMT$_\C$ and ReFMT$_\CF$ are excluded because they require a prespecified covariate importance structure. 
We consider two priority settings. 
Under equal priority, we compare DA-ReO$_\FE$ with CRFE and ReFM, using the infeasible ReO$_\FE$ as an oracle benchmark. 
Under unequal priority, where main effects receive higher priority than interactions, we compare DA-ReO$_\F$ with ReFMT$_\F$, using ReO$_\F$ as the oracle benchmark.
For DA-ReO$_\F$, we set
$\bW=\diag\{5,5,5,1,1,1,1\}$ with an asymptotic acceptance probability of $0.05$. 
For ReFMT$_\F$, we divide factorial effects into main and interaction effects and impose separate Mahalanobis-distance-based criteria with asymptotic acceptance probabilities $0.1$ and $0.5$, respectively.
For both data-adaptive procedures, we consider learning-sample proportions $\rho_n\in\{0.2,0.3,0.4\}$.

For each simulation setting, we generate 100 independent datasets, with 1,000 accepted assignments under each rerandomization criterion per dataset. 
We report the averages of the percentage reduction in variance (PRIV), empirical coverage probability (CP) of the 95\% confidence intervals, and confidence interval length (Len) for each factorial effect estimator.  
For ReFM and ReFMT, confidence intervals are constructed using the method of \citet{li2020rerandomization}; for ReO$_\FE$ and ReO$_\F$, they are constructed analogously to those proposed for DA-ReO$_\FE$ and DA-ReO$_\F$ in Section~\ref{sec:CI_fac}.

Table~\ref{tab:simulation} reports the average PRIV across simulation settings, with the largest value among feasible competitors highlighted in bold. 
Overall, the proposed data-adaptive procedures outperform the Mahalanobis-distance-based methods. 
Under equal priority, DA-ReO$_\FE$ achieves larger variance reductions than ReFM, whereas under unequal priority, DA-ReO$_\F$ outperforms ReFMT$_\F$ for both main effects and interactions. 
For both procedures, the largest PRIV is generally attained when $\rho_n$ is around $0.3$. 
As expected, the infeasible oracle criteria ReO$_\F$ and ReO$_\FE$ achieve the best overall performance.
Table B1 in the Supplementary Materials reports the average empirical coverage probabilities of the resulting confidence intervals, all of which attain the nominal 95\% level. 
Table B2 reports the average percentage reduction in confidence interval length relative to ReFM or ReFMT$_\F$. 
The proposed methods produce shorter confidence intervals, with greater reductions for smaller values of $\rho_n$.

{\setlength{\tabcolsep}{4pt}
\spacingset{1}
\small
\begin{longtable}[]{@{}ccccccccc@{}}
\caption{Average PRIV achieved by different rerandomization procedures.}\label{tab:simulation}\tabularnewline
\toprule\noalign{}
& & \multicolumn{3}{c}{Main}& \multicolumn{4}{c}{Interaction} \\
\cmidrule(lr){3-5}
\cmidrule(lr){6-9}
Setting & Scheme & 1 & 2 & 3 & 12 & 13 & 23 & 123 \\
\midrule\noalign{}
\endfirsthead
\toprule\noalign{}
& & \multicolumn{3}{c}{Main}& \multicolumn{4}{c}{Interaction} \\
\cmidrule(lr){3-5}
\cmidrule(lr){6-9}
Setting & Scheme & 1 & 2 & 3 & 12 & 13 & 23 & 123 \\
\midrule\noalign{}
\endhead
\bottomrule\noalign{}
\endlastfoot
\multirow{5}{*}{\shortstack{Linear\\
\hspace{\fill}\\
(same\\
\hspace{\fill}\\
importance)}} & ReO$_\FE$ & 46.84 & 46.84 & 47.18 & 46.84 & 46.08 & 47.10 & 47.06 \\ 
& DA-ReO$_\FE$ ($\rho_n=0.2$) & 35.94 & 35.91 & 36.06 & 36.37 & 35.12 & 36.41 & 36.01 \\ 
& DA-ReO$_\FE$ ($\rho_n=0.3$) & \bf 37.25 & \bf 36.90 & \bf 36.97 & \bf 36.68 & \bf 36.58 & \bf 37.56 & \bf 36.85 \\ 
& DA-ReO$_\FE$ ($\rho_n=0.4$) & 36.81 & 35.69 & 36.37 & 36.35 & 34.77 & 36.14 & 36.11 \\ 
& ReFM & 26.28 & 26.34 & 25.85 & 25.77 & 24.68 & 26.05 & 25.60 \\ 
\hdashline
\multirow{5}{*}{\shortstack{Linear\\
\hspace{\fill}\\
(varying\\
\hspace{\fill}\\
importance)}} & ReO$_\F$ & 53.74 & 54.20 & 53.65 & 34.18 & 34.63 & 35.15 & 37.36 \\ 
& DA-ReO$_\F$ ($\rho_n=0.2$) & 41.02 & 40.72 & 40.76 & 28.83 & 27.87 & 28.87 & 29.42 \\ 
& DA-ReO$_\F$ ($\rho_n=0.3$) & \bf 41.25 & \bf 40.90 & \bf 41.17 & \bf 29.51 & \bf 29.85 & 30.23 & 30.87 \\ 
& DA-ReO$_\F$ ($\rho_n=0.4$) & 40.11 & 39.16 & 39.82 & 29.14 & 28.37 & \bf 30.43 & \bf 30.97 \\ 
& ReFMT$_\F$ & 31.41 & 30.70 & 31.07 & 15.80 & 14.78 & 15.64 & 16.73 \\ 
\hdashline
\multirow{5}{*}{\shortstack{Nonlinear\\
\hspace{\fill}\\
(same\\
\hspace{\fill}\\
importance)}} & ReO$_\FE$ & 45.45 & 45.34 & 46.20 & 45.48 & 46.80 & 45.93 & 45.72 \\ 
& DA-ReO$_\FE$ ($\rho_n=0.2$) & 34.45 & 34.07 & 34.90 & 34.40 & 35.15 & 34.07 & 35.42 \\ 
& DA-ReO$_\FE$ ($\rho_n=0.3$) & \bf 35.08 & \bf 35.57 & \bf 35.43 & \bf 35.28 & \bf 36.46 & \bf 35.11 & \bf 35.99 \\ 
& DA-ReO$_\FE$ ($\rho_n=0.4$) & 34.73 & 34.31 & 35.02 & 34.42 & 35.29 & 34.21 & 35.45 \\ 
& ReFM & 25.14 & 24.11 & 25.54 & 24.76 & 26.14 & 24.60 & 25.19 \\ 
\hdashline
\multirow{5}{*}{\shortstack{Nonlinear\\
\hspace{\fill}\\
(varying\\
\hspace{\fill}\\
importance)}} & ReO$_\F$ & 50.92 & 52.59 & 53.61 & 36.11 & 35.61 & 34.00 & 35.93 \\ 
& DA-ReO$_\F$ ($\rho_n=0.2$) & 38.13 & 38.48 & 39.04 & 28.10 & 28.16 & 27.22 & 27.94 \\ 
& DA-ReO$_\F$ ($\rho_n=0.3$) & \bf 39.07 & \bf 39.42 & \bf 39.67 & 29.08 & \bf 30.36 & \bf 28.24 & \bf 29.09 \\ 
& DA-ReO$_\F$ ($\rho_n=0.4$) & 38.52 & 37.98 & 38.53 & \bf 29.22 & 30.07 & 28.17 & 29.07 \\ 
& ReFMT$_\F$ & 30.20 & 29.33 & 29.53 & 15.71 & 16.31 & 15.10 & 15.59 \\ 
\end{longtable}
}

\subsection{Simulation studies based on a real dataset}

We further evaluate the proposed data-adaptive rerandomization procedures using the Student Achievement and Retention Project (STAR) data \citep{angrist2009incentives}. 
The STAR study was a $2^2$ completely randomized factorial experiment (CRFE) at a Canadian university to evaluate two interventions: the Student Support Program (SSP), which provided academic support services, and the Student Fellowship Program (SFP), which offered scholarships to students who achieved a specified first-year GPA. 
The experiment included four treatment combinations: neither SSP nor SFP, SSP only, SFP only, and both SSP and SFP, with 1,006, 250, 250, and 150 participants, respectively.
The outcome of interest is first-year grade point average (GPA). 
Following \citet{li2018asymptotic}, we use five covariates: high school GPA, gender, age, whether the student lived at home, and whether the student rarely postponed studying for tests. 
After excluding students with missing covariate values, the final sample sizes are 856, 216, 208, and 118 for the four treatment combinations, respectively.

Since only one potential outcome is observed for each student, the missing potential outcomes are imputed. 
For each treatment combination, we fit a linear regression of the observed outcomes on all covariates and their two-way interaction terms, and use the fitted model to impute the missing potential outcomes. 
The imputed potential outcomes are then truncated to the interval $[0,4]$ to ensure feasible GPA values. 
Although the imputation model is linear, the inclusion of covariate interaction terms induces nonlinear relationships with respect to the original covariates.
There are two main effects and one interaction effect in this experiment. 
For the resulting finite population, the squared multiple correlations between the factorial effect estimators and the covariate mean differences are
$(R_1^2,R_2^2,R_3^2)=(0.432,\,0.427,\,0.432)$.

We consider the same two priority settings as in the model-generated-data simulations: equal priority for all factorial effects and higher priority for main effects. 
Accordingly, we compare DA-ReO$_\FE$ with ReFM under equal priority and DA-ReO$_\F$ with ReFMT$_\F$ under main-effect priority. 
The acceptance probabilities, effect weights, and learning-sample proportions are specified as in the model-generated-data simulations, with the infeasible oracle procedures 
ReO$_\FE$ and ReO$_\F$ included for reference.
Table~\ref{tab:realdata} reports the average PRIV, empirical coverage probability (CP), and confidence interval length (Len).
Consistent with the simulation results, the proposed data-adaptive procedures yield greater variance reductions than their Mahalanobis-distance-based counterparts under both priority settings, with $\rho_n=0.3$ generally performing best. 
All methods achieve over 95\% coverage, while the proposed procedures yield shorter confidence intervals.

{\setlength{\tabcolsep}{4pt}
\spacingset{1.2}
\small
\begin{longtable}[!t]{@{}ccccccccccc@{}}
\caption{Average PRIV, empirical coverage probability (CP, \%), and confidence interval length (Len) in the STAR data.}\label{tab:realdata}\tabularnewline
\toprule\noalign{}
& & \multicolumn{6}{c}{Main}& \multicolumn{3}{c}{Interaction} \\
\cmidrule(lr){3-8}
\cmidrule(lr){9-11}
Setting & Scheme & \multicolumn{3}{c}{1} & \multicolumn{3}{c}{2} & \multicolumn{3}{c}{12} \\
\cmidrule(lr){3-5}
\cmidrule(lr){6-8}
\cmidrule(lr){9-11}
& & PRIV & CP & Len & PRIV & CP & Len & PRIV & CP & Len \\ 
\midrule\noalign{}
\endfirsthead
\toprule\noalign{}
& & \multicolumn{6}{c}{Main}& \multicolumn{3}{c}{Interaction} \\
\cmidrule(lr){3-8}
\cmidrule(lr){9-11}
Setting & Scheme & \multicolumn{3}{c}{1} & \multicolumn{3}{c}{2} & \multicolumn{3}{c}{12} \\
\cmidrule(lr){3-5}
\cmidrule(lr){6-8}
\cmidrule(lr){9-11}
& & PRIV & CP & Len & PRIV & CP & Len & PRIV & CP & Len \\ 
\midrule\noalign{}
\endhead
\bottomrule\noalign{}
\endlastfoot
\multirow{8}{*}{\shortstack{
same\\
\hspace{\fill}\\
importance}} &  ReO$_\FE$ & 39.98 & 96.23 & 0.133 & 39.39 & 96.05 & 0.133 & 39.84 & 96.26 & 0.133 \\ 
& \makecell{DA-ReO$_\FE$\\[-4pt] ($\rho_n=0.2$)} & 30.54 & 96.62 & 0.136 & 30.33 & 96.51 & 0.136 & 30.46 & 96.66 & 0.136 \\ 
& \makecell{DA-ReO$_\FE$\\[-4pt]($\rho_n=0.3$)} & \textbf{32.23} & 96.81 & 0.138 & \textbf{31.63} & 96.71 & 0.138 & \textbf{32.85} & 96.83 & 0.138 \\ 
& \makecell{DA-ReO$_\FE$\\ [-4pt]($\rho_n=0.4$)} & 31.81 & 96.98 & 0.139 & 31.40 & 96.87 & 0.139 & 31.89 & 97.01 & 0.139 \\ 
& ReFM & 25.53 & 96.24 & 0.148 & 25.54 & 96.01 & 0.148 & 25.24 & 96.14 & 0.148 \\
\cmidrule(lr){1-11}
\multirow{8}{*}{\shortstack{
varying\\
\hspace{\fill}\\
importance}} &  ReO$_\F$ & 41.46 & 96.27 & 0.132 &41.76 & 96.48 & 0.132 &37.87 & 96.25 & 0.136 \\ 
& \makecell{DA-ReO$_\F$\\ [-4pt]($\rho_n=0.2$)} & 31.39 & 96.76 & 0.136 & 31.06 & 96.91 & 0.135 & 28.68 & 96.58 & 0.139 \\ 
& \makecell{DA-ReO$_\F$\\ [-4pt]($\rho_n=0.3$)} & \textbf{32.79} & 96.95 & 0.137 & \textbf{32.50} & 97.13 & 0.137 & \textbf{30.68} & 96.72 & 0.140 \\ 
& \makecell{DA-ReO$_\F$\\ [-4pt]($\rho_n=0.4$)} & 32.54 & 97.13 & 0.139 & 31.88 & 97.30 & 0.139 & 30.40 & 96.86 & 0.141\\ 
& ReFMT$_\F$ & 26.30 & 96.22 & 0.147 & 26.51 & 96.12 & 0.147 & 22.61 & 96.11 & 0.150 \\ 
\end{longtable}
}

\section{Summary and discussion}\label{sec:summary}

In this paper, we propose a data-adaptive rerandomization procedure for $2^K$ factorial experiments that incorporates both the relative priority of factorial effects and the covariate importance structure for estimating them. 
We first derive an oracle rerandomization criterion and then develop a data-adaptive procedure to approximate it. 
Finite-population asymptotic theory and numerical studies show that the proposed procedure improves upon existing Mahalanobis-distance-based methods.
Our procedure extends the two-stage framework of \citet{liu2025bayesian} from treatment--control experiments to $2^K$ factorial designs. 
Unlike their method, which requires two independent datasets with identical asymptotic properties, ours requires only a single dataset while preserving valid finite-population asymptotic inference.

Several directions merit further investigation. 
First, the current procedure uses a subset of the data to estimate the covariate importance matrix $\bB$. 
Inspired by the ReB criterion for treatment--control experiments \citep{liu2025bayesian}, a richer characterization of the information from the learning subset, such as incorporating the covariance structure of $\bB$, may lead to further gains in estimation precision. 
Second, combining the proposed procedure with regression adjustment \citep{li2020rerandomization} may further improve its finite-sample performance. 
Finally, the optimal choice of the learning-sample proportion remains an important theoretical problem and warrants further study using more refined asymptotic techniques.

\section{Disclosure statement}\label{disclosure-statement}

The authors have no conflicts of interest to declare.

\section{Data Availability Statement}\label{data-availability-statement}

Deidentified data is available at \url{https://www.scidb.cn/en/anonymous/RTdCUjNp}.

\phantomsection\label{supplementary-material}
\bigskip

\begin{center}

{\large\bf SUPPLEMENTARY MATERIAL}

\end{center}

\begin{description}
\item[Supplementary Material for ``Data-Adaptive  Rerandomization for $2^K$ Factorial Designs'':]
This file includes all the technical proofs and additional simulation results. (.pdf file)
\end{description}

\bibliography{bibliography.bib}

\newpage
\appendix

\setcounter{equation}{0}
\renewcommand{\theequation}{A.\arabic{equation}}
\setcounter{theorem}{0}
\renewcommand{\thetheorem}{A\arabic{theorem}}
\setcounter{lemma}{0}
\renewcommand{\thelemma}{A\arabic{lemma}}

\section{Proof}
\subsection{Proof of Theorem \ref{thm:ReOMA}}
\begin{proof}
Under Assumption \ref{assump:regularity_fac}, for $(\bC',\bD')' \sim \mathcal{N}(\bm 0, \bm V)$, similar to Proposition A1 in \cite{li2018asymptotic}, we have
\begin{equation}\label{eq:Asymptotics4Phi}
\left.\begin{pmatrix}
\sqrt{n}(\hat{\btau}-\btau)  \\
\sqrt{n}\hat\btau_\bx
\end{pmatrix}\right|\sqrt{n}\hat\btau_\bx \in \mathcal{B}_{\phi} \overset{\cdot}{\sim} 
\left.\begin{pmatrix}
\bC  \\
\bD 
\end{pmatrix}\right|\bD \in \mathcal{B}_{\phi},\ \forall \phi \in \Phi_\alpha. 
\end{equation}
It is easy to check that the criterion $\phi_\bomega$ satisfies Assumption \ref{assump:criterion}, and thus, $\phi_\bomega \in \Phi_\alpha$.
Moreover, the target in Eq.~\eqref{eq:new_target} is equivalent to 
\[\begin{split}
&\arg\min_{\phi \in \Phi_\alpha} \sum_{f=1}^F \omega_f \cdot \frac{\bbV_a(\tilde{\tau}_{\bx,f} \mid \phi=1)}{\bbV_a(\tilde{\tau}_{\bx,f})}
= \arg\min_{\phi \in \Phi_\alpha} \sum_{f=1}^F \frac{\omega_f}{\lim_n \lambda_f} \cdot \bbV_a(\tilde{\tau}_{\bx,f} \mid \phi=1)\\
=& \arg\min_{\phi \in \Phi_\alpha}  \lim_n \sum_{f=1}^F \frac{\omega_f}{\lambda_f} \cdot \bbV(\be_f'\bQ_\bB \bB' \bD \mid \bD \in \mathcal{B}_{\phi})
\\
=& \arg\min_{\phi \in \Phi_\alpha} 
\bbE\left(\bD'\bB\bQ_\bB'\bLambda^{-1}\bomega\bQ_\bB\bB'\bD \mid \bD \in \mathcal{B}_{\phi}\right),
\end{split}\]
where $\be_f$ is the $f$th canonical basis vector in $\mathbb{R}^F$.
According to Lemma A1 in \cite{liu2025bayesian}, the solution to the above optimization problem is 
\[\begin{split}
\phi_\bomega(\bD,\bV_{\bx\bx},\bB) = 1\left\{\bD'\bB\bQ_\bB'\bLambda^{-1}\bomega\bQ_\bB\bB'\bD \leq \xi_{\bomega,\alpha}\right\}.
\end{split}\]
Therefore, $\phi_\bomega$ is the solution of the target in Eq.~\eqref{eq:new_target}.
Above all, we complete the proof.
\end{proof}

\subsection{Proof of Theorem \ref{thm:AsymptoticDistri}}
\begin{proof}
Under Assumption \ref{assump:regularity_fac} and Assumption \ref{assump:V_tautau}, $\bQ_\bB$ is well-defined.
By Eq.~\eqref{eq:Asymptotics4Phi},
\[\begin{split}
\sqrt{n}(\hat{\btau}-\btau) \mid \phi_{\bomega}=1 &\overset{\cdot}{\sim} (\bC- \bB'\bm D)+ \bB'\bm D\mid \bm D \in \mathcal{B}_{\phi_{\bomega}}\\&\sim (\bC- \bB'\bm D) + \bB'\bD\mid \bD'\bB\bQ_\bB'\bomega\bLambda^{-1}\bQ_\bB\bB'\bD \leq \xi_{\bomega,\alpha}\\
&\sim (\bC- \bB'\bm D) + \bQ_\bB^{-1}\bQ_\bB\bB'\bD\mid \bD'\bB\bQ_\bB'\bomega\bLambda^{-1}\bQ_\bB\bB'\bD \leq \xi_{\bomega,\alpha},
\end{split}\]
where $(\bC',\bm D')' \sim \cN(\bZero,\bV)$, and $(\bC- \bB'\bm D)\sim \cN(\bZero,\bV_{\btau\btau}^{\perp})$
is independent of $\bD$.
Let $\bm\eta = \bLambda^{-1/2}\bQ_\bB \bB'\bD$, then $\bm\eta \sim \cN(\bZero,\bI_F)$, and thus,
\[\begin{split}
\sqrt{n}(\hat{\btau}-\btau) \mid \phi_{\bomega}=1 &\overset{\cdot}{\sim} 
(\bV_{\btau\btau}^{\perp})^{1/2}\bvarepsilon + \bQ_\bB^{-1}\bLambda^{1/2}\bm\eta \mid \bm\eta'\bomega \bm\eta \leq \xi_{\bomega,\alpha},
\end{split}\]
where $\bvarepsilon\sim \cN(\bZero,\bI_F)$ is independent of $\bm\eta$.

Similarly, 
\[\begin{split}
\sqrt{n}(\hat{\btau}-\btau) \mid \tilde\phi_{\bomega}=1 &\overset{\cdot}{\sim} (\bC- \bB'\bm D) + \bB'\bD\mid \bD'\bB\left(\bV_{\btau\btau}^{||}\right)^{-1}\bB'\bD \leq \xi_{F,\alpha}\\
&\sim (\bC- \bB'\bm D) + \left(\bV_{\btau\btau}^{||}\right)^{1/2}\left(\bV_{\btau\btau}^{||}\right)^{-1/2}\bB'\bD\mid \bD'\bB\left(\bV_{\btau\btau}^{||}\right)^{-1}\bB'\bD \leq \xi_{F,\alpha},
\end{split}\]
where $(\bC',\bm D')' \sim \cN(\bZero,\bV)$, $(\bC- \bB'\bm D)\sim \cN(\bZero,\bV_{\btau\btau}^{\perp})$
is independent of $\bD$, and 
$\bB'\bD \sim \cN(\bZero,\bV_{\btau\btau}^{||})$.
Let $\bm\eta = \left(\bV_{\btau\btau}^{||}\right)^{-1/2} \bB'\bD$, then $\bm\eta \sim \cN(\bZero,\bI_F)$, and thus,
\[\begin{split}
\sqrt{n}(\hat{\btau}-\btau) \mid \tilde\phi_{\bomega}=1 &\overset{\cdot}{\sim} 
(\bV_{\btau\btau}^{\perp})^{1/2}\bvarepsilon + \left(\bV_{\btau\btau}^{||}\right)^{1/2}\bm\eta \mid \bm\eta'\bm\eta \leq \xi_{F,\alpha}
\sim (\bV_{\btau\btau}^{\perp})^{1/2}\bvarepsilon + \left(\bV_{\btau\btau}^{||}\right)^{1/2}\bzeta_{F,\alpha}.
\end{split}\]
\end{proof}

\subsection{Proof of Lemma \ref{lem:R2}}
\begin{proof}
By definition, for $f = 1,\ldots,F$,
\[\begin{split}
\bB_f'\hat{\btau}_\bx &= \be_f'\bB'\hat{\btau}_\bx
= \be_f'\bQ_\bB^{-1}\tilde{\btau}_{\bx},
\end{split}\]
where $\be_f$ is the $F$-dimensional vector with its $f$th entry equal to one and all other entries equal to zero.
Since $\bI = \sum_{i=1}^F \be_i\be_i'$, then 
\[\begin{split}
\bB_f'\hat{\btau}_\bx
= \be_f'\bQ_\bB^{-1}\sum_{i=1}^F \be_i\be_i'\tilde{\btau}_{\bx}=\sum_{i=1}^F \be_f'\bQ_\bB^{-1}\be_i\be_i'\tilde{\btau}_{\bx}=\sum_{i=1}^F \be_f'\bQ_\bB^{-1}\be_i\tilde{\tau}_{\bx,i}.
\end{split}\]
Since $\bQ_\bB$ is a lower triangular matrix, then 
$\bQ_\bB^{-1}$ is also a lower triangular matrix, and thus, $\be_f'\bQ_\bB^{-1}\be_i=0$ for $i>f$.
Hence, 
\[\begin{split}
\bB_f'\hat{\btau}_\bx
=\sum_{i=1}^f \be_f'\bQ_\bB^{-1}\be_i\tilde{\tau}_{\bx,i}.
\end{split}\]
Therefore,
$\cov(\hat{\tau}_f,\tilde{\tau}_{\bx,j}) = \cov(\bB_f'\hat{\btau}_\bx,\tilde{\tau}_{\bx,j}) = 0$ for $j > f$.
Thus,
$R^2_f[j]=0$ for $j > f$.
Moreover,
\[\begin{split}
\bbV(\bB_f'(\sqrt{n}\hat{\btau}_\bx)) &= 
\sum_{i=1}^f (\be_f'\bQ_\bB^{-1}\be_i)^2\cdot\bbV(\sqrt{n}\tilde{\tau}_{\bx,i}) = \sum_{i=1}^f (\be_f'\bQ_\bB^{-1}\be_i)^2\cdot \lambda_i.
\end{split}\]
Since for $1\leq j\leq f$, 
\[\begin{split}
\cov(\bB_f'(\sqrt{n}\hat{\btau}_\bx),\sqrt{n}\tilde{\tau}_{\bx,j}) &= \be_f'\bQ_\bB^{-1}\be_j\cdot \lambda_j,
\end{split}\] 
we have 
\[\begin{split}
\bbV(\sqrt{n}\bB_f'\hat{\btau}_\bx) &= \sum_{j=1}^f (\be_f'\bQ_\bB^{-1}\be_j)^2\cdot \lambda_j = \sum_{j=1}^f\frac{(\cov(\sqrt{n}\bB_f'\hat{\btau}_\bx,\sqrt{n}\tilde{\tau}_{\bx,j}))^2}{\lambda_j}\\ &= \sum_{j=1}^f\frac{(\cov(\sqrt{n}\bB_f'\hat{\btau}_\bx,\sqrt{n}\tilde{\tau}_{\bx,j}))^2}{\bbV(\sqrt{n}\tilde{\tau}_{\bx,j})}.
\end{split}\]
Therefore, $
\sum_{j=1}^f \frac{(\cov(\bB_f'(\sqrt{n}\hat{\btau}_\bx),\sqrt{n}\tilde{\tau}_{\bx,j}))^2}{\bbV(\bB_f'(\sqrt{n}\hat{\btau}_\bx))\bbV(\sqrt{n}\tilde{\tau}_{\bx,j})}= 1$.
By definition, we have 
\[\begin{split}
\sum_{j=1}^f R_{f}^2[j] 
&= \sum_{j=1}^f \frac{(\cov(\bB_f'\hat{\btau}_\bx,\tilde{\tau}_{\bx,j}))^2}{
\bbV(\bB_f'\hat{\btau}_\bx)\bbV(\tilde{\tau}_{\bx,j})}\cdot R^2_f
= R^2_f.
\end{split}\]
\end{proof}

\subsection{Proof of Lemma \ref{lem:cf}}
\begin{proof}
First, to show that $c_1\leq\ldots\leq c_F$ for $\omega_1\geq\ldots\geq \omega_F\geq 0$, it suffices to show that for $\omega_1,\ldots,\omega_F \geq 0$, if $\omega_i \geq \omega_k$, then $$\bbE\left[\eta_i^2 \left| \sum_{j=1}^F \omega_j\eta_j^2 \leq a\right.\right] \leq \bbE\left[\eta_k^2 \left| \sum_{j=1}^F \omega_j\eta_j^2 \leq a\right.\right],$$
where $\eta_1,\ldots,\eta_F$ are i.i.d. standard normal random variables.
By the law of iterated expectation, 
\[\begin{split}
&\bbE\left[\eta_i^2 \left| \sum_{j=1}^F \omega_j\eta_j^2 \leq a\right.\right]\\ 
=& \bbE\left[\bbE\left[\eta_i^2 \left| \left.\sum_{j=1}^F \omega_j\eta_j^2 \leq a, \sum_{j\not= i,k} \omega_j\eta_j^2\right.\right]\right|\sum_{j=1}^F \omega_j\eta_j^2 \leq a\right]\\
=&\int\bbE\left[\eta_i^2 \left| \sum_{j=1}^F \omega_j\eta_j^2 \leq a, \sum_{j\not= i,k} \omega_j\eta_j^2 = m\right.\right]
f_{\sum_{j\not= i,k} \omega_j\eta_j^2 \mid \sum_{j=1}^F \omega_j\eta_j^2 \leq a}(m)
dm
\\
=&\int\bbE\left[\eta_i^2 \left| \omega_i\eta_i^2 + \omega_k\eta_k^2  \leq a-m, \sum_{j\not= i,k} \omega_j\eta_j^2 = m\right.\right]
f_{\sum_{j\not= i,k} \omega_j\eta_j^2 \mid \sum_{j=1}^F \omega_j\eta_j^2 \leq a}(m) dm
\\
=&\int\bbE\left[\eta_i^2 \left| \omega_i\eta_i^2 + \omega_k\eta_k^2  \leq a-m\right.\right]
f_{\sum_{j\not= i,k} \omega_j\eta_j^2 \mid \sum_{j=1}^F \omega_j\eta_j^2 \leq a}(m) dm.
\end{split}\]
It suffices to show that for any constant $c \geq 0$, 
$$\bbE\left[\eta_i^2 \left| \omega_i\eta_i^2 + \omega_k\eta_k^2 \leq c\right.\right] \leq \bbE\left[\eta_k^2 \left| \omega_i\eta_i^2 + \omega_k\eta_k^2 \leq c\right.\right].$$
Since 
\[\begin{split}
\bbE\left[\eta_i^2 \mid \omega_i\eta_i^2 + \omega_k\eta_k^2 \leq c\right] = \int_0^\infty 1-\bbP(\eta_i^2 \leq t \mid \omega_i\eta_i^2 + \omega_k\eta_k^2 \leq c) dt,
\end{split}\]
and 
\[\begin{split}
\bbE\left[\eta_k^2 \mid \omega_i\eta_i^2 + \omega_k\eta_k^2 \leq c\right]&= \bbE\left[\eta_i^2 \mid \omega_k\eta_i^2 + \omega_i\eta_k^2 \leq c\right]
\\&= \int_0^\infty 1-\bbP(\eta_i^2 \leq t \mid \omega_k\eta_i^2 + \omega_i\eta_k^2 \leq c) dt,
\end{split}\]
it suffices to show that for any $t,c\geq 0$, 
\[\begin{split}
\bbP(\eta_i^2 \leq t \mid \omega_i\eta_i^2 + \omega_k\eta_k^2 \leq c) \geq \bbP(\eta_i^2 \leq t \mid \omega_k\eta_i^2 + \omega_i\eta_k^2 \leq c).
\end{split}\]
Since
\[\begin{split}
\bbP(\omega_i\eta_i^2 + \omega_k\eta_k^2 \leq c) = \bbP(\omega_k\eta_i^2 + \omega_i\eta_k^2 \leq c),
\end{split}\]
it suffices to show that for any $t,c\geq 0$, 
\[\begin{split}
\bbP(\eta_i^2 \leq t, \omega_i\eta_i^2 + \omega_k\eta_k^2 \leq c) \geq \bbP(\eta_i^2 \leq t, \omega_k\eta_i^2 + \omega_i\eta_k^2 \leq c).
\end{split}\]
Since
\[\begin{split}
\bbP(\eta_i^2 \leq t,\eta_k^2 \leq t, \omega_i\eta_i^2 + \omega_k\eta_k^2 \leq c) = \bbP(\eta_i^2 \leq t,\eta_k^2 \leq t, \omega_k\eta_i^2 + \omega_i\eta_k^2 \leq c),
\end{split}\]
it suffices to show that for any $t,c\geq 0$, 
\begin{equation}\label{eq:comparisonEq}
\bbP(\eta_i^2 \leq t,\eta_k^2 > t, \omega_i\eta_i^2 + \omega_k\eta_k^2 \leq c) \geq \bbP(\eta_i^2 \leq t,\eta_k^2 > t, \omega_k\eta_i^2 + \omega_i\eta_k^2 \leq c).
\end{equation}
When $\eta_i^2 \leq t,\eta_k^2 > t$, we have 
\[\begin{split}
(\omega_i\eta_i^2 + \omega_k\eta_k^2)-(\omega_k\eta_i^2 + \omega_i\eta_k^2) = (\omega_i-\omega_k)(\eta_i^2-\eta_k^2) \leq 0,
\end{split}\]
and thus, 
\[\begin{split}
\{\eta_i^2 \leq t,\eta_k^2 > t, \omega_k\eta_i^2 + \omega_i\eta_k^2 \leq c\} \subseteq \{\eta_i^2 \leq t,\eta_k^2 > t, \omega_i\eta_i^2 + \omega_k\eta_k^2 \leq c\}.
\end{split}\]
Hence, Eq.~\eqref{eq:comparisonEq} holds and we complete the proof.
Therefore, for $\omega_1\geq \ldots \geq \omega_F$, we have $c_1 \leq \ldots\leq c_F$.

Next, we show that $c_f \leq 1$ for $f=1,\ldots,F$.
If $W_f = 0$, then $c_f = 1$.
Otherwise,
by definition,
\[\begin{split}
c_f &= \bbE\left[\eta_f^2 \mid \sum_{i=1}^F \omega_i \cdot \eta_i^2 \leq \xi_{\bomega,\alpha}\right]\\
&=  \bbE\left[\left.\bbE\left[\eta_f^2 \mid \sum_{i=1}^F \omega_i \cdot \eta_i^2 \leq \xi_{\bomega,\alpha},\sum_{i \neq f} \omega_i \cdot \eta_i^2\right]\right|\sum_{i=1}^F \omega_i \cdot \eta_i^2 \leq \xi_{\bomega,\alpha}\right]\\
&= \bbE\left[\left.\bbE\left[\eta_f^2 \mid \eta_f^2 \leq \frac{\xi_{\bomega,\alpha} - \sum_{i \neq f} \omega_i \cdot \eta_i^2}{\omega_f},\sum_{i \neq f} \omega_i \cdot \eta_i^2\right]\right|\sum_{i=1}^F \omega_i \cdot \eta_i^2 \leq \xi_{\bomega,\alpha}\right].
\end{split}\]
Since $\bbE[\eta_f^2 \mid \eta_f^2 \leq u] \leq 1$ for any $u\geq 0$,
then we have $c_f \leq 1$.

According to the proof of Theorem 3 in  \cite{lu2023design}, for $\alpha \rightarrow 0$, when $w_1,\ldots,w_F > 0$,
\[\begin{split}
c_f = p_F \omega_f^{-1} \left(\prod_{j=1}^F \omega_j\right)^{1/F}\alpha^{2/F}(1+o(1)),
\end{split}\]
where $p_F$ is a constant dependent on $F$ only. 
When $\omega_1=\ldots=\omega_F=1$, we have 
$c_f = v_{F,\alpha} = p_F\alpha^{2/F}(1+o(1))$ for all $f=1,\ldots,F$.
Hence, 
\[\begin{split}
c_f = \omega_f^{-1} \left(\prod_{j=1}^F \omega_j\right)^{1/F} v_{F,\alpha}(1+o(1)).
\end{split}\]
\end{proof}

\subsection{Proof of Theorem \ref{thm:PRIASV}}
\begin{proof}
Theorem \ref{thm:AsymptoticDistri} implies that the asymptotic covariance of $\sqrt{n}(\hat{\btau}-\btau)$ under $\phi_{\bomega}$ is 
\[\begin{split}
\lim_{n\to\infty} \left(\bV_{\btau\btau}^{\perp} +  \bQ_\bB^{-1}\bC\bLambda(\bQ_\bB^{-1})'\right),
\end{split}\]
where $\bC=\text{diag}\{c_1,\ldots,c_F\}$.
Since $\bQ_\bB\bB'\bV_{\bx\bx}\bB\bQ_\bB' = \bLambda$, we have $\bB'\bV_{\bx\bx}\bB = \bQ_\bB^{-1}\bLambda(\bQ_\bB^{-1})'$. Hence, 
\[\begin{split}
\bQ_\bB^{-1}\bC\bLambda(\bQ_\bB^{-1})' &= (\bB'\bV_{\bx\bx}\bB)(\bB'\bV_{\bx\bx}\bB)^{-1}\bQ_\bB^{-1}\bC\bLambda(\bQ_\bB^{-1})'(\bB'\bV_{\bx\bx}\bB)^{-1}(\bB'\bV_{\bx\bx}\bB)\\&= (\bB'\bV_{\bx\bx}\bB)\bQ_\bB'\bLambda^{-1}\bC\bQ_\bB(\bB'\bV_{\bx\bx}\bB).
\end{split}\]
Denote $(\bB'\bV_{\bx\bx}\bB)\bQ_\bB'$ as $(\bm m_1,\ldots,\bm m_F)$.
Since $(\bB'\bV_{\bx\bx}\bB)\bQ_\bB' = \cov(\bB'(\sqrt{n}\hat{\btau}_\bx),\bQ_\bB\bB'(\sqrt{n}\hat{\btau}_\bx)) = \cov(\bB'(\sqrt{n}\hat{\btau}_\bx),\sqrt{n}\Tilde{\btau}_\bx)$, then $\bm m_f = \cov(\bB'(\sqrt{n}\hat{\btau}_\bx),\sqrt{n}\tilde{\tau}_{\bx,f})=\cov(\sqrt{n}(\hat{\btau}-\btau),\sqrt{n}\tilde{\tau}_{\bx,f})$. Then 
\[\begin{split}
\bQ_\bB^{-1}\bC\bLambda(\bQ_\bB^{-1})' &= \sum_{f=1}^F \frac{c_f}{\lambda_f} \cdot\bm m_f \bm m_f'\\ &= \sum_{f=1}^F c_f \cdot \cov(\sqrt{n}(\hat{\btau}-\btau),\sqrt{n}\tilde{\tau}_{\bx,f})(\bbV(\sqrt{n}\tilde{\tau}_{\bx,f}))^{-1}\cov(\sqrt{n}\tilde{\tau}_{\bx,f},\sqrt{n}(\hat{\btau}-\btau)),
\end{split}\]
where $\cov(\sqrt{n}(\hat{\btau}-\btau),\sqrt{n}\tilde{\tau}_{\bx,f})(\bbV(\sqrt{n}\tilde{\tau}_{\bx,f}))^{-1}\cov(\sqrt{n}\tilde{\tau}_{\bx,f},\sqrt{n}(\hat{\btau}-\btau))$ is the covariance of $\sqrt{n}(\hat{\btau}-\btau)$ explained by $\sqrt{n}\tilde{\tau}_{\bx,f}$, which equals to $\bV_{\btau\btau}^{||}[f]$. Then we have 
\[\begin{split}
\bQ_\bB^{-1}\bC\bLambda(\bQ_\bB^{-1})' &= \sum_{f=1}^F c_f \cdot\bV_{\btau\btau}^{||}[f].
\end{split}\]
Since the asymptotic covariance of $\sqrt{n}(\hat{\btau}-\btau)$ under CRFE is $\lim_{n\to\infty}\bV_{\btau\btau}$, the reduction is 
\[\begin{split}
\lim_{n\to\infty} \left(\bV_{\btau\btau}^{||} - \sum_{f=1}^F c_f \cdot\bV_{\btau\btau}^{||}[f]\right)
&= \lim_{n\to\infty} \left(\sum_{f=1}^F (1-c_f) \cdot\bV_{\btau\btau}^{||}[f]\right).
\end{split}\]

%
Since
\[\begin{split}
&\be_f'\left[\sum_{j=1}^F (1-c_j) \cdot\bV_{\btau\btau}^{||}[j]\right]\be_f\\
=&\sum_{j=1}^F (1-c_j) \cdot \be_f'\cov(\sqrt{n}(\hat{\btau}-\btau),\sqrt{n}\tilde{\tau}_{\bx,j})(\bbV(\sqrt{n}\tilde{\tau}_{\bx,j}))^{-1}\cov(\sqrt{n}\tilde{\tau}_{\bx,j},\sqrt{n}(\hat{\btau}-\btau))\be_f\\ 
=& \sum_{j=1}^F (1-c_j) \cdot \cov(\sqrt{n}(\hat{\tau}_f-\tau_f),\sqrt{n}\tilde{\tau}_{\bx,j})(\bbV(\sqrt{n}\tilde{\tau}_{\bx,j}))^{-1}\cov(\sqrt{n}\tilde{\tau}_{\bx,j},\sqrt{n}(\hat{\tau}_f-\tau_f))\\=& \sum_{j=1}^F (1-c_j) \cdot R^2_f[j]\cdot \bbV(\sqrt{n}(\hat{\tau}_f-\tau_f)),
\end{split}\]
the reduction in asymptotic sampling variance of $\sqrt{n}(\hat{\tau}_f-\tau_f)$ is 
\[\begin{split}
\lim_{n\to\infty} \sum_{j=1}^F (1-c_j) \cdot R^2_f[j]\cdot \bbV_a(\sqrt{n}(\hat{\tau}_f-\tau_f)).
\end{split}\]
Hence, the PRIASV of $\sqrt{n}(\hat{\tau}_f-\tau_f)$ under $\phi_{\bomega}$ is 
\[\begin{split}
100 \times \left[1-
\frac{\bbV_a(\sqrt{n}(\hat{\tau}_f-\tau_f)\mid \phi_\bomega = 1)}{\bbV_a(\sqrt{n}(\hat{\tau}_f-\tau_f))}\right]
= \lim_{n\to\infty} 100 \times \sum_{j=1}^F (1-c_j) \cdot R^2_f[j].
\end{split}\]
By Lemma \ref{lem:R2}, $R^2_f[j] = 0$ for $j > f$, the PRIASV is 
$\lim_{n\to\infty} 100 \times \sum_{j=1}^f (1-c_j) \cdot R^2_f[j]$.

When $\omega_1=\ldots=\omega_F$, we have $c_1=\ldots=c_F = v_{F,\alpha}$, and the results for $\tilde{\phi}_{\bomega}$ is trivial.
\end{proof}

\subsection{Proof of Corollary  \ref{cor:LowerBound}}
\begin{proof}
Lemma \ref{lem:cf} shows that for $\omega_1\geq \ldots \geq \omega_f$, we have 
$c_1\leq \ldots \leq c_F$, and thus, 
\[\begin{split}
\sum_{j=1}^f (1-c_j)R^2_f[j]\geq (1-c_f) \sum_{j=1}^f R^2_f[j] = (1-c_f) R^2_f,
\end{split}\]
where the last equation comes from Lemma \ref{lem:R2}.
\end{proof}


\subsection{Proof of Theorem \ref{thm:peakedness}}
\begin{proof}
We first prove the following two lemmas.
\begin{lemma}\label{lem:peakedlemma}
Let $\bEta \sim \mathcal{N}(\bZero,\bI_{m_1})$ and $\bxi \sim \mathcal{N}(\bZero,\bI_{m_2})$ be independent, then $(\bEta\mid \bEta'\bEta\leq \xi_{m_1,\alpha}) \succ (\bEta\mid \bEta'\bEta+\bxi'\bxi\leq \xi_{m_1+m_2,\alpha})$.
\end{lemma}

\begin{lemma}\label{lem:ReOisMorePeaked}
Let $\bV_1$ be an $m\times m$ positive semi-definite matrix, $\bvarepsilon \sim \mathcal{N}(\bm 0, \bI_m)$ be a standard Gaussian random vector, $\bzeta_{m,\alpha} \sim \bD_1|\bD_1'\bD_1 \leq \xi_{m,\alpha}$, $\bzeta_{k,\alpha} \sim \bD_2|\bD_2'\bD_2 \leq \xi_{k,\alpha}$, where $k \geq m$, $\bD_1 \sim \mathcal{N}(\bm 0,\bI_m)$, and $\bD_2 \sim \mathcal{N}(\bm 0,\bI_k)$.
Then $(\bV_1)_m^{1/2}\bzeta_{m,\alpha} \succ (\bV_1)_k^{1/2}\bzeta_{k,\alpha}$.
\end{lemma}

\begin{proof}[Proof of Lemma \ref{lem:peakedlemma}]
Let $a_1 = \xi_{m_1,\alpha}$, and $a_2 = \xi_{m_1+m_2,\alpha}$.
For any convex symmetric set $\mathcal{K} \subset \mathbb{R}^{m_1}$, $\bvarepsilon \sim \mathcal{N}(\bm 0,\bI_{m_1})$, let $||\bvarepsilon||_2 = (\bvarepsilon'\bvarepsilon)^{1/2}$, and $G(r) = \bbP(\bvarepsilon\in \mathcal{K}\mid ||\bvarepsilon||_2 = r)$. Let $\bvarphi$ be a random vector uniformly distributed on $m_1$ dimensional unit sphere, then $G(r)$ can be simplified as 
\[\begin{split}
    G(r) =\bbP\left(||\bvarepsilon||_2\cdot \frac{\bvarepsilon}{||\bvarepsilon||_2} \in \mathcal{K}\mid ||\bvarepsilon||_2 = r\right) = \bbP(r\bvarphi \in \mathcal{K}) = \bbP(\bvarphi \in r^{-1}\mathcal{K}),
\end{split}\]
where $r^{-1}\mathcal{K} = \{r^{-1}\bx:\bx \in \mathcal{K}\}$. According to the proof of Lemma A13 in \cite{li2020rerandomization},  $G(r)$ is nonincreasing in $r \in [0,\infty)$. Hence, 
\[\begin{split}
    \bbP(\bEta \in \mathcal{K}\mid \bEta'\bEta\leq a_1) &= \bbE\left\{\bbP(\bEta \in \mathcal{K}\mid \bEta'\bEta\leq a_1,||\bEta||_2)\mid \bEta'\bEta\leq a_1\right\}\\
    &= \bbE\left\{G(||\bEta||_2)\mid \bEta'\bEta\leq a_1\right\},\\
    \bbP(\bEta \in \mathcal{K}\mid \bEta'\bEta + \bxi'\bxi\leq a_2) &= \bbE\left\{\bbP(\bEta \in \mathcal{K}\mid \bEta'\bEta+ \bxi'\bxi\leq a_2,||\bEta||_2)\mid \bEta'\bEta+ \bxi'\bxi\leq a_2\right\}\\
    &= \bbE\left\{G(||\bEta||_2)\mid \bEta'\bEta+ \bxi'\bxi\leq a_2\right\}.
\end{split}\]
By the law of total expectation, 
\[\begin{split}
\bbE\left\{G(||\bEta||_2)\mid \bEta'\bEta\leq a_1\right\} &= \bbE\left\{G(||\bEta||_2)\mid \bEta'\bEta\leq a_1, \bEta'\bEta+ \bxi'\bxi\leq a_2\right\}\bbP(\bEta'\bEta+ \bxi'\bxi\leq a_2\mid \bEta'\bEta\leq a_1)\\ &+ \bbE\left\{G(||\bEta||_2)\mid \bEta'\bEta\leq a_1, \bEta'\bEta+ \bxi'\bxi> a_2\right\}\bbP(\bEta'\bEta+ \bxi'\bxi> a_2\mid \bEta'\bEta\leq a_1),\\
\bbE\left\{G(||\bEta||_2)\mid \bEta'\bEta+ \bxi'\bxi\leq a_2\right\} &= \bbE\left\{G(||\bEta||_2)\mid \bEta'\bEta\leq a_1, \bEta'\bEta+ \bxi'\bxi\leq a_2\right\}\bbP(\bEta'\bEta\leq a_1\mid \bEta'\bEta+ \bxi'\bxi\leq a_2)\\ &+ \bbE\left\{G(||\bEta||_2)\mid \bEta'\bEta> a_1, \bEta'\bEta+ \bxi'\bxi \leq a_2\right\}\bbP(\bEta'\bEta> a_1\mid \bEta'\bEta+ \bxi'\bxi\leq a_2).
\end{split}\]
Since $\bbP(\bEta'\bEta\leq a_1)=\bbP(\bEta'\bEta+ \bxi'\bxi \leq a_2) = \alpha$, we have 
\[\begin{split}
    \bbP(\bEta'\bEta+ \bxi'\bxi\leq a_2\mid \bEta'\bEta\leq a_1) &= \frac{\bbP(\bEta'\bEta+ \bxi'\bxi\leq a_2, \bEta'\bEta\leq a_1)}{\bbP(\bEta'\bEta\leq a_1)}\\
    &= \frac{\bbP(\bEta'\bEta+ \bxi'\bxi\leq a_2, \bEta'\bEta\leq a_1)}{\bbP(\bEta'\bEta+ \bxi'\bxi \leq a_2)}\\
    &= \bbP(\bEta'\bEta\leq a_1\mid \bEta'\bEta+ \bxi'\bxi\leq a_2),
\end{split}\]
and
\[\begin{split}
\bbP(\bEta'\bEta+ \bxi'\bxi> a_2\mid \bEta'\bEta\leq a_1) 
&= \bbP(\bEta'\bEta> a_1\mid \bEta'\bEta+ \bxi'\bxi\leq a_2).
\end{split}\]
Since $G(r)$ is nonincreasing in $r$, we have 
\[\begin{split}
    E\left\{G(||\bEta||_2)\mid \bEta'\bEta\leq a_1, \bEta'\bEta+ \bxi'\bxi> a_2\right\} \geq G(\sqrt{a_1}) \geq E\left\{G(||\bEta||_2)\mid \bEta'\bEta> a_1, \bEta'\bEta+ \bxi'\bxi \leq a_2\right\}.
\end{split}\]
Therefore, 
\[\begin{split}
    E\left\{G(||\bEta||_2)\mid \bEta'\bEta\leq a_1\right\} \geq E\left\{G(||\bEta||_2)\mid \bEta'\bEta+ \bxi'\bxi\leq a_2\right\},
\end{split}\]
and thus, 
$P(\bEta \in \mathcal{K}\mid \bEta'\bEta\leq a_1) \geq P(\bEta \in \mathcal{K}\mid \bEta'\bEta+ \bxi'\bxi \leq a_2)$. By definition of peakedness, we have $\bEta|\bEta'\bEta\leq a_1 \succ \bEta|\bEta'\bEta+\bxi'\bxi\leq a_2$.
\end{proof}

\begin{proof}[Proof of Lemma \ref{lem:ReOisMorePeaked}]
Let $\bm\phi_1 =  (\bV_1)_m^{1/2}\bzeta_{m,\alpha}$, and $\bm\phi_2 =  (\bV_1)_k^{1/2}\bzeta_{k,\alpha}$.
We first introduce Lemma A9 in \cite{li2020rerandomization}: Let $\bzeta_{m,\alpha} \sim \bD|\bD'\bD \leq \xi_{m,\alpha}$, where $\bD \sim \mathcal{N}(\bm 0,\bI_m)$. If two matrices $\bA$ and $\bB$ in $\mathbb{R}^{L\times m}$ satisfy $\bA\bA' = \bB\bB'$, then $\bA\bzeta_{m,\alpha} \sim \bB\bzeta_{m,\alpha}$.

Let the eigen-decomposition of $\bV_1$ be $\bGamma\bDelta^2\bGamma'$, where $\bGamma \in \mathbb{R}^{m\times m}$ is an orthogonal matrix, and $\bDelta^2 = \diag\{\rho_1^2,\ldots,\rho_m^2\}$. Then from Lemma A9 in \cite{li2020rerandomization}, 
\[\begin{split}
\bGamma'\bm\phi_2 &\sim  \bGamma'(\bV_1)_k^{1/2}\bzeta_{k,\alpha}\sim \left(\bGamma'\bV_1\bGamma\right)_k^{1/2}\bzeta_{k,\alpha}\\&\sim \left(\bGamma'\bGamma\bDelta^2\bGamma'\bGamma\right)_k^{1/2}\bzeta_{k,\alpha}\sim\left(\bDelta^2\right)_k^{1/2}\bzeta_{k,\alpha}\\
&\sim \left(\bDelta,\bZero_{m\times (k-m)}\right)\bzeta_{k,\alpha}.
\end{split}\]
Similarly, 
\[\begin{split}
\bGamma'\bm\phi_1 &\sim  \bGamma'(\bV_1)_m^{1/2}\bzeta_{m,\alpha}\sim \left(\bGamma'\bV_1\bGamma\right)_m^{1/2}\bzeta_{m,\alpha}\\
&\sim \left(\bGamma'\bGamma\bDelta^2\bGamma'\bGamma\right)_m^{1/2}\bzeta_{m,\alpha}\\
&\sim \bDelta\bzeta_{m,\alpha}.
\end{split}\]
Let $\bEta \sim \mathcal{N}(\bZero,\bI_m)$ and $\bxi \sim \mathcal{N}(\bZero,\bI_{k-m})$ be independent. Then 
\[\begin{split}
\bGamma'\bm\phi_2 
&\sim \left.\left(\bDelta,\bZero_{m\times (k-m)}\right)\binom{\bEta}{\bxi}\right|\bEta'\bEta + \bxi'\bxi \leq \xi_{k,\alpha}\\
&\sim \bDelta\bEta \mid \bEta'\bEta + \bxi'\bxi \leq \xi_{k,\alpha},
\end{split}\]
and 
\[\begin{split}
\bGamma'\bm\phi_1 
&\sim \bDelta\bEta \mid \bEta'\bEta \leq \xi_{m,\alpha}.
\end{split}\]
Lemma \ref{lem:peakedlemma} implies that $\bEta\mid\bEta'\bEta\leq \xi_{m,\alpha} \succ \bEta\mid\bEta'\bEta+\bxi'\bxi\leq \xi_{k,\alpha}$ for $k\geq m$, and Lemma 7.2 in \cite{dharmadhikari1988unimodality} implies that $(\bDelta\bEta\mid \bEta'\bEta\leq \xi_{m,\alpha}) \succ (\bDelta\bEta\mid\bEta'\bEta+\bxi'\bxi\leq \xi_{k,\alpha})$. 
Hence, $\bGamma'\bm\phi_1 \succ \bGamma'\bm\phi_2$, and $\bm\phi_1 = \bGamma\bGamma'\bm\phi_1 \succ \bGamma\bGamma'\bm\phi_2 = \bm\phi_2.$
\end{proof}

\begin{proof}[Proof of Theorem \ref{thm:peakedness}]
Based on the above two lemmas, we prove the theorem.
Let $m = F$, $k = Fp$, and $\bV_1 = \bV_{\btau\btau}^{||}$, 
Lemma \ref{lem:ReOisMorePeaked} 
implies 
\[\begin{split}
\big(\bV_{\btau\btau}^{||}\big)^{1/2}\bzeta_{F,\alpha} \succ \big(\bV_{\btau\btau}^{||}\big)_{Fp}^{1/2}\bzeta_{Fp,\alpha}.
\end{split}\]
Moreover,
as $\bvarepsilon$ is central symmetric unimodal,  Theorem 7.5 in \cite{dharmadhikari1988unimodality} implies that 
\[\begin{split}
\big(\bV_{\btau\btau}^{\perp}\big)^{1/2}\bvarepsilon + \big(\bV_{\btau\btau}^{||}\big)^{1/2}\bzeta_{F,\alpha} \succ\big(\bV_{\btau\btau}^{\perp}\big)^{1/2}\bvarepsilon + \big(\bV_{\btau\btau}^{||}\big)_{Fp}^{1/2}\bzeta_{Fp,\alpha}.
\end{split}\]
Recall that the asymptotic distribution of $\sqrt{n}(\hat{\btau}-\btau)$ under ReFM is 
$\big(\bV_{\btau\btau}^{\perp}\big)^{1/2}\bvarepsilon + \big(\bV_{\btau\btau}^{||}\big)_{Fp}^{1/2}\bzeta_{Fp,\alpha}$,
we complete the proof.
\end{proof}
\end{proof}

\subsection{Proof of Theorem \ref{thm:comparison}}

\begin{proof}
Suppose ReFMT$_\C$ splits the $p$-dimensional covariates into $T$ tiers of different importance, with the dimension of covariates being $p_t$ for $1\leq t \leq T$. 
Then the asymptotic distribution of $\sqrt{n}(\hat{\btau}-\btau)$ under ReFMT$_\C$ is 
\[\begin{split}
\sqrt{n}(\hat{\btau}-\btau) \overset{\cdot}{\sim}
(\bV_{\btau\btau}^{\perp})^{1/2}\bvarepsilon + \sum_{t=1}^T (\bV_t)^{1/2}_{Fp_t}\bzeta_{Fp_t,\alpha_t},
\end{split}\]
where $\sum_{t=1}^T \bV_t = \bV_{\btau\btau}^{||}$,
$p_t \geq 1$, and $\alpha = \prod_{t=1}^T \alpha_t$. 
Since $\alpha_t \leq 1$, then we have $\alpha_t \geq \alpha$, and hence $v_{Fp_t,\alpha_t} \geq v_{F,\alpha}$, for $1\leq t\leq T$. 
Then the asymptotic covariance of $\sqrt{n}(\hat{\btau}-\btau)$
\[
\lim_{n\to\infty}\left(\bV_{\btau\btau}^{\perp} + \sum_{t=1}^T v_{Fp_t,\alpha_t}\bV_t\right)
\]
satisfies
\[\begin{split}
\bV_{\btau\btau}^{\perp} + \sum_{t=1}^T v_{Fp_t,\alpha_t}\bV_t \geq \bV_{\btau\btau}^{\perp} + \sum_{t=1}^T v_{F,\alpha}\bV_t = \bV_{\btau\btau}^{\perp} + v_{F,\alpha}\bV_{\btau\btau}^{||},
\end{split}\]
i.e., the asymptotic covariance of $\sqrt{n}(\hat{\btau}-\btau)$ under ReO$_\FE$ is no larger than that under ReFMT$_\C$.
\end{proof}

\subsection{Proof of Theorem \ref{thm:DistTSMCRE}}
\begin{proof}
The procedure of S-CRFE is characterized by the joint distribution of $\bS$, $\bZ^{(1)}$, and $\bZ^{(2)}$.
Equivalently, the $n$ units can be regarded as being assigned into $2Q$ groups.
Let the treatment group indicator of the population be $\bq = \bq(\bS,\bZ) = \bq(\bS,\bZ^{(1)},\bZ^{(2)}) = (q_1,\ldots,q_n)'$, 
where for each $i=1,\ldots,n$ and $q=1,\ldots,Q$,
\[\begin{split}
q_i = \begin{cases}
q, & \text{if } S_i = 1 \text{ and } {Z}_i = q;\\
q+Q, & \text{if } S_i = 0 \text{ and } {Z}_i = q.
\end{cases}
\end{split}\]
The number of units assigned to each group is given by $\tilde{n}_q = n_{1q}$ and $\tilde{n}_{q+Q} = n_{2q}$ for $q=1,\ldots,Q$.
Then the estimators can be rearranged as
\[\begin{split}
\hat{\btau}_1(\bS,\bZ^{(1)}) &= \sum_{q=1}^Q \bA_q \cdot \tilde{n}_{q}^{-1}\sum_{i:q_i=q} Y_i(q),\quad
\hat{\btau}_2(\bS,\bZ^{(2)}) = \sum_{q=1}^Q \bA_q \cdot \tilde{n}_{q+Q}^{-1}\sum_{i:q_i=q+Q} Y_i(q),\\
\hat{\btau}_{\bx,1}(\bS,\bZ^{(1)}) &= \sum_{q=1}^Q \bA_q \otimes \left(\tilde{n}_{q}^{-1}\sum_{i:q_i=q} \bX_i\right),\quad
\hat{\btau}_{\bx,2}(\bS,\bZ^{(2)}) = \sum_{q=1}^Q \bA_q \otimes \left(\tilde{n}_{q+Q}^{-1}\sum_{i:q_i=q+Q} \bX_i\right).
\end{split}\]
For $i=1,\ldots,n$, define the pseudo potential outcomes under each group as
\[\begin{split}
\bU_i(q)=\bU_i(q+Q) = \begin{pmatrix}
Y_i(q)\\ \bX_i
\end{pmatrix},\ q=1,\ldots,Q.
\end{split}\]
Let $\bA_q = \bA_{q+Q}$ for $q=1,\ldots,Q$.
Then we have 
\[\begin{split}
&\begin{pmatrix}
\hat{\btau}_1(\bS,\bZ^{(1)})\\ 
\hat{\btau}_{\bx,1}(\bS,\bZ^{(1)})
\\
\hat{\btau}_2(\bS,\bZ^{(2)})\\ 
\hat{\btau}_{\bx,2}(\bS,\bZ^{(2)})
\end{pmatrix}\\ 
=& \begin{pmatrix}
\sum_{q=1}^Q \begin{pmatrix}
\bA_q & \bZero_{F \times p}\\ 
\bZero_{Fp \times 1} & \bA_q \otimes \bI_p
\end{pmatrix} \cdot 
\tilde{n}_{q}^{-1}\sum_{i:q_i=q} \bU_i(q)\\
\sum_{q=1+Q}^{2Q} \begin{pmatrix}
\bA_q & \bZero_{F \times p}\\ 
\bZero_{Fp \times 1} & \bA_q \otimes \bI_p
\end{pmatrix} \cdot 
\tilde{n}_{q}^{-1}\sum_{i:q_i=q} \bU_i(q)
\end{pmatrix}\\
=& \sum_{q=1}^{Q}
\begin{pmatrix}
\bA_q & \bZero_{F \times p}\\ 
\bZero_{Fp \times 1} & \bA_q \otimes \bI_p\\
\bZero_{F \times 1} & \bZero_{F \times p}\\ 
\bZero_{Fp \times 1} & \bZero_{Fp \times p}
\end{pmatrix}\cdot 
\tilde{n}_{q}^{-1}\sum_{i:q_i=q} \bU_i(q)+
\sum_{q=1+Q}^{2Q}
\begin{pmatrix}
\bZero_{F \times 1} & \bZero_{F \times p}\\ 
\bZero_{Fp \times 1} & \bZero_{Fp \times p}\\
\bA_q & \bZero_{F \times p}\\ 
\bZero_{Fp \times 1} & \bA_q \otimes \bI_p
\end{pmatrix}\cdot 
\tilde{n}_{q}^{-1}\sum_{i:q_i=q} \bU_i(q)\\
=& \sum_{q=1}^{2Q} \tilde\bA_q \hat{\bar{\bU}}(q)
= \hat{\btau}(\tilde\bA),
\end{split}\]
where $\hat{\bar{\bU}}(q) = \tilde{n}_q^{-1}\sum_{i:q_i=q} \bU_i(q)$, and 
for $q=1,\ldots,Q$,
\[\begin{split}
\tilde\bA_q = \begin{pmatrix}
\bA_q & \bZero_{F \times p}\\ 
\bZero_{Fp \times 1} & \bA_q \otimes \bI_p\\
\bZero_{F \times 1} & \bZero_{F \times p}\\ 
\bZero_{Fp \times 1} & \bZero_{Fp \times p}
\end{pmatrix},\ 
\tilde\bA_{q+Q} = 
\begin{pmatrix}
\bZero_{F \times 1} & \bZero_{F \times p}\\ 
\bZero_{Fp \times 1} & \bZero_{Fp \times p}\\
\bA_q & \bZero_{F \times p}\\ 
\bZero_{Fp \times 1} & \bA_q \otimes \bI_p
\end{pmatrix}.
\end{split}\]
Theorem 3 in \cite{li2017general} shows that under the S-CRFE, 
the mean vector of $\hat{\btau}(\tilde\bA)$ is 
\[\begin{split}
\btau(\tilde\bA) = \sum_{q=1}^{2Q} \tilde\bA_q \bar{\bU}(q) = 
\begin{pmatrix}
\btau \\
\bZero\\
\btau \\
\bZero
\end{pmatrix}.
\end{split}\]
By definition,
\[\begin{split}
\btau_i(\tilde\bA)&
= \sum_{q=1}^{2Q} 
\tilde{\bA}_q \bU_i(q)
= \begin{pmatrix}
\btau_i\\
\bZero\\
\btau_i\\
\bZero
\end{pmatrix}.
\end{split}\]
For $q = 1,\ldots,Q$,
the finite population covariance of $\bU_i(q)$ under group $q$ is 
\[\begin{split}
\bS^2_q = \frac{1}{n-1}\sum_{i=1}^n (\bU_i(q) - \bar{\bU}(q))(\bU_i(q) - \bar{\bU}(q))'
= \begin{pmatrix}
S_{qq} & \bS_{q,\bx}\\
\bS_{\bx,q} & \bS_{\bx\bx}
\end{pmatrix},
\end{split}\]
and $\bS^2_{q+Q} = \bS^2_q$.
The finite population covariance of the individual causal effects is 
\[\begin{split}
\bS^2_{\btau(\tilde\bA)} &= 
\frac{1}{n-1}\sum_{i=1}^n (\btau_i(\tilde\bA)-\btau(\tilde\bA))(\btau_i(\tilde\bA)-\btau(\tilde\bA))'\\
&= \begin{pmatrix}
\bS_{\btau\btau} & \bZero_{F\times Fp}
& \bS_{\btau\btau} & \bZero_{F\times Fp}\\
\bZero_{Fp\times F} & \bZero_{Fp\times Fp} & 
\bZero_{Fp\times F} & \bZero_{Fp\times Fp}\\
\bS_{\btau\btau} & \bZero_{F\times Fp}
& \bS_{\btau\btau} & \bZero_{F\times Fp}\\
\bZero_{Fp\times F} & \bZero_{Fp\times Fp}
& \bZero_{Fp\times F} & \bZero_{Fp\times Fp}
\end{pmatrix}.
\end{split}\]
Theorem 3 in \cite{li2017general} shows that under the S-CRFE, 
the covariance matrix of $\hat{\btau}(\tilde\bA)$ is
\begin{equation}\label{eq:Vtilde}
\cov\left(\hat{\btau}(\tilde\bA)\right) = 
\sum_{q=1}^{2Q} \tilde{n}_q^{-1} \tilde\bA_q \bS^2_q \tilde\bA_q' - n^{-1}\bS^2_{\btau(\tilde\bA)}
= \tilde\bV = 
\begin{pmatrix}
\tilde\bV_{11} & \tilde\bV_{12} & \tilde\bV_{13} & \tilde\bV_{14}\\
\tilde\bV_{21} & \tilde\bV_{22} & \tilde\bV_{23} & \tilde\bV_{24}\\
\tilde\bV_{31} & \tilde\bV_{32} & \tilde\bV_{33} & \tilde\bV_{34}\\
\tilde\bV_{41} & \tilde\bV_{42} & \tilde\bV_{43} & \tilde\bV_{44}
\end{pmatrix},
\end{equation}
where 
\[\begin{split}
\tilde\bV_{11} &= \sum_{q=1}^Q \frac{\bA_q\bA_q'S_{qq}}{\tilde{n}_q}-\frac{\bS_{\btau\btau}}{n},\  
\tilde\bV_{12} = \sum_{q=1}^Q \frac{(\bA_q\bA_q')\otimes \bS_{q,\bx}}{\tilde{n}_q},\\
\tilde\bV_{21} &= \sum_{q=1}^Q \frac{(\bA_q\bA_q')\otimes \bS_{\bx,q}}{\tilde{n}_q},\
\tilde\bV_{22} = \sum_{q=1}^Q \frac{(\bA_q\bA_q')\otimes \bS_{\bx\bx}}{\tilde{n}_q},\\
\tilde\bV_{13} &= \tilde\bV_{31}= -\frac{\bS_{\btau\btau}}{n},\
\tilde\bV_{14} =\tilde\bV_{32}= \bZero,\\
\tilde\bV_{23} &= \tilde\bV_{41} = \bZero,\
\tilde\bV_{24} = \tilde\bV_{42} = \bZero,\\
\tilde\bV_{33} &= \sum_{q=1}^Q \frac{\bA_q\bA_q'S_{qq}}{\tilde{n}_{q+Q}}-\frac{\bS_{\btau\btau}}{n},\
\tilde\bV_{34} = \sum_{q=1}^Q \frac{(\bA_q\bA_q')\otimes \bS_{q,\bx}}{\tilde{n}_{q+Q}},\\
\tilde\bV_{43} &= \sum_{q=1}^Q \frac{(\bA_q\bA_q')\otimes \bS_{\bx,q}}{\tilde{n}_{q+Q}},\
\tilde\bV_{44} = \sum_{q=1}^Q \frac{(\bA_q\bA_q')\otimes \bS_{\bx\bx}}{\tilde{n}_{q+Q}}.
\end{split}\]
Then the corresponding covariance matrix of $$\big(\sqrt{n^{(1)}}(\hat{\btau}_1(\bS,\bZ^{(1)})-\btau)',
\sqrt{n^{(1)}}\hat{\btau}_{\bx,1}'(\bS,\bZ^{(1)}),
\sqrt{n^{(2)}}(\hat{\btau}_2(\bS,\bZ^{(2)})-\btau)',
\sqrt{n^{(2)}}\hat{\btau}_{\bx,2}'(\bS,\bZ^{(2)})\big)'$$ is trivial.
\end{proof}

\subsection{Proof of Theorem \ref{thm:TSCRFE_asymp}}
\begin{proof}
Given the $\tilde{\bA}_q$, $\bU_i(q)$ for $i=1,\ldots,n$ and $q=1,\ldots,2Q$ defined in the proof of Theorem \ref{thm:DistTSMCRE},
for $1\leq k\leq 2F(p+1)$, 
define the maximum squared distance of the $k$th coordinate of the $\tilde\bA_q \bU_i(q)$’s from their population mean, the finite population variance of the $k$th coordinate of the $\tilde\bA_q \bU_i(q)$’s, and the finite population variance of the $k$th coordinate of the $\btau_i(\tilde\bA)$'s as 
\[\begin{split}
m_q(k) &= \max_{1\leq i\leq n} \left[\tilde\bA_q\bU_i(q)-\tilde\bA_q\bar{\bU}(q)\right]^2_{(k)},\\
v_q(k) &= \frac{1}{n-1} \sum_{i=1}^n \left[\tilde\bA_q\bU_i(q)-\tilde\bA_q\bar{\bU}(q)\right]^2_{(k)},\\
v_{\btau}(k) &= \frac{1}{n-1} \sum_{i=1}^n \left[\btau_i(\tilde\bA)-\btau(\tilde\bA)\right]^2_{(k)},
\end{split}\]
respectively, where $[\bu]_{(k)}$ denotes the $k$th element of vector $\bu$.
Then, for $\tilde\bV$ defined in Eq.~\eqref{eq:Vtilde}, 
\[\begin{split}
\sum_{q=1}^{2Q} \tilde{n}_q^{-1} v_q(k) - n^{-1} v_{\btau}(k) = \tilde\bV_{kk}
= \begin{cases}
\bV_{\TSCRFE,kk}/n^{(1)}, & 1\leq k\leq F(p+1)\\
\bV_{\TSCRFE,kk}/n^{(2)}, & F(p+1)+1\leq k\leq 2F(p+1)
\end{cases},
\end{split}\]
where $\tilde\bV_{kk}$ and ${\bV}_{\TSCRFE,kk}$ are the 
$k$-th diagonal element of $\tilde\bV$ and ${\bV}_\TSCRFE$, respectively.
For any $1 \leq q \leq 2Q$ and $1 \leq k \leq 2F(p+1)$, define
\[\begin{split}
G(k,q) = \frac{1}{\tilde{n}_q^2} \frac{m_q(k)}{\sum_{r=1}^{2Q} \tilde{n}_r^{-1} v_r(k) - n^{-1} v_{\btau}(k)}
= \frac{m_q(k)}{\tilde{n}_q^2 \tilde\bV_{kk}}.
\end{split}\]
Under Assumption \ref{assump:regularity_fac_full}, $\bV_\TSCRFE$ must have a limiting value $\bV_{\TSCRFE,\infty}$ as $n \rightarrow \infty$.\\
\hspace{\fill}\\
(1)
If all the diagonal elements of $\bV_{\TSCRFE,\infty}$ are positive, then there exists some $a$ and $N_a > 0$ s.t.
for $n > N_a$, we have $\bV_{\TSCRFE,kk} \geq a$ for all $k$, and for $q=1,\ldots,Q$,
\[\begin{split}
G(k,q) &= \frac{m_q(k)}{\tilde{n}_q^2 \tilde\bV_{kk}} 
= \begin{cases}
\frac{m_q(k)}{n^{(1)}r_q^2 {\bV}_{\TSCRFE,kk}}, & 1\leq k\leq F(p+1)\\
0, & F(p+1)+1\leq k\leq 2F(p+1)
\end{cases},\\
G(k,q+Q) &= \frac{m_{q+Q}(k)}{\tilde{n}_{q+Q}^2 \tilde\bV_{kk}} 
= \begin{cases}
0, & 1\leq k\leq F(p+1)\\
\frac{m_{q+Q}(k)}{n^{(2)}r_q^2 {\bV}_{\TSCRFE,kk}}, & F(p+1)+1\leq k\leq 2F(p+1)
\end{cases}.
\end{split}\]
Under Assumption \ref{assump:regularity_fac_full} (1) and (4), we have 
\[\begin{split}
\max_{1\leq i\leq n} |Y_i(q)-\bar{Y}(q)|^2/n^{(2)} &= \frac{\rho_n}{1-\rho_n}\cdot \max_{1\leq i\leq n} |Y_i(q)-\bar{Y}(q)|^2/n^{(1)}
\rightarrow 0,\\
\max_{1\leq i\leq n} ||\bX_i-\bar{\bX}||_2^2/n^{(2)} &= \frac{\rho_n}{1-\rho_n}\cdot \max_{1\leq i\leq n} ||\bX_i-\bar{\bX}||_2^2/n^{(1)} \rightarrow 0.
\end{split}\]
Therefore, under Assumption \ref{assump:regularity_fac_full},
$G(k,q)$ converge to 0 for all $1\leq k\leq 2F(p+1)$ and $1\leq q \leq 2Q$,
and the correlation matrix of 
$\big(\sqrt{n^{(1)}}(\hat{\btau}_1(\bS,\bZ^{(1)})-\btau)',
\sqrt{n^{(1)}}\hat{\btau}_{\bx,1}'(\bS,\bZ^{(1)}),
\sqrt{n^{(2)}}(\hat{\btau}_2(\bS,\bZ^{(2)})-\btau)',
\sqrt{n^{(2)}}\hat{\btau}_{\bx,2}'(\bS,\bZ^{(2)})\big)'$ must converge to
$\diag(\bV_{\TSCRFE,\infty})^{-1/2}\bV_{\TSCRFE,\infty}\diag(\bV_{\TSCRFE,\infty})^{-1/2}$,
where $\diag(\bv)$ denotes a diagonal matrix whose diagonal elements are the same as that of the square matrix $\bv$. 
Therefore, Theorem 4 in \cite{li2017general} implies that 
\[
\begin{split}
&\diag(\bV_\TSCRFE)^{-1/2}
\big(\sqrt{n^{(1)}}(\hat{\btau}_1(\bS,\bZ^{(1)})-\btau)',
\sqrt{n^{(1)}}\hat{\btau}_{\bx,1}'(\bS,\bZ^{(1)}),\\
&\quad \quad \quad \quad \quad \quad \quad \quad
\sqrt{n^{(2)}}(\hat{\btau}_2(\bS,\bZ^{(2)})-\btau)',
\sqrt{n^{(2)}}\hat{\btau}_{\bx,2}'(\bS,\bZ^{(2)})\big)'\\
\xrightarrow{d} \ &\cN\!\left(\bZero,\diag(\bV_{\TSCRFE,\infty})^{-1/2}\bV_{\TSCRFE,\infty}\diag(\bV_{\TSCRFE,\infty})^{-1/2}\right).
\end{split}
\]
By Slutsky's theorem, we have
\[\begin{split}
&\big(\sqrt{n^{(1)}}(\hat{\btau}_1(\bS,\bZ^{(1)})-\btau)',
\sqrt{n^{(1)}}\hat{\btau}_{\bx,1}'(\bS,\bZ^{(1)}),
\sqrt{n^{(2)}}(\hat{\btau}_2(\bS,\bZ^{(2)})-\btau)',
\sqrt{n^{(2)}}\hat{\btau}_{\bx,2}'(\bS,\bZ^{(2)})\big)'\\
\xrightarrow{d} &\ \ \ \cN(\bZero,\bV_{\TSCRFE,\infty}).
\end{split}\]

\hspace{\fill}\\
(2) If some of the diagonal elements of $\bV_{\TSCRFE,\infty}$ are zero,
let $$\bm\eta = \big(\sqrt{n^{(1)}}(\hat{\btau}_1(\bS,\bZ^{(1)})-\btau)',
\sqrt{n^{(1)}}\hat{\btau}_{\bx,1}'(\bS,\bZ^{(1)}),
\sqrt{n^{(2)}}(\hat{\btau}_2(\bS,\bZ^{(2)})-\btau)',
\sqrt{n^{(2)}}\hat{\btau}_{\bx,2}'(\bS,\bZ^{(2)})\big)',$$
and $\cD$  be the set of indices corresponding to positive diagonal elements of $\bV_{\TSCRFE,\infty}$,
i.e., $\cD = \{k: \bV_{\TSCRFE,\infty,kk} > 0\}$.
Let $\bm\eta_{\cD}$ and $\bm\eta_{\cD^c}$ be subvectors of $\bm\eta$ with indices in $\cD$ and $\cD^c$, respectively.
Let $\bV_{\TSCRFE,\infty,\cD\cD}$ and $\bV_{\TSCRFE,\infty,\cD^c\cD} = \bV_{\TSCRFE,\infty,\cD\cD^c}'$ be submatrices of $\bV_{\TSCRFE,\infty}$ with indices in $\cD \times \cD$ and $\cD^c \times \cD$, respectively.
By the Cauchy–Schwarz inequality, for any $k$ and $j$,
\[\begin{split}
|\bV_{\TSCRFE,kj}| = |\cov(\eta_k,\eta_j)| \leq \sqrt{\bV_{\TSCRFE,kk}\cdot \bV_{\TSCRFE,jj}}.
\end{split}\] 
Hence, $\bV_{\TSCRFE,\infty,kj} = 0$ for $k \in \cD^c$ and any $j$.
Therefore, any elements in submatrices $\bV_{\TSCRFE,\infty,\cD^c\cD^c}$ and 
$\bV_{\TSCRFE,\infty,\cD^c\cD} = \bV_{\TSCRFE,\infty,\cD\cD^c}'$ are zero.
Apply the proof in (1) implies that 
\[\begin{split}
\bm\eta_{\cD} \xrightarrow{d}\cN(\bZero,\bV_{\TSCRFE,\infty,\cD\cD}).
\end{split}\]
For $k \in \cD^c$,
$\bbV(\eta_k) = \bV_{\TSCRFE,kk} \rightarrow \bV_{\TSCRFE,\infty,kk} = 0$. 
Chebyshev's inequality implies that for any $a > 0$,
\[\begin{split}
\bbP(|\eta_k| \geq a) \leq \frac{\bV_{\TSCRFE,kk}}{a^2} \rightarrow 0.
\end{split}\]
Hence, $\eta_k = o_p(1)$.
Therefore, $\bm\eta_{\cD^c} \xrightarrow{p} \bZero$, and thus,
\[\begin{split}
\begin{pmatrix}
    \bm\eta_{\cD} \\ \bm\eta_{\cD^c}
\end{pmatrix} \xrightarrow{d} \cN\left(\bZero,\begin{pmatrix}
    \bV_{\TSCRFE,\infty,\cD\cD} & \bZero\\
    \bZero & \bZero
\end{pmatrix}\right)
\sim \cN\left(\bZero,\begin{pmatrix}
    \bV_{\TSCRFE,\infty,\cD\cD} & \bV_{\TSCRFE,\infty,\cD\cD^c}\\
    \bV_{\TSCRFE,\infty,\cD^c\cD} & \bV_{\TSCRFE,\infty,\cD^c\cD^c}
\end{pmatrix}\right).
\end{split}\]
Equivalently, $\bm\eta \xrightarrow{d} \cN(\bZero,\bV_{\TSCRFE,\infty})$.

\hspace{\fill}\\
Above all, we finished the proof.
\end{proof}

\subsection{Proof of Proposition \ref{pro:consistent_fac}}
\begin{proof}
Before proving Proposition \ref{pro:consistent_fac}, we first prove the following lemmas.
\begin{lemma}\label{lem:joint}
For random vectors $\bX_n,\bY_n$ and $n\rightarrow\infty$,
if $\bX_n \overset{d}{\rightarrow} \bX$, and $\bY_n \overset{p}{\rightarrow} \bc$, where $\bc$ is a constant vector, then $(\bX_n,\bY_n) \overset{d}{\rightarrow} (\bX,\bc)$.
\end{lemma}

\begin{lemma}\label{lem:prob1}
Under Assumption \ref{assump:regularity_fac_full} and as $n\rightarrow \infty$, for random splitting $\bS$ and completely random assignment $\bZ^{(1)}$, we have 
\[\begin{split}
&\bbP(\phi_M(\sqrt{n^{(1)}}\hat{\btau}_{\bx,1}(\bS,\bZ^{(1)}),\bV_{\bx\bx,1}(\bS))=1 \mid \bS) 
= \bbP(\phi_M(\bD_{1,\infty},\bV_{\bx\bx,\infty}) = 1) + o_p(1),\\
\end{split}\]
where $\bD_{1,\infty}\sim \cN(\bZero,\bV_{\bx\bx,\infty})$, and $\bV_{\bx\bx,\infty} = \sum_{q=1}^Q \frac{(\bA_q\bA_q')\otimes \bS_{\bx\bx,\infty}}{r_q}.$
\end{lemma}

\begin{lemma}\label{lem:equiv}
Under Assumption \ref{assump:regularity_fac_full}, let $\cD=\{0,1\}^n \times \{1,\ldots,Q\}^{n^{(1)}}$. 
As $n \rightarrow \infty$, 
for the random splitting indicator $\bS$, 
treatment assignment $\bZ^{(1)}$ in S-CRFE, and 
treatment assignment $\bW^{(1)}$ in DA-ReO$_\F$ or DA-ReO$_\FE$,
we have 
\[\begin{split}
\sup_{\cA \subset \cD} \big|\bbP((\bS,\bW^{(1)}) \in \cA)-
\bbP((\bS,\bZ^{(1)}) \in \cA \mid 
\phi_M(\sqrt{n^{(1)}}\hat{\btau}_{\bx,1}(\bS,\bZ^{(1)}),\bV_{\bx\bx,1}(\bS))=1)\big|
\rightarrow 0.
\end{split}\]
\end{lemma}

\begin{lemma}\label{lem:inverse_lemma}
Suppose matrix $\bA_n \rightarrow \bA_{\infty}$, where $\bA_{\infty}$ is nonsingular, $\bB_n = o_p(1)$, then $(\bA_n + \bB_n)^{-1} = \bA_n^{-1} + o_p(1)$,
where $(\bA_n + \bB_n)^{-1}$ (or $\bA_n^{-1}$) is $\bZero$ when $\bA_n + \bB_n$ (or $\bA_n$) is singular.
\end{lemma}

\begin{proof}[Proof of Lemma \ref{lem:joint}]
For any constant vectors $\bm a$ and $\bm b$, since $\bm a'\bX_n \xrightarrow{d} \bm a' \bX$, and $\bm b'\bY_n \xrightarrow{p} \bm b' \bc$, it follows from Slutsky’s theorem that $\bm a'\bX_n + \bm b'\bY_n \xrightarrow{d} \bm a' \bX + \bm b' \bc$.
Therefore, by the Cram\'{e}r--Wold theorem, we have $(\bX_n,\bY_n) \overset{d}{\rightarrow} (\bX,\bc)$.
\end{proof}

\begin{proof}[Proof of Lemma \ref{lem:prob1}]
Let 
\[\begin{split}
\psi_{n,1}(\bS) &\equiv \bbP(\phi_M(\sqrt{n^{(1)}}\hat{\btau}_{\bx,1}(\bS,\bZ^{(1)}),\bV_{\bx\bx,1}(\bS))=1 \mid \bS).
\end{split}\]
By Theorem 2.3.2 in \cite{durrett2019probability}, it suffices to show that for any subsequence $n(m) \rightarrow \infty$,
there exists a further subsequence $n(m_k) \rightarrow \infty$ s.t.
\begin{align}
\label{eq:res1}
&\psi_{n(m_k),1}(\bS) - \bbP(\phi_M(\bD_{1,\infty},\bV_{\bx\bx,\infty}) = 1) \xrightarrow{a.s.} 0.
\end{align}
By definition, 
\[\begin{split}
\bV_{\bx\bx,1}(\bS) =\sum_{q=1}^Q \frac{(\bA_q\bA_q')\otimes \bms_{\bx\bx}(\cI_1(\bS))}{r_{q}},
\end{split}\]
and according to Lemma B4 in \cite{yang2023rejective}, under Assumption \ref{assump:regularity_fac_full}, we have 
$\bms_{\bx\bx}(\cI_1(\bS)) - \bS_{\bx\bx} = o_p(1)$.
Therefore, $\bV_{\bx\bx,1}(\bS) - \bV_{\bx\bx,\infty} = o_p(1)$.
Then Theorem 2.3.2 in \cite{durrett2019probability} implies that for any subsequence $n(m) \rightarrow \infty$,
there exists a further subsequence $n(m_k) \rightarrow \infty$ s.t. $\bV_{\bx\bx,1}(\bS) \xrightarrow{a.s.} \bV_{\bx\bx,\infty}$.
Moreover, for
$\bar{\bX}_1(\bS) = (n^{(1)})^{-1}\sum_{i=1}^n S_i \bX_i$, 
we have 
\[\begin{split}
||\bX_i - \bar{\bX}_1(\bS)||_2 &\leq ||\bX_i - \bar{\bX}||_2 + ||\bar{\bX} - \bar{\bX}_1(\bS)||_2
\leq 2\max_{1\leq j\leq n} ||\bX_j - \bar{\bX}||_2.
\end{split}\]
Under Assumption \ref{assump:regularity_fac_full}, we have 
\[\begin{split}
&\frac{1}{n^{(1)}}\max_{i:S_i=1}||\bX_i - \bar{\bX}_1(\bS)||_2^2 \leq 
\frac{4}{n^{(1)}}\max_{1\leq j\leq n} ||\bX_j - \bar{\bX}||_2^2 = o(1).
\end{split}\]
According to the finite population CLT for CRFE, for the subsequence $n(m_k)\rightarrow\infty$,
when the sequence of $\bms_{n(m_k)}$ satisfies that 
$\bV_{\bx\bx,1}(\bms_{n(m_k)}) - \bV_{\bx\bx,\infty} = o(1)$,
we have 
\[\begin{split}
(\sqrt{n^{(1)}}\hat{\btau}_{\bx,1}(\bS,\bZ^{(1)})\mid \bS = \bms_{n(m_k)}) \xrightarrow{d} 
\bD_{1,\infty} \sim \cN(\bZero,\bV_{\bx\bx,\infty}).
\end{split}\]
Lemma \ref{lem:joint} and the continuous mapping theorem implies that 
\[\begin{split}
(\phi_M(\sqrt{n^{(1)}}\hat{\btau}_{\bx,1}(\bS,\bZ^{(1)}),\bV_{\bx\bx,1}(\bS))\mid \bS = \bms_{n(m_k)})
\xrightarrow{d} 
\phi_M(\bD_{1,\infty},\bV_{\bx\bx,\infty}),
\end{split}\]
and thus, 
\[\begin{split}
\psi_{n(m_k),1}(\bms_{n(m_k)}) &= \bbP(\phi_M(\sqrt{n^{(1)}}\hat{\btau}_{\bx,1}(\bS,\bZ^{(1)}),\bV_{\bx\bx,1}(\bS) = 1\mid \bS = \bms_{n(m_k)})\\ &\rightarrow 
\bbP(\phi_M(\bD_{1,\infty},\bV_{\bx\bx,\infty})=1).
\end{split}\]
Since $\bV_{\bx\bx,1}(\bS) \rightarrow \bV_{\bx\bx,\infty}$ a.s. for $n({m_k})\rightarrow\infty$, we have Eq.~\eqref{eq:res1} holds.
Then we complete the proof.
\end{proof}

\begin{proof}[Proof of Lemma \ref{lem:equiv}]
By Theorem \ref{thm:TSCRFE_asymp}, under S-CRFE and Assumption \ref{assump:regularity_fac_full},
\[\begin{split}
&\big(\sqrt{n^{(1)}}(\hat{\btau}_1(\bS,\bZ^{(1)})-\btau)',
\sqrt{n^{(1)}}\hat{\btau}_{\bx,1}'(\bS,\bZ^{(1)}),
\sqrt{n^{(2)}}(\hat{\btau}_2(\bS,\bZ^{(2)})-\btau)',
\sqrt{n^{(2)}}\hat{\btau}_{\bx,2}'(\bS,\bZ^{(2)})\big)'\\
&
\overset{\cdot}{\sim}
\cN(\bZero, \bV_{\TSCRFE}),
\end{split}\]
Suppose $(\bC_1',\bD_1',\bC_2',\bD_2')'
\sim \cN(\bZero,\bV_\TSCRFE)$.
According to Lemma B4 in \cite{yang2023rejective}, under S-CRFE, we have 
$\bV_{\bx\bx,1}(\bS) \xrightarrow{p} \bV_{\bx\bx,\infty}$.
Then Lemma \ref{lem:joint} and the continuous mapping theorem imply that 
\[\begin{split}
&\phi_M(\sqrt{n^{(1)}}\hat{\btau}_{\bx,1}(\bS,\bZ^{(1)}),\bV_{\bx\bx,1}(\bS))
\xrightarrow{d} 
\phi_M(\bD_{1,\infty},\bV_{\bx\bx,\infty}),
\end{split}\]
where $\bD_{1,\infty}\sim \cN\left(\bZero,\bV_{\bx\bx,\infty}\right)$.
Therefore,
\[\begin{split}
\lim_{n\rightarrow\infty}
\bbP(\phi_M(\sqrt{n^{(1)}}\hat{\btau}_{\bx,1}(\bS,\bZ^{(1)}),\bV_{\bx\bx,1}(\bS))=1) = \bbP(\phi_M(\bD_{1,\infty},\bV_{\bx\bx,\infty})=1) = \alpha.
\end{split}\]
Then for any $0 < a < \alpha$, there exists some $N_a$ s.t. 
for any $n > N_a$, we have 
\begin{equation}\label{eq:probb21}
\bbP(\phi_M(\sqrt{n^{(1)}}\hat{\btau}_{\bx,1}(\bS,\bZ^{(1)}),\bV_{\bx\bx,1}(\bS))=1)
\in [\alpha-a,\alpha+a].
\end{equation}
Moreover,
Lemma \ref{lem:prob1} implies that 
\[\begin{split}
&\bbP(\phi_M(\sqrt{n^{(1)}}\hat{\btau}_{\bx,1}(\bS,\bZ^{(1)}),\bV_{\bx\bx,1}(\bS))=1 \mid \bS) 
= \bbP(\phi_M(\bD_{1,\infty},\bV_{\bx\bx,\infty}) = 1) + o_p(1) = \alpha + o_p(1).
\end{split}\]
Therefore, 
for any $a > 0$, define 
\[\begin{split}
\cS_{a,1} &= \left\{\bms:
\left|\bbP(\phi_M(\sqrt{n^{(1)}}\hat{\btau}_{\bx,1}(\bms,\bZ^{(1)}),\bV_{\bx\bx,1}(\bms))=1)-\bbP(\phi_M(\bD_{1,\infty},\bV_{\bx\bx,\infty}) = 1)\right| \leq a\right\},
\end{split}\]
there exists $N_a$ s.t. for $n > N_a$,
$\bbP(\bS \not\in \cS_{a,1}) \leq a.$
Moreover, for $\bms \in \cS_{a,1}$ and $a < \alpha$,
\begin{equation}\label{eq:probb22}
0 < \bbP(\phi_M(\sqrt{n^{(1)}}\hat{\btau}_{\bx,1}(\bms,\bZ^{(1)}),\bV_{\bx\bx,1}(\bms))=1) \in [\alpha-a,\alpha+a].
\end{equation}

By definition, for $n > N_a$, and any $\cA \subset \cD$,
\[\begin{split}
\bbP((\bS,\bW^{(1)}) \in \cA)
&\leq 
\bbP((\bS,\bW^{(1)}) \in \cA,
\bS \in \cS_{a,1}) + \bbP(\bS \not\in \cS_{a,1})\\
&\leq \sum_{\bms \in \cS_{a,1}} 
\bbP((\bS,\bW^{(1)}) \in \cA,
\bS=\bms) + a,
\end{split}\]
and 
\[\begin{split}
\bbP((\bS,\bW^{(1)}) \in \cA)
&\geq  \bbP((\bS,\bW^{(1)}) \in \cA,\bS \in \cS_{a,1})
= \sum_{\bms \in \cS_{a,1}}
\bbP((\bS,\bW^{(1)}) \in \cA,
\bS=\bms).
\end{split}\]
Since 
\[\begin{split}
&\bbP((\bS,\bW^{(1)}) \in \cA,
\bS=\bms)
= \bbP((\bS,\bW^{(1)}) \in \cA\mid 
\bS=\bms)\bbP(\bS=\bms)\\
=& \bbP((\bms,\bW^{(1)}) \in \cA \mid \bS=\bms)\bbP(\bS=\bms)\\
=& \bbP((\bms,\bZ^{(1)}) \in \cA \mid \phi_M(\sqrt{n^{(1)}}\hat{\btau}_{\bx,1}(\bms,\bZ^{(1)}),\bV_{\bx\bx,1}(\bms))=1)\bbP(\bS=\bms),
\end{split}\]
then we have 
\[\begin{split}
&\sum_{\bms \in \cS_{a,1}}
\bbP((\bS,\bW^{(1)}) \in \cA,
\bS=\bms)\\
=& \sum_{\bms \in \cS_{a,1}} 
\bbP((\bms,\bZ^{(1)}) \in \cA \mid \phi_M(\sqrt{n^{(1)}}\hat{\btau}_{\bx,1}(\bms,\bZ^{(1)}),\bV_{\bx\bx,1}(\bms))=1)\bbP(\bS=\bms).
\end{split}\]
By definition, for $\bms \in \cS_{a,1}$,
\[\begin{split}
& \bbP((\bS,\bZ^{(1)}) \in \cA,
\bS = \bms\mid \phi_M(\sqrt{n^{(1)}}\hat{\btau}_{\bx,1}(\bS,\bZ^{(1)}),\bV_{\bx\bx,1}(\bS))=1)\\
=& \bbP((\bS,\bZ^{(1)}) \in \cA
\mid \bS = \bms, 
\phi_M(\sqrt{n^{(1)}}\hat{\btau}_{\bx,1}(\bS,\bZ^{(1)}),\bV_{\bx\bx,1}(\bS))=1)\cdot \\
& \bbP(\bS = \bms\mid 
\phi_M(\sqrt{n^{(1)}}\hat{\btau}_{\bx,1}(\bS,\bZ^{(1)}),\bV_{\bx\bx,1}(\bS))=1)\\
=& \bbP((\bms,\bZ^{(1)}) \in \cA \mid \phi_M(\sqrt{n^{(1)}}\hat{\btau}_{\bx,1}(\bms,\bZ^{(1)}),\bV_{\bx\bx,1}(\bms))=1)\cdot\\& \bbP(\bS = \bms\mid 
\phi_M(\sqrt{n^{(1)}}\hat{\btau}_{\bx,1}(\bS,\bZ^{(1)}),\bV_{\bx\bx,1}(\bS))=1).
\end{split}\]
Thus,
\[\begin{split}
&\sum_{\bms \in \cS_{a,1}}
\bbP((\bS,\bW^{(1)}) \in \cA,
\bS=\bms)\\
=& \sum_{\bms \in \cS_{a,1}} 
\bbP((\bms,\bZ^{(1)}) \in \cA \mid \phi_M(\sqrt{n^{(1)}}\hat{\btau}_{\bx,1}(\bms,\bZ^{(1)}),\bV_{\bx\bx,1}(\bms))=1)\bbP(\bS=\bms)\\
=& \sum_{\bms \in \cS_{a,1}} 
\frac{\bbP((\bS,\bZ^{(1)}) \in \cA,
\bS = \bms\mid \phi_M(\sqrt{n^{(1)}}\hat{\btau}_{\bx,1}(\bS,\bZ^{(1)}),\bV_{\bx\bx,1}(\bS))=1)\bbP(\bS=\bms)}{\bbP(\bS = \bms\mid 
\phi_M(\sqrt{n^{(1)}}\hat{\btau}_{\bx,1}(\bS,\bZ^{(1)}),\bV_{\bx\bx,1}(\bS))=1)}.
\end{split}\]
Since 
\[\begin{split}
&\bbP(\bS = \bms\mid 
\phi_M(\sqrt{n^{(1)}}\hat{\btau}_{\bx,1}(\bS,\bZ^{(1)}),\bV_{\bx\bx,1}(\bS))=1)\\
=& \frac{\bbP(\bS = \bms, 
\phi_M(\sqrt{n^{(1)}}\hat{\btau}_{\bx,1}(\bS,\bZ^{(1)}),\bV_{\bx\bx,1}(\bS))=1)}{\bbP(\phi_M(\sqrt{n^{(1)}}\hat{\btau}_{\bx,1}(\bS,\bZ^{(1)}),\bV_{\bx\bx,1}(\bS))=1)} \\
=& \frac{\bbP(\bS = \bms)\bbP(\phi_M(\sqrt{n^{(1)}}\hat{\btau}_{\bx,1}(\bms,\bZ^{(1)}),\bV_{\bx\bx,1}(\bms))=1)}{\bbP(\phi_M(\sqrt{n^{(1)}}\hat{\btau}_{\bx,1}(\bS,\bZ^{(1)}),\bV_{\bx\bx,1}(\bS))=1)},
\end{split}\]
we have 
\[\begin{split}
&\sum_{\bms \in \cS_{a,1}}
\bbP((\bS,\bW^{(1)}) \in \cA,
\bS=\bms)\\
=& \sum_{\bms \in \cS_{a,1}} 
\frac{\bbP((\bS,\bZ^{(1)}) \in \cA,
\bS = \bms\mid \phi_M(\sqrt{n^{(1)}}\hat{\btau}_{\bx,1}(\bS,\bZ^{(1)}),\bV_{\bx\bx,1}(\bS))=1)\bbP(\bS=\bms)}{\bbP(\bS = \bms\mid 
\phi_M(\sqrt{n^{(1)}}\hat{\btau}_{\bx,1}(\bS,\bZ^{(1)}),\bV_{\bx\bx,1}(\bS))=1)}\\
=& \sum_{\bms \in \cS_{a,1}} 
\frac{\bbP((\bS,\bZ^{(1)}) \in \cA,
\bS = \bms\mid \phi_M(\sqrt{n^{(1)}}\hat{\btau}_{\bx,1}(\bS,\bZ^{(1)}),\bV_{\bx\bx,1}(\bS))=1)}{\bbP(\phi_M(\sqrt{n^{(1)}}\hat{\btau}_{\bx,1}(\bms,\bZ^{(1)}),\bV_{\bx\bx,1}(\bms))=1)}\cdot \\
&\quad\quad\quad\quad
\bbP(\phi_M(\sqrt{n^{(1)}}\hat{\btau}_{\bx,1}(\bS,\bZ^{(1)}),\bV_{\bx\bx,1}(\bS))=1)
\end{split}\]
For $n > N_a$, by Eq.'s~\eqref{eq:probb21} and  \eqref{eq:probb22},
\[\begin{split}
\frac{\bbP(\phi_M(\sqrt{n^{(1)}}\hat{\btau}_{\bx,1}(\bS,\bZ^{(1)}),\bV_{\bx\bx,1}(\bS))=1)}{\bbP(\phi_M(\sqrt{n^{(1)}}\hat{\btau}_{\bx,1}(\bms,\bZ^{(1)}),\bV_{\bx\bx,1}(\bms))=1)}
\in \left[\frac{
\alpha-a
}{\alpha+a},\frac{
\alpha+a
}{\alpha-a}\right].
\end{split}\]
Since 
\[\begin{split}
&\bbP((\bS,\bZ^{(1)}) \in \cA, \bS \not\in \cS_{a,1}\mid \phi_M(\sqrt{n^{(1)}}\hat{\btau}_{\bx,1}(\bS,\bZ^{(1)}),\bV_{\bx\bx,1}(\bS))=1)\\
\leq & \frac{\bbP(\bS \not\in \cS_{a,1})}{\bbP(\phi_M(\sqrt{n^{(1)}}\hat{\btau}_{\bx,1}(\bS,\bZ^{(1)}),\bV_{\bx\bx,1}(\bS))=1)}
\leq \frac{a}{\alpha-a},
\end{split}\]
we have 
\[\begin{split}
&\bbP((\bS,\bZ^{(1)}) \in \cA, \bS \in \cS_{a,1}\mid \phi_M(\sqrt{n^{(1)}}\hat{\btau}_{\bx,1}(\bS,\bZ^{(1)}),\bV_{\bx\bx,1}(\bS))=1)\\
\geq & \bbP((\bS,\bZ^{(1)}) \in \cA\mid \phi_M(\sqrt{n^{(1)}}\hat{\btau}_{\bx,1}(\bS,\bZ^{(1)}),\bV_{\bx\bx,1}(\bS))=1)
- \frac{a}{\alpha-a}.
\end{split}\]
Hence, 
\[\begin{split}
&\sum_{\bms \in \cS_{a,1}}
\bbP((\bS,\bW^{(1)}) \in \cA,
\bS=\bms)\\
\leq & \frac{
\alpha+a
}{\alpha-a}\cdot 
\sum_{\bms \in \cS_{a,1}}\bbP((\bS,\bZ^{(1)}) \in \cA,
\bS = \bms\mid \phi_M(\sqrt{n^{(1)}}\hat{\btau}_{\bx,1}(\bS,\bZ^{(1)}),\bV_{\bx\bx,1}(\bS))=1)\\
\leq & \frac{
\alpha+a
}{\alpha-a}\cdot 
\bbP((\bS,\bZ^{(1)}) \in \cA\mid \phi_M(\sqrt{n^{(1)}}\hat{\btau}_{\bx,1}(\bS,\bZ^{(1)}),\bV_{\bx\bx,1}(\bS))=1),
\end{split}\]
and 
\[\begin{split}
&\sum_{\bms \in \cS_{a,1}}
\bbP((\bS,\bW^{(1)}) \in \cA,
\bS=\bms)\\
\geq & \frac{
\alpha-a
}{\alpha+a}\cdot 
\sum_{\bms \in \cS_{a,1}}\bbP((\bS,\bZ^{(1)}) \in \cA,
\bS = \bms\mid \phi_M(\sqrt{n^{(1)}}\hat{\btau}_{\bx,1}(\bS,\bZ^{(1)}),\bV_{\bx\bx,1}(\bS))=1)\\
\geq & \frac{
\alpha-a
}{\alpha+a}\cdot 
\bbP((\bS,\bZ^{(1)}) \in \cA\mid \phi_M(\sqrt{n^{(1)}}\hat{\btau}_{\bx,1}(\bS,\bZ^{(1)}),\bV_{\bx\bx,1}(\bS))=1)
- \frac{a}{\alpha+a}.
\end{split}\]
Therefore, for $n > N_a$,
\[\begin{split}
\bbP((\bS,\bW^{(1)}) \in \cA)
&\leq \sum_{\bms \in \cS_{a,1}} 
\bbP((\bS,\bW^{(1)}) \in \cA,
\bS=\bms) + a\\
&\leq \frac{
\alpha+a
}{\alpha-a}\cdot 
\bbP((\bS,\bZ^{(1)}) \in \cA\mid \phi_M(\sqrt{n^{(1)}}\hat{\btau}_{\bx,1}(\bS,\bZ^{(1)}),\bV_{\bx\bx,1}(\bS))=1) + a,
\end{split}\]
and 
\[\begin{split}
\bbP((\bS,\bW^{(1)}) \in \cA)
&\geq  \sum_{\bms \in \cS_{a,1}}
\bbP((\bS,\bW^{(1)}) \in \cA,
\bS=\bms)\\
&\geq \frac{
\alpha-a
}{\alpha+a}\cdot 
\bbP((\bS,\bZ^{(1)}) \in \cA\mid \phi_M(\sqrt{n^{(1)}}\hat{\btau}_{\bx,1}(\bS,\bZ^{(1)}),\bV_{\bx\bx,1}(\bS))=1)
- \frac{a}{\alpha+a}.
\end{split}\]
Then for any $0 < a < \alpha$ and any $\cA \subset \cD$, there exists $N_a > 0$ s.t. for any $n > N_a$, we have 
\[\begin{split}
&\left|\bbP((\bS,\bW^{(1)}) \in \cA) - \bbP((\bS,\bZ^{(1)}) \in \cA\mid \phi_M(\sqrt{n^{(1)}}\hat{\btau}_{\bx,1}(\bS,\bZ^{(1)}),\bV_{\bx\bx,1}(\bS))=1)\right| \\
\leq &\max\left\{\frac{3a}{\alpha+a},\frac{2a}{\alpha-a}+a\right\}.
\end{split}\]
Let $a \to 0$, we have
\[\begin{split}
\sup_{\cA \subset \cD}\left|\bbP((\bS,\bW^{(1)}) \in \cA) - \bbP((\bS,\bZ^{(1)}) \in \cA\mid \phi_M(\sqrt{n^{(1)}}\hat{\btau}_{\bx,1}(\bS,\bZ^{(1)}),\bV_{\bx\bx,1}(\bS))=1)\right| \rightarrow 0.
\end{split}\]
\end{proof}

\begin{proof}[Proof of Lemma \ref{lem:inverse_lemma}]
Since the limit of $\bA_n$, denoted as $\bA_{\infty}$, is nonsingular,
we have 
\[\begin{split}
||\bI-\bA_{\infty}^{-1}\bA_{n}||_2 
= ||\bA_{\infty}^{-1}(\bA_{\infty}-\bA_n)||_2
\leq ||\bA_{\infty}^{-1}||_2 ||\bA_{\infty}-\bA_n||_2 \rightarrow 0,
\end{split}\]
where $||\cdot ||_2$ is the spectral norm.
Hence, there exists some $N_0$ s.t. for $n > N_0$,
$||\bI-\bA_{\infty}^{-1}\bA_{n}||_2 \leq 1/2 < 1$.
According to the Von Neumann Lemma, 
when $||\bI-\bA_{\infty}^{-1}\bA_{n}||_2 < 1$,
$\bA_{\infty}^{-1}\bA_{n}$ is nonsingular, and 
\[\begin{split}
(\bA_{\infty}^{-1}\bA_{n})^{-1} = \sum_{k=0}^\infty (\bI-\bA_{\infty}^{-1}\bA_{n})^k.
\end{split}\]
Denote $(\bA_{\infty}^{-1}\bA_{n})^{-1} = \bC$,
then 
\[\begin{split}
    (\bA_{\infty}^{-1}\bA_{n})\bC = \bC(\bA_{\infty}^{-1}\bA_{n}) = \bI.
\end{split}\]
Hence, $\bA_{\infty}(\bA_{\infty}^{-1}\bA_{n})\bC = \bA_n\bC = \bA_{\infty}$,
i.e., $\bA_n\bC\bA_{\infty}^{-1} = \bI$.
Therefore, 
$\bA_n$ is nonsingular, and
\[\begin{split}
\bA_n^{-1} =\bC\bA_{\infty}^{-1} = \sum_{k=0}^\infty (\bI-\bA_{\infty}^{-1}\bA_n)^k\bA_{\infty}^{-1}.
\end{split}\]
Therefore, for $n > N_0$, $\bA_n$ is nonsingular, and 
\[\begin{split}
||\bA_n^{-1} - \bA_{\infty}^{-1}||_2 &\leq 
\sum_{k=1}^\infty ||\bI-\bA_{\infty}^{-1}\bA_n||_2^k||\bA_{\infty}^{-1}||_2
= \frac{||\bI-\bA_{\infty}^{-1}\bA_n||_2||\bA_{\infty}^{-1}||_2}{1-||\bI-\bA_{\infty}^{-1}\bA_n||_2}
\leq ||\bA_{\infty}^{-1}||_2.
\end{split}\]
Hence, $||\bA_n^{-1}||_2 \leq ||\bA_n^{-1} - \bA_{\infty}^{-1}||_2 + ||\bA_{\infty}^{-1}||_2
\leq 2||\bA_{\infty}^{-1}||_2$.

\hspace{\fill}\\
By definition, $\bA_n + \bB_n = \bA_n(\bI+\bA_n^{-1}\bB_n)$.
According to the Von Neumann Lemma, when $||\bA_n^{-1}\bB_n||_2 < 1$, then
$\bI+\bA_n^{-1}\bB_n$ is nonsingular, and
\[\begin{split}
(\bI + \bA_n^{-1}\bB_n)^{-1} = \sum_{k=0}^\infty (-\bA_n^{-1}\bB_n)^k.
\end{split}\]
Hence, for $n > N_0$ and when $||\bA_n^{-1}\bB_n||_2 < 1$,
$\bA_n + \bB_n$ is nonsingular, and 
\[\begin{split}
(\bA_n + \bB_n)^{-1} - \bA_n^{-1} &= (\bI + \bA_n^{-1}\bB_n)^{-1}\bA_n^{-1} - \bA_n^{-1}
= \sum_{k=1}^\infty (-\bA_n^{-1}\bB_n)^k\bA_n^{-1}.
\end{split}\]
Then for $n > N_0$ and when $||\bA_n^{-1}\bB_n||_2 < 1$,
\[\begin{split}
||(\bA_n + \bB_n)^{-1} - \bA_n^{-1}||_2 &\leq \sum_{k=1}^\infty ||\bA_n^{-1}\bB_n||_2^k \cdot ||\bA_n^{-1}||_2
= \frac{||\bA_n^{-1}\bB_n||_2||\bA_n^{-1}||_2}{1-||\bA_n^{-1}\bB_n||_2}
\\&
\leq \frac{2||\bA_\infty^{-1}||_2||\bA_n^{-1}\bB_n||_2}{1-||\bA_n^{-1}\bB_n||_2}.
\end{split}\]
Since $\bA_n^{-1}\bB_n = o_p(1)$, 
then $||\bA_n^{-1}\bB_n||_2 = o_p(1)$, i.e.,
for any $\delta,a > 0$, 
there exists some $N_{\delta,a}$ s.t. for $n > N_{\delta,a}$,
\[\begin{split}
\bbP\left(||\bA_n^{-1}\bB_n||_2 \geq \frac{\delta}{\delta + 2||\bA_\infty^{-1}||_2}\right) \leq a.
\end{split}\]
Hence, for any $\delta > 0$,
when $n > \{N_0,N_{\delta,a}\}$,
\[\begin{split}
1-a &\leq \bbP\left(||\bA_n^{-1}\bB_n||_2 < \frac{\delta}{\delta + 2||\bA_\infty^{-1}||_2}\right)\\&
\leq \bbP\left((\bA_n + \bB_n)\text{ nonsingular, and }||(\bA_n + \bB_n)^{-1} - \bA_n^{-1}||_2 < \delta \right).
\end{split}\]
Therefore, for $n > \max\{N_{\delta,a},N_0\}$,
\[\begin{split}
\bbP\left((\bA_n + \bB_n)\text{ nonsingular, and }
||(\bA_n + \bB_n)^{-1} - \bA_n^{-1}||_2 \geq \delta \right)
+ \bbP((\bA_n + \bB_n) \text{ singular}) \leq a.
\end{split}\]
When $\bA_n + \bB_n$ is singular,  $(\bA_n + \bB_n)^{-1}$ is defined to be $\bZero$.
Then when $n>\max\{N_0,N_{\delta,a}\}$,
\[\begin{split}
&\bbP\left(
||(\bA_n + \bB_n)^{-1} - \bA_n^{-1}||_2 \geq \delta \right)\\
=& \bbP\left((\bA_n + \bB_n)\text{ nonsingular, and }
||(\bA_n + \bB_n)^{-1} - \bA_n^{-1}||_2 \geq \delta \right)
\\&+
\bbP\left((\bA_n + \bB_n)\text{ singular, and }
||(\bA_n + \bB_n)^{-1} - \bA_n^{-1}||_2 \geq \delta \right)\\
\leq& \bbP\left((\bA_n + \bB_n)\text{ nonsingular, and }
||(\bA_n + \bB_n)^{-1} - \bA_n^{-1}||_2 \geq \delta \right)
+ \bbP((\bA_n + \bB_n) \text{ singular})\\
\leq& a,
\end{split}\]
i.e., $||(\bA_n + \bB_n)^{-1} - \bA_n^{-1}||_2 = o_p(1)$.
Therefore, $(\bA_n + \bB_n)^{-1} = \bA_n^{-1} + o_p(1)$.
\end{proof}

\begin{proof}[Proof of Proposition \ref{pro:consistent_fac}]
    It suffices to show that for $q=1,\ldots,Q$,
$\hat{\bbeta}_q(\bS,\bW^{(1)})-\bbeta_q = o_p(1)$, where 
$\bbeta_q = \bS_{\bx\bx}^{-1}\bS_{\bx,q}$.
Then it suffices to show that 
\[\begin{split}
\bms_{\bx\bx,q}(\bS,\bW^{(1)}) - \bS_{\bx\bx} = o_p(1),\      
\bms_{\bx,q}(\bS,\bW^{(1)}) - \bS_{\bx,q} = o_p(1).
\end{split}\]
According to \cite{cochran1977sampling}, under S-CRFE, since units receiving treatment arm $q$ in the first stage 
is a simple random sample of size $n_{1q}$, then 
$\bms_{\bx\bx,q}(\bS,\bZ^{(1)})$ and $\bms_{\bx,q}(\bS,\bZ^{(1)})$ are unbiased estimators of $\bS_{\bx\bx}$ and $\bS_{\bx,q}$, respectively.
Similar to the proof of Lemma B16 in \cite{yang2023rejective}, under Assumption \ref{assump:regularity_fac_full} and S-CRFE,
for any $(A_i, B_i)$ equal to $(Y_i(q), Y_i(q))$, $(Y_i(q), X_{ik})$, 
or $(X_{ik},X_{il})$, we have
\[\begin{split}
\bbE(s_{AB}(\bS,\bZ^{(1)})) = S_{AB},\ 
\bbV(s_{AB}(\bS,\bZ^{(1)})) = o(1).
\end{split}\]
By definition, 
\[\begin{split}
&\bbE[(s_{AB}(\bS,\bZ^{(1)})-S_{AB})^2 \mid \phi_M(\sqrt{n^{(1)}}\hat{\btau}_{\bx,1}(\bS,\bZ^{(1)}),\bV_{\bx\bx,1}(\bS))=1]\\
\leq & \frac{\bbE[(s_{AB}(\bS,\bZ^{(1)})-S_{AB})^2]}{\bbP(\phi_M(\sqrt{n^{(1)}}\hat{\btau}_{\bx,1}(\bS,\bZ^{(1)}),\bV_{\bx\bx,1}(\bS))=1)}\\
=& \frac{\bbV(s_{AB}(\bS,\bZ^{(1)}))}{\bbP(\phi_M(\sqrt{n^{(1)}}\hat{\btau}_{\bx,1}(\bS,\bZ^{(1)}),\bV_{\bx\bx,1}(\bS))=1)}
\rightarrow 0,
\end{split}\]
i.e., $(s_{AB}(\bS,\bZ^{(1)}) \mid \phi_M(\sqrt{n^{(1)}}\hat{\btau}_{\bx,1}(\bS,\bZ^{(1)}),\bV_{\bx\bx,1}(\bS))=1)-S_{AB} = o_p(1)$.
Lemma \ref{lem:equiv} shows that 
\[\begin{split}
\sup_\cA \big|\bbP((\bS,\bW^{(1)}) \in \cA)-
\bbP((\bS,\bZ^{(1)}) \in \cA \mid 
\phi_M(\sqrt{n^{(1)}}\hat{\btau}_{\bx,1}(\bS,\bZ^{(1)}),\bV_{\bx\bx,1}(\bS))=1)\big|
\rightarrow 0.
\end{split}\]
Therefore, for any $c$,
\[\begin{split}
\sup_c \big|&\bbP(|s_{AB,q}(\bS,\bW^{(1)})-S_{AB}| \geq c)
-\\& \bbP(|s_{AB,q}(\bS,\bZ^{(1)})-S_{AB}| \geq c \mid 
\phi_M(\sqrt{n^{(1)}}\hat{\btau}_{\bx,1}(\bS,\bZ^{(1)}),\bV_{\bx\bx,1}(\bS))=1)\big|
\rightarrow 0.
\end{split}\]
Hence, $s_{AB}(\bS,\bW^{(1)})-S_{AB} = o_p(1)$.
Therefore, 
under DA-ReO$_\F$ or DA-ReO$_\FE$, we have
\[\begin{split}
\bms_{\bx\bx,q}(\bS,\bW^{(1)}) - \bS_{\bx\bx} = o_p(1),\      
\bms_{\bx,q}(\bS,\bW^{(1)}) - \bS_{\bx,q} = o_p(1) \text{ for }q = 1,\ldots,Q.
\end{split}\]
Since $\bS_{\bx\bx} \rightarrow \bS_{\bx\bx,\infty}$ is nonsingular, then Lemma \ref{lem:inverse_lemma} implies that 
\[\begin{split}
(\bms_{\bx\bx,q}(\bS,\bW^{(1)}))^{-1} - \bS_{\bx\bx}^{-1} = o_p(1).
\end{split}\]
Hence, 
$\hat{\bbeta}_q(\bS,\bW^{(1)})-\bbeta_q = o_p(1)$ for $q = 1,\ldots,Q$, and we complete the proof.
\end{proof}
\end{proof}

\subsection{Proof of Theorem \ref{thm:distribution_fac}}

\begin{proof}
Before proving Theorem \ref{thm:distribution_fac}, we first prove the following lemmas.
\begin{lemma}\label{lem:prob2}
Under Assumption \ref{assump:V_tautau} and Assumption \ref{assump:regularity_fac_full}, as $n\rightarrow \infty$, for random splitting $\bS$, completely random assignment $\bZ^{(2)}$, and rerandomization assignment $\bW^{(1)}$ defined in Eq.~\eqref{eq:assign1_fac}, we have 
\[\begin{split}
&\bbP(\phi_\bomega(\sqrt{n^{(2)}}\hat{\btau}_{\bx,2}(\bS,\bZ^{(2)}),\bV_{\bx\bx,2}(\bS),\hat{\bB}(\bS,\bW^{(1)})) = 1 \mid \bS,\bW^{(1)})\\
=& 
\bbP(\phi_\bomega(\bD_{2,\infty},\bV_{\bx\bx,\infty},\bB_\infty) = 1) + o_p(1),
\end{split}\]
where $\bD_{2,\infty}\sim \cN(\bZero,\bV_{\bx\bx,\infty})$, $\bV_{\bx\bx,\infty} = \sum_{q=1}^Q \frac{(\bA_q\bA_q')\otimes \bS_{\bx\bx,\infty}}{r_q}$, and 
$\bB_\infty$ is the limit of $\bB$.
\end{lemma}

\begin{lemma}\label{lem:equiv2}
Under Assumption \ref{assump:V_tautau} and Assumption \ref{assump:regularity_fac_full}, let $\cD=\{0,1\}^n \times \{1,\ldots,Q\}^{n^{(1)}} \times \{1,\ldots,Q\}^{n^{(2)}}$. As $n \rightarrow \infty$, 
for random splitting $\bS$, 
treatment assignments $\bZ^{(1)},\bZ^{(2)}$ in S-CRFE, and 
treatment assignments $\bW^{(1)},\bW^{(2)}$ in DA-ReO$_\F$ (or DA-ReO$_\FE$),
we have 
\[\begin{split}
&\sup_{\cA \subset \cD} |\bbP((\bS,\bW^{(1)},\bW^{(2)}) \in \cA) - 
\bbP((\bS,\bZ^{(1)},\bZ^{(2)}) \in \cA \mid \phi_M(\sqrt{n^{(1)}}\hat{\btau}_{\bx,1}(\bS,\bZ^{(1)}),\bV_{\bx\bx,1}(\bS))=1, \\
&\ \ \ \ \ \ \ \ \ \ \ \ \phi_\bomega(\sqrt{n^{(2)}}\hat{\btau}_{\bx,2}(\bS,\bZ^{(2)}),\bV_{\bx\bx,2}(\bS),
\hat{\bB}(\bS,\bZ^{(1)}))=1)|
\rightarrow 0.
\end{split}\]
\end{lemma}

\begin{proof}[Proof of Lemma \ref{lem:prob2}]
Let 
\[\begin{split}
\psi_{n,2}(\bS,\bW^{(1)}) &\equiv \bbP(\phi_\bomega(\sqrt{n^{(2)}}\hat{\btau}_{\bx,2}(\bS,\bZ^{(2)}),\bV_{\bx\bx,2}(\bS),\hat{\bB}(\bS,\bW^{(1)})) = 1 \mid \bS,\bW^{(1)}).
\end{split}\]
By Theorem 2.3.2 in \cite{durrett2019probability}, it suffices to show that for any subsequence $n(m) \rightarrow \infty$,
there exists a further subsequence $n(m_k) \rightarrow \infty$ s.t.
\begin{align}
\label{eq:res2}
&\psi_{n(m_k),2}(\bS,\bW^{(1)}) - \bbP(\phi_\bomega(\bD_{2,\infty},\bV_{\bx\bx,\infty},\bB_\infty) = 1) \xrightarrow{a.s.} 0.
\end{align}
By definition, 
\[\begin{split} 
\bV_{\bx\bx,2}(\bS) =\sum_{q=1}^Q \frac{(\bA_q\bA_q')\otimes \bms_{\bx\bx}(\cI_2(\bS))}{r_{q}},
\end{split}\]
and according to Lemma B4 in \cite{yang2023rejective}, under Assumption \ref{assump:regularity_fac_full}, we have 
$\bms_{\bx\bx}(\cI_2(\bS)) - \bS_{\bx\bx} = o_p(1).$
Therefore, $\bV_{\bx\bx,2}(\bS) - \bV_{\bx\bx,\infty} = o_p(1).$
By Proposition \ref{pro:consistent_fac}, $\hat{\bB}(\bS,\bW^{(1)}) - \bB = o_p(1)$.
Then Theorem 2.3.2 in \cite{durrett2019probability} implies that for any subsequence $n(m) \rightarrow \infty$,
there exists a further subsequence $n(m_k) \rightarrow \infty$ s.t. $\bV_{\bx\bx,2}(\bS) \xrightarrow{a.s.} \bV_{\bx\bx,\infty}$ and $\hat
\bB(\bS,\bW^{(1)}) \xrightarrow{a.s.} \bB_\infty$.
Moreover, for
$\bar{\bX}_2(\bS) = (n^{(2)})^{-1}\sum_{i=1}^n (1-S_i) \bX_i$, 
we have 
\[\begin{split}
||\bX_i - \bar{\bX}_2(\bS)||_2 &\leq ||\bX_i - \bar{\bX}||_2 + ||\bar{\bX} - \bar{\bX}_2(\bS)||_2
\leq 2\max_{1\leq j\leq n} ||\bX_j - \bar{\bX}||_2.
\end{split}\]
Under Assumption \ref{assump:regularity_fac_full}, we have 
\[\begin{split}
&\frac{1}{n^{(2)}}\max_{i:S_i=0}||\bX_i - \bar{\bX}_2(\bS)||_2^2 \leq 
\frac{4}{n^{(2)}}\max_{1\leq j\leq n} ||\bX_j - \bar{\bX}||_2^2 = o(1).
\end{split}\]
According to the finite population CLT for CRFE, for the subsequence $n(m_k)\rightarrow\infty$,
when the sequence of $(\bms_{n(m_k)},\bw_{n(m_k)})$ satisfies that 
$\bV_{\bx\bx,2}(\bms_{n(m_k)}) - \bV_{\bx\bx,\infty} = o(1)$, and 
$\hat\bB(\bms_{n(m_k)},\bw_{n(m_k)}) - \bB_\infty = o(1)$,
we have 
\[\begin{split}
(\sqrt{n^{(2)}}\hat{\btau}_{\bx,2}(\bS,\bZ^{(2)})\mid \bS = \bms_{n(m_k)}) \xrightarrow{d} 
\bD_{2,\infty} \sim \cN(\bZero,\bV_{\bx\bx,\infty}).
\end{split}\]
Lemma \ref{lem:joint} and the continuous mapping theorem implies that 
\[\begin{split}
(\phi_\bomega(\sqrt{n^{(2)}}\hat{\btau}_{\bx,2}(\bS,\bZ^{(2)}),\bV_{\bx\bx,2}(\bS),\hat{\bB}(\bS,\bW^{(1)}))\mid \bS = \bms_{n(m_k)},\bW^{(1)}=\bw_{n(m_k)})
\xrightarrow{d} 
\phi_\bomega(\bD_{2,\infty},\bV_{\bx\bx,\infty},\bB_\infty),
\end{split}\]
and thus, 
\[\begin{split}
&\psi_{n(m_k),2}(\bms_{n(m_k)},\bw_{n(m_k)})\\ =& \bbP(\phi_\bomega(\sqrt{n^{(2)}}\hat{\btau}_{\bx,2}(\bS,\bZ^{(2)}),\bV_{\bx\bx,2}(\bS),\hat{\bB}(\bS,\bW^{(1)})) = 1 \mid \bS=\bms_{n(m_k)},\bW^{(1)}=\bw_{n(m_k)}) \\ \rightarrow &
\bbP(\phi_\bomega(\bD_{2,\infty},\bV_{\bx\bx,\infty},\bB_\infty)=1).
\end{split}\]
Since $\bV_{\bx\bx,2}(\bS) \rightarrow \bV_{\bx\bx,\infty}$ a.s. for $n({m_k})\rightarrow\infty$, we have Eq.~\eqref{eq:res2} holds.
Then we complete the proof.
\end{proof}

\begin{proof}[Proof of Lemma \ref{lem:equiv2}]
By Theorem \ref{thm:TSCRFE_asymp}, under S-CRFE and Assumption \ref{assump:regularity_fac_full},
\[\begin{split}
&\big(\sqrt{n^{(1)}}(\hat{\btau}_1(\bS,\bZ^{(1)})-\btau)',
\sqrt{n^{(1)}}\hat{\btau}_{\bx,1}'(\bS,\bZ^{(1)}),
\sqrt{n^{(2)}}(\hat{\btau}_2(\bS,\bZ^{(2)})-\btau)',
\sqrt{n^{(2)}}\hat{\btau}_{\bx,2}'(\bS,\bZ^{(2)})\big)'\\
&
\overset{\cdot}{\sim}
\cN(\bZero, \bV_{\TSCRFE}).
\end{split}\]
Suppose $(\bC_1',\bD_1',\bC_2',\bD_2')'
\sim \cN(\bZero,\bV_\TSCRFE)$.
According to Lemma B4 in \cite{yang2023rejective}, under S-CRFE, we have 
$\bV_{\bx\bx,1}(\bS) \xrightarrow{p} \bV_{\bx\bx,\infty}$, and $\bV_{\bx\bx,2}(\bS) \xrightarrow{p} \bV_{\bx\bx,\infty}$.
Hence, Lemma \ref{lem:joint} and continuous mapping theorem imply that 
\[\begin{split}
&(\phi_M(\sqrt{n^{(1)}}\hat{\btau}_{\bx,1}(\bS,\bZ^{(1)}),\bV_{\bx\bx,1}(\bS)),
\phi_\bomega(\sqrt{n^{(2)}}\hat{\btau}_{\bx,2}(\bS,\bZ^{(2)}),\bV_{\bx\bx,2}(\bS),
\hat{\bB}(\bS,\bZ^{(1)}))) \\
\xrightarrow{d}&
(\phi_M(\bD_{1,\infty},\bV_{\bx\bx,\infty}),
\phi_\bomega(\bD_{2,\infty},\bV_{\bx\bx,\infty},\bB_\infty)),
\end{split}\]
where $(\bD_{1,\infty}',\bD_{2,\infty}')' \sim \cN\left(\bZero,\begin{pmatrix}
\bV_{\bx\bx,\infty} & \bZero\\
\bZero & \bV_{\bx\bx,\infty}
\end{pmatrix}\right)$.
Therefore,
\begin{align}
&\notag
\lim_{n\rightarrow\infty}
\bbP(\phi_M(\sqrt{n^{(1)}}\hat{\btau}_{\bx,1}(\bS,\bZ^{(1)}),\bV_{\bx\bx,1}(\bS))=1,\\
&\notag\quad\quad\quad
\phi_\bomega(\sqrt{n^{(2)}}\hat{\btau}_{\bx,2}(\bS,\bZ^{(2)}),\bV_{\bx\bx,2}(\bS),
\hat{\bB}(\bS,\bZ^{(1)}))=1)\\
= &\label{eq:prob1}\bbP(\phi_M(\bD_{1,\infty},\bV_{\bx\bx,\infty})=1)\bbP(
\phi_\bomega(\bD_{2,\infty},\bV_{\bx\bx,\infty},\bB_\infty)=1)=\alpha^2.
\end{align}
Then for any $0 < a < \alpha^2$, there exists some $N_a$ s.t. 
for any $n > N_a$, we have 
\begin{equation}\label{eq:probb0}
\begin{split}
&\bbP(\phi_M(\sqrt{n^{(1)}}\hat{\btau}_{\bx,1}(\bS,\bZ^{(1)}),\bV_{\bx\bx,1}(\bS))=1,\\
&\quad\phi_\bomega(\sqrt{n^{(2)}}\hat{\btau}_{\bx,2}(\bS,\bZ^{(2)}),\bV_{\bx\bx,2}(\bS),
\hat{\bB}(\bS,\bZ^{(1)}))=1)
\in [\alpha^2-a,\alpha^2+a].
\end{split}
\end{equation}
Moreover,
Lemma \ref{lem:prob1} and Lemma \ref{lem:prob2} imply that 
\[\begin{split}
&\bbP(\phi_M(\sqrt{n^{(1)}}\hat{\btau}_{\bx,1}(\bS,\bZ^{(1)}),\bV_{\bx\bx,1}(\bS))=1 \mid \bS) 
= \bbP(\phi_M(\bD_{1,\infty},\bV_{\bx\bx,\infty}) = 1) + o_p(1),\\
&\bbP(\phi_\bomega(\sqrt{n^{(2)}}\hat{\btau}_{\bx,2}(\bS,\bZ^{(2)}),\bV_{\bx\bx,2}(\bS),\hat{\bB}(\bS,\bW^{(1)})) = 1 \mid \bS,\bW^{(1)})\\ =& 
\bbP(\phi_\bomega(\bD_{2,\infty},\bV_{\bx\bx,\infty},\bB_\infty) = 1) + o_p(1),
\end{split}\]
where $\bbP(\phi_M(\bD_{1,\infty},\bV_{\bx\bx,\infty}) = 1) = \alpha$, $\bbP(\phi_\bomega(\bD_{2,\infty},\bV_{\bx\bx,\infty},\bB_\infty) = 1) = \alpha$.
Similarly, we have 
\[\begin{split}
&\bbP(\phi_\bomega(\sqrt{n^{(2)}}\hat{\btau}_{\bx,2}(\bS,\bZ^{(2)}),\bV_{\bx\bx,2}(\bS),\hat{\bB}(\bS,\bZ^{(1)})) = 1 \mid \bS,\bZ^{(1)})\\ =& 
\bbP(\phi_\bomega(\bD_{2,\infty},\bV_{\bx\bx,\infty},\bB_\infty) = 1) + o_p(1).
\end{split}\]
Therefore, for any $a > 0$, define 
\[\begin{split}
\cS_a &= \left\{\bms:
\left|\bbP(\phi_M(\sqrt{n^{(1)}}\hat{\btau}_{\bx,1}(\bms,\bZ^{(1)}),\bV_{\bx\bx,1}(\bms))=1)-\alpha\right| \leq a\right\}\times \{1,\ldots,Q\}^{n^{(1)}} \cap \\
&\left\{(\bms,\bw): \left|\bbP(\phi_\bomega(\sqrt{n^{(2)}}\hat{\btau}_{\bx,2}(\bms,\bZ^{(2)}),\bV_{\bx\bx,2}(\bms),\hat{\bB}(\bms,\bw))=1)-\alpha\right| \leq a\right\}
\end{split}\]
there exists $N_a$ s.t. for $n > N_a$,
\[\begin{split}
\bbP((\bS,\bZ^{(1)}) \not\in \cS_a) \leq a,\text{ and }
\bbP((\bS,\bW^{(1)}) \not\in \cS_a) \leq a.
\end{split}\]
Moreover, for $(\bms,\bw) \in \cS_a$ and $a < \alpha^2$, then 
\begin{align}
\label{eq:probb1}
0 < \bbP(\phi_M(\sqrt{n^{(1)}}\hat{\btau}_{\bx,1}(\bms,\bZ^{(1)}),\bV_{\bx\bx,1}(\bms))=1) \in [\alpha-a,\alpha+a],\\
\label{eq:probb2}
0 < \bbP(\phi_\bomega(\sqrt{n^{(2)}}\hat{\btau}_{\bx,2}(\bms,\bZ^{(2)}),\bV_{\bx\bx,2}(\bms),\hat{\bB}(\bms,\bw))=1) \in [\alpha-a,\alpha+a],
\end{align}

By definition, for $n > N_a$ and any $\cA \subset \cD$,
\[\begin{split}
\bbP((\bS,\bW^{(1)},\bW^{(2)}) \in \cA)
&\leq 
\bbP((\bS,\bW^{(1)},\bW^{(2)}) \in \cA,
(\bS,\bW^{(1)}) \in \cS_a) + \bbP((\bS,\bW^{(1)}) \not\in \cS_a)\\
&\leq \sum_{(\bms,\bw) \in \cS_a} 
\bbP((\bS,\bW^{(1)},\bW^{(2)}) \in \cA,
\bS=\bms,\bW^{(1)}=\bw) + a,
\end{split}\]
and 
\[\begin{split}
\bbP((\bS,\bW^{(1)},\bW^{(2)}) \in \cA)
&\geq  \bbP((\bS,\bW^{(1)},\bW^{(2)}) \in \cA,(\bS,\bW^{(1)}) \in \cS_a)
\\
&= \sum_{(\bms,\bw) \in \cS_a}
\bbP((\bS,\bW^{(1)},\bW^{(2)}) \in \cA,
\bS=\bms,\bW^{(1)} = \bw).
\end{split}\]
For any $\bms$, define $\cA_1(\bms) = \{\bw:\phi_M(\sqrt{n^{(1)}}\hat{\btau}_{\bx,1}(\bms,\bw),\bV_{\bx\bx,1}(\bms))=1\}$.
For $(\bms,\bw) \in \cS_a$,
if $\bw \not\in \cA_1(\bms)$,
then 
\[\begin{split}
\bbP(\bS = \bms,\bW^{(1)}=\bw) &= 
\bbP(\bS = \bms)\bbP(\bW^{(1)}=\bw\mid \bS = \bms)\\
&= \bbP(\bS = \bms)\bbP(\bZ^{(1)}=\bw\mid  \phi_M(\sqrt{n^{(1)}}\hat{\btau}_{\bx,1}(\bms,\bZ^{(1)}),\bV_{\bx\bx,1}(\bms))=1)
= 0,
\end{split}\]
and thus, $\bbP((\bS,\bW^{(1)},\bW^{(2)}) \in \cA,
\bS=\bms,\bW^{(1)} = \bw) = 0$.
Hence, for $(\bms,\bw) \in \cS_a$,
\[\begin{split}
&\bbP((\bS,\bW^{(1)},\bW^{(2)}) \in \cA,
\bS=\bms,\bW^{(1)}=\bw)\\
=& 1\{\bw \in \cA_1(\bms)\}\cdot
\bbP((\bS,\bW^{(1)},\bW^{(2)}) \in \cA,
\bS=\bms,\bW^{(1)}=\bw)\\
=& 1\{\bw \in \cA_1(\bms)\}\cdot
\bbP((\bS,\bW^{(1)},\bW^{(2)}) \in \cA\mid
\bS=\bms,\bW^{(1)}=\bw)\cdot \bbP(\bS=\bms,\bW^{(1)}=\bw),
\end{split}\]
where for $(\bms,\bw) \in \cS_a$ and $\bw \in \cA_1(\bms)$,
\[\begin{split}
&\bbP((\bS,\bW^{(1)},\bW^{(2)}) \in \cA
\mid 
\bS=\bms,\bW^{(1)}=\bw)
\\
=&\frac{
\bbP((\bms,\bw,\bW^{(2)}) \in \cA,
\bS=\bms,\bW^{(1)}=\bw)
}{\bbP(\bS=\bms,\bW^{(1)}=\bw)}
\\
=&
\bbP((\bms,\bw,\bW^{(2)}) \in \cA\mid
\bS=\bms,\bW^{(1)}=\bw)
\\
=& 
\bbP((\bms,\bw,\bZ^{(2)}) \in \cA
\mid 
\phi_\bomega(\sqrt{n^{(2)}}\hat{\btau}_{\bx,2}(\bms,\bZ^{(2)}),\bV_{xx,2}(\bms),\hat{\bB}(\bms,\bw))=1)\\
& \bbP((\bS,\bZ^{(1)},\bZ^{(2)}) \in \cA
\mid \bS = \bms, \bZ^{(1)} = \bw,\\
& \ \ \ \ \ \ \ \ \phi_M(\sqrt{n^{(1)}}\hat{\btau}_{\bx,1}(\bS,\bZ^{(1)}),\bV_{\bx\bx,1}(\bS))=1,\\
& \ \ \ \ \ \ \ \ 
\phi_\bomega(\sqrt{n^{(2)}}\hat{\btau}_{\bx,2}(\bS,\bZ^{(2)}),\bV_{\bx\bx,2}(\bS),\hat{\bB}(\bS,\bZ^{(1)}))=1).
\end{split}\]
Then for $(\bms,\bw) \in \cS_a$,
\[\begin{split}
    &\bbP((\bS,\bW^{(1)},\bW^{(2)}) \in \cA,
\bS=\bms,\bW^{(1)}=\bw)\\
=&1\{\bw \in \cA_1(\bms)\}\cdot\bbP((\bS,\bW^{(1)},\bW^{(2)}) \in \cA\mid
\bS=\bms,\bW^{(1)}=\bw)\cdot \bbP(\bS=\bms,\bW^{(1)}=\bw)\\
=& 1\{\bw \in \cA_1(\bms)\}\cdot 
\bbP(\bS = \bms, \bW^{(1)} = \bw)\cdot\bbP((\bS,\bZ^{(1)},\bZ^{(2)}) \in \cA
\mid \bS = \bms, \bZ^{(1)} = \bw,\\
& \ \ \ \ \ \ \ \ \phi_M(\sqrt{n^{(1)}}\hat{\btau}_{\bx,1}(\bS,\bZ^{(1)}),\bV_{\bx\bx,1}(\bS))=1,\\
& \ \ \ \ \ \ \ \ 
\phi_\bomega(\sqrt{n^{(2)}}\hat{\btau}_{\bx,2}(\bS,\bZ^{(2)}),\bV_{\bx\bx,2}(\bS),\hat{\bB}(\bS,\bZ^{(1)}))=1).
\end{split}\]
Thus,
\[\begin{split}
&\sum_{(\bms,\bw) \in \cS_a}
\bbP((\bS,\bW^{(1)},\bW^{(2)}) \in \cA,
\bS=\bms,\bW^{(1)} = \bw)\\
=& \sum_{(\bms,\bw) \in \cS_a}
1\{\bw \in \cA_1(\bms)\}\cdot 
\bbP(\bS = \bms, \bW^{(1)} = \bw)\cdot\bbP((\bS,\bZ^{(1)},\bZ^{(2)}) \in \cA
\mid \bS = \bms, \bZ^{(1)} = \bw,\\
& \ \ \ \ \ \ \ \ \phi_M(\sqrt{n^{(1)}}\hat{\btau}_{\bx,1}(\bS,\bZ^{(1)}),\bV_{\bx\bx,1}(\bS))=1,
\phi_\bomega(\sqrt{n^{(2)}}\hat{\btau}_{\bx,2}(\bS,\bZ^{(2)}),\bV_{\bx\bx,2}(\bS),\hat{\bB}(\bS,\bZ^{(1)}))=1).
\end{split}\]
For $(\bms,\bw) \in \cS_a$, since 
\[\begin{split}
&\bbP(\bS = \bms, \bZ^{(1)} = \bw\mid 
\phi_M(\sqrt{n^{(1)}}\hat{\btau}_{\bx,1}(\bS,\bZ^{(1)}),\bV_{\bx\bx,1}(\bS))=1,\\
&\quad\quad
\phi_\bomega(\sqrt{n^{(2)}}\hat{\btau}_{\bx,2}(\bS,\bZ^{(2)}),\bV_{\bx\bx,2}(\bS),\hat{\bB}(\bS,\bZ^{(1)}))=1)\\
=& \bbP(\bS = \bms, \bZ^{(1)} = \bw, 
\phi_M(\sqrt{n^{(1)}}\hat{\btau}_{\bx,1}(\bS,\bZ^{(1)}),\bV_{\bx\bx,1}(\bS))=1)\cdot \\
&\bbP(
\phi_\bomega(\sqrt{n^{(2)}}\hat{\btau}_{\bx,2}(\bms,\bZ^{(2)}),\bV_{\bx\bx,2}(\bms),\hat{\bB}(\bms,\bw))=1)/\\
&
\bbP(\phi_M(\sqrt{n^{(1)}}\hat{\btau}_{\bx,1}(\bS,\bZ^{(1)}),\bV_{\bx\bx,1}(\bS))=1,
\phi_\bomega(\sqrt{n^{(2)}}\hat{\btau}_{\bx,2}(\bS,\bZ^{(2)}),\bV_{\bx\bx,2}(\bS),\hat{\bB}(\bS,\bZ^{(1)}))=1),
\end{split}\]
where
\[\begin{split}
&\bbP(\bS = \bms, \bZ^{(1)} = \bw, 
\phi_M(\sqrt{n^{(1)}}\hat{\btau}_{\bx,1}(\bS,\bZ^{(1)}),\bV_{\bx\bx,1}(\bS))=1)\\
=& \bbP(\bZ^{(1)} = \bw\mid \bS = \bms, 
\phi_M(\sqrt{n^{(1)}}\hat{\btau}_{\bx,1}(\bS,\bZ^{(1)}),\bV_{\bx\bx,1}(\bS))=1)\cdot \\
&
\bbP(\bS = \bms, 
\phi_M(\sqrt{n^{(1)}}\hat{\btau}_{\bx,1}(\bS,\bZ^{(1)}),\bV_{\bx\bx,1}(\bS))=1)\\
=& \bbP(\bZ^{(1)} = \bw\mid 
\phi_M(\sqrt{n^{(1)}}\hat{\btau}_{\bx,1}(\bms,\bZ^{(1)}),\bV_{\bx\bx,1}(\bms))=1)
\bbP(\bS = \bms, 
\phi_M(\sqrt{n^{(1)}}\hat{\btau}_{\bx,1}(\bms,\bZ^{(1)}),\bV_{\bx\bx,1}(\bms))=1)\\
=& \bbP(\bW^{(1)} = \bw\mid 
\bS = \bms)
\bbP(\bS = \bms) \bbP(
\phi_M(\sqrt{n^{(1)}}\hat{\btau}_{\bx,1}(\bms,\bZ^{(1)}),\bV_{\bx\bx,1}(\bms))=1)\\
=& \bbP(\bS = \bms,\bW^{(1)} = \bw)
\bbP(
\phi_M(\sqrt{n^{(1)}}\hat{\btau}_{\bx,1}(\bms,\bZ^{(1)}),\bV_{\bx\bx,1}(\bms))=1),
\end{split}\]
we have 
\[\begin{split}
&\bbP(\bS = \bms,\bW^{(1)} = \bw)\\ =& 
\frac{\bbP(\bS = \bms, \bZ^{(1)} = \bw, 
\phi_M(\sqrt{n^{(1)}}\hat{\btau}_{\bx,1}(\bms,\bZ^{(1)}),\bV_{\bx\bx,1}(\bms))=1)}{\bbP(
\phi_M(\sqrt{n^{(1)}}\hat{\btau}_{\bx,1}(\bms,\bZ^{(1)}),\bV_{\bx\bx,1}(\bms))=1)}\\
=&\frac{\bbP(\phi_M(\sqrt{n^{(1)}}\hat{\btau}_{\bx,1}(\bS,\bZ^{(1)}),\bV_{\bx\bx,1}(\bS))=1,
\phi_\bomega(\sqrt{n^{(2)}}\hat{\btau}_{\bx,2}(\bS,\bZ^{(2)}),\bV_{\bx\bx,2}(\bS),\hat{\bB}(\bS,\bZ^{(1)}))=1)}{\bbP(
\phi_M(\sqrt{n^{(1)}}\hat{\btau}_{\bx,1}(\bms,\bZ^{(1)}),\bV_{\bx\bx,1}(\bms))=1)
\bbP(
\phi_\bomega(\sqrt{n^{(2)}}\hat{\btau}_{\bx,2}(\bms,\bZ^{(2)}),\bV_{\bx\bx,2}(\bms),\hat{\bB}(\bms,\bw))=1)}\cdot \\
& \bbP(\bS = \bms, \bZ^{(1)} = \bw\mid 
\phi_M(\sqrt{n^{(1)}}\hat{\btau}_{\bx,1}(\bS,\bZ^{(1)}),\bV_{\bx\bx,1}(\bS))=1,\\
&\quad\quad
\phi_\bomega(\sqrt{n^{(2)}}\hat{\btau}_{\bx,2}(\bS,\bZ^{(2)}),\bV_{\bx\bx,2}(\bS),\hat{\bB}(\bS,\bZ^{(1)}))=1).
\end{split}\]
Thus,
\[\begin{split}
&\sum_{(\bms,\bw) \in \cS_a}
\bbP((\bS,\bW^{(1)},\bW^{(2)}) \in \cA,
\bS=\bms,\bW^{(1)} = \bw)\\
=&\sum_{(\bms,\bw) \in \cS_a} 
1\{\bw \in \cA_1(\bms)\}\cdot 
\bbP((\bS,\bZ^{(1)},\bZ^{(2)}) \in \cA
\mid \bS = \bms, \bZ^{(1)} = \bw,\\
& \ \ \ \ \ \ \ \ \phi_M(\sqrt{n^{(1)}}\hat{\btau}_{\bx,1}(\bS,\bZ^{(1)}),\bV_{\bx\bx,1}(\bS))=1,
\phi_\bomega(\sqrt{n^{(2)}}\hat{\btau}_{\bx,2}(\bS,\bZ^{(2)}),\bV_{\bx\bx,2}(\bS),\hat{\bB}(\bS,\bZ^{(1)}))=1)\cdot \\
& \bbP(\bS = \bms, \bZ^{(1)} = \bw\mid 
\phi_M(\sqrt{n^{(1)}}\hat{\btau}_{\bx,1}(\bS,\bZ^{(1)}),\bV_{\bx\bx,1}(\bS))=1,\\
&\ \ \ \ \ \ \ \ 
\phi_\bomega(\sqrt{n^{(2)}}\hat{\btau}_{\bx,2}(\bS,\bZ^{(2)}),\bV_{\bx\bx,2}(\bS),\hat{\bB}(\bS,\bZ^{(1)}))=1)\cdot \\
&\frac{\bbP(\phi_M(\sqrt{n^{(1)}}\hat{\btau}_{\bx,1}(\bS,\bZ^{(1)}),\bV_{\bx\bx,1}(\bS))=1,
\phi_\bomega(\sqrt{n^{(2)}}\hat{\btau}_{\bx,2}(\bS,\bZ^{(2)}),\bV_{\bx\bx,2}(\bS),\hat{\bB}(\bS,\bZ^{(1)}))=1)}{\bbP(
\phi_M(\sqrt{n^{(1)}}\hat{\btau}_{\bx,1}(\bms,\bZ^{(1)}),\bV_{\bx\bx,1}(\bms))=1)
\bbP(
\phi_\bomega(\sqrt{n^{(2)}}\hat{\btau}_{\bx,2}(\bms,\bZ^{(2)}),\bV_{\bx\bx,2}(\bms),\hat{\bB}(\bms,\bw))=1)}\\
=& \sum_{(\bms,\bw) \in \cS_a}  
1\{\bw \in \cA_1(\bms)\}\cdot 
\bbP((\bS,\bZ^{(1)},\bZ^{(2)}) \in \cA, \bS = \bms, \bZ^{(1)} = \bw \mid\\
& \ \ \ \ \ \ \ \ \phi_M(\sqrt{n^{(1)}}\hat{\btau}_{\bx,1}(\bS,\bZ^{(1)}),\bV_{\bx\bx,1}(\bS))=1,
\phi_\bomega(\sqrt{n^{(2)}}\hat{\btau}_{\bx,2}(\bS,\bZ^{(2)}),\bV_{\bx\bx,2}(\bS),\hat{\bB}(\bS,\bZ^{(1)}))=1)\cdot \\
&\frac{\bbP(\phi_M(\sqrt{n^{(1)}}\hat{\btau}_{\bx,1}(\bS,\bZ^{(1)}),\bV_{\bx\bx,1}(\bS))=1,
\phi_\bomega(\sqrt{n^{(2)}}\hat{\btau}_{\bx,2}(\bS,\bZ^{(2)}),\bV_{\bx\bx,2}(\bS),\hat{\bB}(\bS,\bZ^{(1)}))=1)}{\bbP(
\phi_M(\sqrt{n^{(1)}}\hat{\btau}_{\bx,1}(\bms,\bZ^{(1)}),\bV_{\bx\bx,1}(\bms))=1)
\bbP(
\phi_\bomega(\sqrt{n^{(2)}}\hat{\btau}_{\bx,2}(\bms,\bZ^{(2)}),\bV_{\bx\bx,2}(\bms),\hat{\bB}(\bms,\bw))=1)}.
\end{split}\]
Moreover, since for $(\bms,\bw) \in \cS_a$ and $\bw \not\in \cA_1(\bms)$,
\[\begin{split}
&\bbP((\bS,\bZ^{(1)},\bZ^{(2)}) \in \cA, \bS = \bms, \bZ^{(1)} = \bw \mid\\
& \ \ \ \ \ \ \ \ \phi_M(\sqrt{n^{(1)}}\hat{\btau}_{\bx,1}(\bS,\bZ^{(1)}),\bV_{\bx\bx,1}(\bS))=1,
\phi_\bomega(\sqrt{n^{(2)}}\hat{\btau}_{\bx,2}(\bS,\bZ^{(2)}),\bV_{\bx\bx,2}(\bS),\hat{\bB}(\bS,\bZ^{(1)}))=1)
= 0,
\end{split}\]
we have 
\[\begin{split}
&\sum_{(\bms,\bw) \in \cS_a}
\bbP((\bS,\bW^{(1)},\bW^{(2)}) \in \cA,
\bS=\bms,\bW^{(1)} = \bw)\\
=& \sum_{(\bms,\bw) \in \cS_a}   
\bbP((\bS,\bZ^{(1)},\bZ^{(2)}) \in \cA, \bS = \bms, \bZ^{(1)} = \bw \mid\\
& \ \ \ \ \ \ \ \ \phi_M(\sqrt{n^{(1)}}\hat{\btau}_{\bx,1}(\bS,\bZ^{(1)}),\bV_{\bx\bx,1}(\bS))=1,
\phi_\bomega(\sqrt{n^{(2)}}\hat{\btau}_{\bx,2}(\bS,\bZ^{(2)}),\bV_{\bx\bx,2}(\bS),\hat{\bB}(\bS,\bZ^{(1)}))=1)\cdot \\
&\frac{\bbP(\phi_M(\sqrt{n^{(1)}}\hat{\btau}_{\bx,1}(\bS,\bZ^{(1)}),\bV_{\bx\bx,1}(\bS))=1,
\phi_\bomega(\sqrt{n^{(2)}}\hat{\btau}_{\bx,2}(\bS,\bZ^{(2)}),\bV_{\bx\bx,2}(\bS),\hat{\bB}(\bS,\bZ^{(1)}))=1)}{\bbP(
\phi_M(\sqrt{n^{(1)}}\hat{\btau}_{\bx,1}(\bms,\bZ^{(1)}),\bV_{\bx\bx,1}(\bms))=1)
\bbP(
\phi_\bomega(\sqrt{n^{(2)}}\hat{\btau}_{\bx,2}(\bms,\bZ^{(2)}),\bV_{\bx\bx,2}(\bms),\hat{\bB}(\bms,\bw))=1)}.
\end{split}\]
For $(\bms,\bw) \in \cS_a$ and 
$n>N_a$,
we have 
\[\begin{split}
&\frac{\bbP(\phi_M(\sqrt{n^{(1)}}\hat{\btau}_{\bx,1}(\bS,\bZ^{(1)}),\bV_{\bx\bx,1}(\bS))=1,
\phi_\bomega(\sqrt{n^{(2)}}\hat{\btau}_{\bx,2}(\bS,\bZ^{(2)}),\bV_{\bx\bx,2}(\bS),\hat{\bB}(\bS,\bZ^{(1)}))=1)}{\bbP(
\phi_M(\sqrt{n^{(1)}}\hat{\btau}_{\bx,1}(\bms,\bZ^{(1)}),\bV_{\bx\bx,1}(\bms))=1)
\bbP(
\phi_\bomega(\sqrt{n^{(2)}}\hat{\btau}_{\bx,2}(\bms,\bZ^{(2)}),\bV_{\bx\bx,2}(\bms),\hat{\bB}(\bms,\bw))=1)}\\
\in & \left[\frac{
\alpha^2-a
}{(\alpha+a)^2},\frac{
\alpha^2+a
}{(\alpha-a)^2}\right].
\end{split}\]
Hence, 
\[\begin{split}
&\sum_{(\bms,\bw) \in \cS_a}
\bbP((\bS,\bW^{(1)},\bW^{(2)}) \in \cA,
\bS=\bms,\bW^{(1)} = \bw)\\
\leq & \frac{
\alpha^2+a
}{(\alpha-a)^2}\cdot 
\sum_{(\bms,\bw) \in \cS_a}
\bbP((\bS,\bZ^{(1)},\bZ^{(2)}) \in \cA, \bS = \bms, \bZ^{(1)} = \bw \mid\\
& \ \ \ \ \ \ \ \ \phi_M(\sqrt{n^{(1)}}\hat{\btau}_{\bx,1}(\bS,\bZ^{(1)}),\bV_{\bx\bx,1}(\bS))=1,
\phi_\bomega(\sqrt{n^{(2)}}\hat{\btau}_{\bx,2}(\bS,\bZ^{(2)}),\bV_{\bx\bx,2}(\bS),\hat{\bB}(\bS,\bZ^{(1)}))=1)\\ 
\leq & \frac{
\alpha^2+a
}{(\alpha-a)^2}\cdot 
\bbP((\bS,\bZ^{(1)},\bZ^{(2)}) \in \cA\mid\phi_M(\sqrt{n^{(1)}}\hat{\btau}_{\bx,1}(\bS,\bZ^{(1)}),\bV_{\bx\bx,1}(\bS))=1,\\
& \ \ \ \ \ \ \ \ \phi_\bomega(\sqrt{n^{(2)}}\hat{\btau}_{\bx,2}(\bS,\bZ^{(2)}),\bV_{\bx\bx,2}(\bS),\hat{\bB}(\bS,\bZ^{(1)}))=1),
\end{split}\]
and 
\[\begin{split}
&\sum_{(\bms,\bw) \in \cS_a}
\bbP((\bS,\bW^{(1)},\bW^{(2)}) \in \cA,
\bS=\bms,\bW^{(1)} = \bw)\\
\geq & \frac{
\alpha^2-a
}{(\alpha+a)^2}\cdot 
\sum_{(\bms,\bw) \in \cS_a}\bbP((\bS,\bZ^{(1)},\bZ^{(2)}) \in \cA, \bS = \bms,
\bZ^{(1)} =\bw\mid\\
& \ \ \ \ \ \ \ \ \phi_M(\sqrt{n^{(1)}}\hat{\btau}_{\bx,1}(\bS,\bZ^{(1)}),\bV_{\bx\bx,1}(\bS))=1,
\phi_\bomega(\sqrt{n^{(2)}}\hat{\btau}_{\bx,2}(\bS,\bZ^{(2)}),\bV_{\bx\bx,2}(\bS),\hat{\bB}(\bS,\bZ^{(1)}))=1).
\end{split}\]
Since for $n>N_a$,
\[\begin{split}
&\sum_{(\bms,\bw) \not\in \cS_a}\bbP((\bS,\bZ^{(1)},\bZ^{(2)}) \in \cA,
\bS = \bms,\bZ^{(1)} =\bw\mid\\
& \ \ \ \ \ \ \ \ \phi_M(\sqrt{n^{(1)}}\hat{\btau}_{\bx,1}(\bS,\bZ^{(1)}),\bV_{\bx\bx,1}(\bS))=1,
\phi_\bomega(\sqrt{n^{(2)}}\hat{\btau}_{\bx,2}(\bS,\bZ^{(2)}),\bV_{\bx\bx,2}(\bS),\hat{\bB}(\bS,\bZ^{(1)}))=1)\\
\leq & \frac{\bbP((\bS,\bZ^{(1)}) \not\in \cS_a)}{\bbP(\phi_M(\sqrt{n^{(1)}}\hat{\btau}_{\bx,1}(\bS,\bZ^{(1)}),\bV_{\bx\bx,1}(\bS))=1,
\phi_\bomega(\sqrt{n^{(2)}}\hat{\btau}_{\bx,2}(\bS,\bZ^{(2)}),\bV_{\bx\bx,2}(\bS),\hat{\bB}(\bS,\bZ^{(1)}))=1)}\\
\leq & \frac{a}{\alpha^2-a},
\end{split}\]
we have 
\[\begin{split}
&\sum_{(\bms,\bw) \in \cS_a}\bbP((\bS,\bZ^{(1)},\bZ^{(2)}) \in \cA, \bS = \bms,
\bZ^{(1)} =\bw\mid\\
& \ \ \ \ \ \ \ \ \phi_M(\sqrt{n^{(1)}}\hat{\btau}_{\bx,1}(\bS,\bZ^{(1)}),\bV_{\bx\bx,1}(\bS))=1,
\phi_\bomega(\sqrt{n^{(2)}}\hat{\btau}_{\bx,2}(\bS,\bZ^{(2)}),\bV_{\bx\bx,2}(\bS),\hat{\bB}(\bS,\bZ^{(1)}))=1)\\
\geq & \bbP((\bS,\bZ^{(1)},\bZ^{(2)}) \in \cA\mid\phi_M(\sqrt{n^{(1)}}\hat{\btau}_{\bx,1}(\bS,\bZ^{(1)}),\bV_{\bx\bx,1}(\bS))=1,\\
& \ \ \ \ \ \ \ \ 
\phi_\bomega(\sqrt{n^{(2)}}\hat{\btau}_{\bx,2}(\bS,\bZ^{(2)}),\bV_{\bx\bx,2}(\bS),\hat{\bB}(\bS,\bZ^{(1)}))=1)
-\frac{a}{\alpha^2-a}.
\end{split}\]
Therefore, 
\[\begin{split}
&\sum_{(\bms,\bw) \in \cS_a}
\bbP((\bS,\bW^{(1)},\bW^{(2)}) \in \cA,
\bS=\bms,\bW^{(1)} = \bw)\\
\geq & \frac{
\alpha^2-a
}{(\alpha+a)^2}
\cdot 
\bbP((\bS,\bZ^{(1)},\bZ^{(2)}) \in \cA\mid\phi_M(\sqrt{n^{(1)}}\hat{\btau}_{\bx,1}(\bS,\bZ^{(1)}),\bV_{\bx\bx,1}(\bS))=1,\\
& \ \ \ \ \ \ \ \ 
\phi_\bomega(\sqrt{n^{(2)}}\hat{\btau}_{\bx,2}(\bS,\bZ^{(2)}),\bV_{\bx\bx,2}(\bS),\hat{\bB}(\bS,\bZ^{(1)}))=1)
-\frac{
\alpha^2-a
}{(\alpha+a)^2}
\cdot \frac{a}{\alpha^2-a}\\
=& \frac{
\alpha^2-a
}{(\alpha+a)^2}
\cdot 
\bbP((\bS,\bZ^{(1)},\bZ^{(2)}) \in \cA\mid\phi_M(\sqrt{n^{(1)}}\hat{\btau}_{\bx,1}(\bS,\bZ^{(1)}),\bV_{\bx\bx,1}(\bS))=1,\\
& \ \ \ \ \ \ \ \ 
\phi_\bomega(\sqrt{n^{(2)}}\hat{\btau}_{\bx,2}(\bS,\bZ^{(2)}),\bV_{\bx\bx,2}(\bS),\hat{\bB}(\bS,\bZ^{(1)}))=1)
-\frac{a}{(\alpha+a)^2}.
\end{split}\]
Above all, for $n > N_a$,
we have 
\[\begin{split}
&\bbP((\bS,\bW^{(1)},\bW^{(2)}) \in \cA)\\
\leq & \sum_{(\bms,\bw) \in \cS_a}
\bbP((\bS,\bW^{(1)},\bW^{(2)}) \in \cA,
\bS=\bms,\bW^{(1)} = \bw) + a\\
\leq & \frac{
\alpha^2+a
}{(\alpha-a)^2}
\cdot \bbP((\bS,\bZ^{(1)},\bZ^{(2)}) \in \cA\mid\phi_M(\sqrt{n^{(1)}}\hat{\btau}_{\bx,1}(\bS,\bZ^{(1)}),\bV_{\bx\bx,1}(\bS))=1,\\
& \ \ \ \ \ \ \ \ 
\phi_\bomega(\sqrt{n^{(2)}}\hat{\btau}_{\bx,2}(\bS,\bZ^{(2)}),\bV_{\bx\bx,2}(\bS),\hat{\bB}(\bS,\bZ^{(1)}))=1)
+a,
\end{split}\]
and 
\[\begin{split}
&\bbP((\bS,\bW^{(1)},\bW^{(2)}) \in \cA)\\
\geq & \sum_{(\bms,\bw) \in \cS_a}
\bbP((\bS,\bW^{(1)},\bW^{(2)}) \in \cA,
\bS=\bms,\bW^{(1)} = \bw)\\
\geq & \frac{
\alpha^2-a
}{(\alpha+a)^2}
\cdot 
\bbP((\bS,\bZ^{(1)},\bZ^{(2)}) \in \cA\mid\phi_M(\sqrt{n^{(1)}}\hat{\btau}_{\bx,1}(\bS,\bZ^{(1)}),\bV_{\bx\bx,1}(\bS))=1,\\
& \ \ \ \ \ \ \ \ 
\phi_\bomega(\sqrt{n^{(2)}}\hat{\btau}_{\bx,2}(\bS,\bZ^{(2)}),\bV_{\bx\bx,2}(\bS),\hat{\bB}(\bS,\bZ^{(1)}))=1)
-\frac{a}{(\alpha+a)^2}.
\end{split}\]
Therefore, for any $\cA \subset \cD$,
for $n > N_a$,
\[\begin{split}
&|\bbP((\bS,\bW^{(1)},\bW^{(2)}) \in \cA) - \bbP((\bS,\bZ^{(1)},\bZ^{(2)}) \in \cA \mid 
\phi_M(\sqrt{n^{(1)}}\hat{\btau}_{\bx,1}(\bS,\bZ^{(1)}),\bV_{\bx\bx,1}(\bS))=1, \\
&\ \ \ \ \ \ \ 
\phi_\bomega(\sqrt{n^{(2)}}\hat{\btau}_{\bx,2}(\bS,\bZ^{(2)}),\bV_{\bx\bx,2}(\bS),\hat{\bB}(\bS,\bZ^{(1)}))=1)|\\
\leq &\max\left\{a + \frac{(2\alpha+1)a - a^2}{
    (\alpha-a)(\alpha-a)
},\frac{(2\alpha+2)a + a^2}{
    (\alpha+a)^2
}
\right\}.
\end{split}\]
Above all, we complete the proof.
\end{proof}

\begin{proof}[Proof of Theorem \ref{thm:distribution_fac}]
According to Lemma \ref{lem:equiv2}, 
\[\begin{split}
\sup_{\cA \subset \bbR^{2F}}\Big|&\bbP\big((\sqrt{n^{(1)}}(\hat{\btau}_1(\bS,\bW^{(1)})-\btau)',
\sqrt{n^{(2)}}(\hat{\btau}_2(\bS,\bW^{(2)})-\btau)')' \in \cA\big)\\
&-\bbP\big((\sqrt{n^{(1)}}(\hat{\btau}_1(\bS,\bZ^{(1)})-\btau)',
\sqrt{n^{(2)}}(\hat{\btau}_2(\bS,\bZ^{(2)})-\btau)')' \in \cA\mid \\
&\quad\quad\quad
\phi_M(\sqrt{n^{(1)}}\hat{\btau}_{\bx,1}(\bS,\bZ^{(1)}),\bV_{\bx\bx,1}(\bS))=1,\\
&\quad\quad\quad
\phi_\bomega(\sqrt{n^{(2)}}\hat{\btau}_{\bx,2}(\bS,\bZ^{(2)}),\bV_{\bx\bx,2}(\bS),\hat{\bB}(\bS,\bZ^{(1)}))=1\big)\Big|\rightarrow 0.
\end{split}\]
Hence, 
the asymptotic distribution of $(\sqrt{n^{(1)}}(\hat{\btau}_1(\bS,\bW^{(1)})-\btau),
\sqrt{n^{(2)}}(\hat{\btau}_2(\bS,\bW^{(2)})-\btau))$ is the same as that of 
\[\begin{split}
&(\sqrt{n^{(1)}}(\hat{\btau}_1(\bS,\bZ^{(1)})-\btau),
\sqrt{n^{(2)}}(\hat{\btau}_2(\bS,\bZ^{(2)})-\btau))
\mid \\
& \ \ \ \ \ \ \ \ \phi_M(\sqrt{n^{(1)}}\hat{\btau}_{\bx,1}(\bS,\bZ^{(1)}),\bV_{\bx\bx,1}(\bS))=1,
\phi_\bomega(\sqrt{n^{(2)}}\hat{\btau}_{\bx,2}(\bS,\bZ^{(2)}),\bV_{\bx\bx,2}(\bS),\hat{\bB}(\bS,\bZ^{(1)}))=1.
\end{split}\]
By Lemma \ref{lem:joint},  
\[\begin{split}
&(\sqrt{n^{(1)}}(\hat{\btau}_1(\bS,\bZ^{(1)})-\btau),\sqrt{n^{(1)}}\hat{\btau}_{\bx,1}(\bS,\bZ^{(1)}),\bV_{\bx\bx,1}(\bS),\\
&\ \ \ \ 
\sqrt{n^{(2)}}(\hat{\btau}_2(\bS,\bZ^{(2)})-\btau),\sqrt{n^{(2)}}\hat{\btau}_{\bx,2}(\bS,\bZ^{(2)}),\bV_{\bx\bx,2}(\bS),
\hat{\bB}(\bS,\bZ^{(1)})) \\
\overset{\cdot}{\sim}&(\bC_1,\bD_1,\bV_{\bx\bx},\bC_2,\bD_2,\bV_{\bx\bx},\bB),
\end{split}\]
where $(\bC_1',\bD_1',\bC_2',\bD_2')' \sim \cN(\bZero,\bV_{\TSCRFE})$,
we have 
\[\begin{split}
&(\sqrt{n^{(1)}}(\hat{\btau}_1(\bS,\bZ^{(1)})-\btau),
\sqrt{n^{(2)}}(\hat{\btau}_2(\bS,\bZ^{(2)})-\btau))
\mid \\
& \ \ \ \ \ \ \ \ \phi_M(\sqrt{n^{(1)}}\hat{\btau}_{\bx,1}(\bS,\bZ^{(1)}),\bV_{\bx\bx,1}(\bS))=1,
\phi_\bomega(\sqrt{n^{(2)}}\hat{\btau}_{\bx,2}(\bS,\bZ^{(2)}),\bV_{\bx\bx,2}(\bS),\hat{\bB}(\bS,\bZ^{(1)}))=1\\
\overset{\cdot}{\sim}&
(\bC_1,\bC_2)\mid \phi_M(\bD_1,\bV_{\bx\bx})=1,
\phi_\bomega(\bD_2,\bV_{\bx\bx},\bB)=1\\
\sim & ((\bC_1 \mid \phi_M(\bD_1,\bV_{\bx\bx})=1),
(\bC_2 \mid \phi_\bomega(\bD_2,\bV_{\bx\bx},\bB)=1)).
\end{split}\]
The definition of $\bV_\TSCRFE$ implies that 
\[\begin{split}
\cov(\bC_1, \bD_2) = \bZero_{F \times Fp},\ 
\cov(\bC_2, \bD_1) = \bZero_{F \times Fp},
\ \cov(\bD_1, \bD_2) = \bZero_{Fp \times Fp},
\end{split}\]
i.e., $\bC_1$ is independent of $\bD_2$, $\bC_2$ is independent of $\bD_1$, and $\bD_1$ is independent of $\bD_2$.
Hence, 
\[\begin{split}
&\left.\Big(\sqrt{n^{(1)}}(\hat{\btau}_1(\bS,\bZ^{(1)})-\btau),
\sqrt{n^{(2)}}(\hat{\btau}_2(\bS,\bZ^{(2)})-\btau)\Big)
\right| \\
& \quad\quad
\phi_M(\sqrt{n^{(1)}}\hat{\btau}_{\bx,1}(\bS,\bZ^{(1)}),\bV_{\bx\bx,1}(\bS))=1,\\
& \quad\quad
\phi_\bomega(\sqrt{n^{(2)}}\hat{\btau}_{\bx,2}(\bS,\bZ^{(2)}),\bV_{\bx\bx,2}(\bS),\hat{\bB}(\bS,\bZ^{(1)}))=1\\
\overset{\cdot}{\sim} & \big((\bC_1 \mid \phi_M(\bD_1,\bV_{\bx\bx})=1),
(\bC_2 \mid \phi_\bomega(\bD_2,\bV_{\bx\bx},\bB)=1)\big).
\end{split}\]
Moreover, recall that $\bB = \bV_{\bx\bx}^{-1}\bV_{\bx\btau}$,
then
\[\begin{split}
\bC_1 = \bB' \bD_1 + \bepsilon_1,
\ \bC_2 = \bB' \bD_2 + \bepsilon_2,
\end{split}\]
where $\bepsilon_1 \sim \cN(\bZero,\bV_{\btau\btau}+(1-\rho_n)\bS_{\btau\btau}-\bV_{\btau \bx}\bV_{\bx\bx}^{-1}\bV_{\bx\btau})$,  
$\bepsilon_2 \sim \cN(\bZero,\bV_{\btau\btau}+\rho_n\bS_{\btau\btau}-\bV_{\btau\bx}\bV_{\bx\bx}^{-1}\bV_{\bx\btau})$.
By definition,
$\bV_{\btau \bx}\bV_{\bx\bx}^{-1}\bV_{\bx\btau} 
= \bV_{\btau\btau}^{||}$.
Moreover, since $\cov(\bD_1, \bepsilon_1) = \bZero$,
and $\cov(\bD_2, \bepsilon_2) = \bZero$,
then $\bepsilon_1$ is independent of $\bD_1$, and $\bepsilon_2$ is independent of $\bD_2$.
Since $(\bC_1,\bD_1)$ is independent of $\bD_2$, then 
$\bepsilon_1$ is independent of $\bD_2$.
Similarly, $\bepsilon_2$ is independent of $\bD_1$.
Therefore,
\[\begin{split}
(\bC_1 \mid \phi_M(\bD_1,\bV_{\bx\bx})=1)
&\sim (\bB' \bD_1 \mid \bD_1'\bV_{\bx\bx}^{-1}\bD_1 \leq \xi_{Fp,\alpha}) + \bepsilon_1\\
&\sim 
(\bV_{\btau\btau}^{||})^{1/2}_{Fp}\cdot \bm\zeta_{Fp,\alpha}
+\bepsilon_1,
\end{split}\]
where $\bzeta_{Fp,\alpha} \sim \tilde{\bEta}|\tilde{\bEta}'\tilde{\bEta} \leq \xi_{Fp,\alpha}$ with $\tilde\bEta \sim \cN(\bZero,\bI_{Fp})$, and $\tilde\bEta$ is independent of $\bepsilon_1$.
Similarly, 
\[\begin{split}
(\bC_2 \mid \phi_\bomega(\bD_2,\bV_{\bx\bx},\bB)=1)
&\sim (\bB' \bD_2 \mid 
\bD_2'\bB\bQ_\bB'\bLambda^{-1}\bomega\bQ_\bB\bB'\bD_2 \leq 
\xi_{\bomega,\alpha}) + \bepsilon_2\\
&\sim \bQ_\bB^{-1}\bLambda^{1/2}\bm\eta \mid \bm\eta'\bomega\bm\eta \leq \xi_{\bomega,\alpha}
+ \bepsilon_2,
\end{split}\]
where $\bEta \sim \cN(\bZero,\bI_F)$ is independent of $\bepsilon_2$.
Since $\bD_1$ and $\bD_2$ are independent, then $\tilde\bEta$ and $\bEta$ are independent.
When $\bm\omega = c\bI$, we have 
\[\begin{split}
(\bC_2 \mid \phi_\bomega(\bD_2,\bV_{\bx\bx},\bB)=1)
&\sim (\bB' \bD_2 \mid 
\bD_2'\bB\bQ_\bB'\bLambda^{-1}\bomega\bQ_\bB\bB'\bD_2 \leq 
\xi_{\bomega,\alpha}) + \bepsilon_2\\
&\sim (\bV_{\btau\btau}^{||})^{1/2}\cdot \bm\zeta_{F,\alpha}
+ \bepsilon_2,
\end{split}\]
where $\bzeta_{F,\alpha} \sim {\bEta}|{\bEta}'{\bEta} \leq \xi_{F,\alpha}$.
Moreover, since
\[\begin{split}
\cov(\bC_1,\bC_2)
= \cov(\bB' \bD_1 + \bepsilon_1, \bB' \bD_2 + \bepsilon_2)
= \cov(\bepsilon_1, \bepsilon_2)
= -\sqrt{\rho_n(1-\rho_n)}\cdot \bS_{\btau\btau},
\end{split}\]
we have 
\[\begin{split}
\begin{pmatrix}
\bepsilon_1\\
\bepsilon_2
\end{pmatrix}\sim\cN\left(\begin{pmatrix}
\bZero \\ \bZero
\end{pmatrix},\begin{pmatrix}
\bV_{\btau\btau}+(1-\rho_n)\bS_{\btau\btau}-\bV_{\btau\btau}^{||} & 
-\sqrt{\rho_n(1-\rho_n)}\cdot \bS_{\btau\btau}\\
-\sqrt{\rho_n(1-\rho_n)}\cdot \bS_{\btau\btau} & 
\bV_{\btau\btau}+\rho_n\bS_{\btau\btau}-\bV_{\btau\btau}^{||}
\end{pmatrix}\right).
\end{split}\]
Hence,
\[\begin{split}
\begin{pmatrix}
\bepsilon_1\\
\bepsilon_2
\end{pmatrix}\sim\cN\left(\begin{pmatrix}
\bZero \\ \bZero
\end{pmatrix},\begin{pmatrix}
\bV_{\btau\btau}^{\perp}+(1-\rho_n)\bS_{\btau\btau} & 
-\sqrt{\rho_n(1-\rho_n)}\cdot \bS_{\btau\btau}\\
-\sqrt{\rho_n(1-\rho_n)}\cdot \bS_{\btau\btau} & 
\bV_{\btau\btau}^{\perp}+\rho_n \bS_{\btau\btau}
\end{pmatrix}\right).
\end{split}\]
We complete the proof.
\end{proof}
\end{proof}

\subsection{Proof of Corollary \ref{cor:distribution2_fac}}

\begin{proof}
By definition,
\[\begin{split}
&\sqrt{n}\big(\hat{\btau}(\bS,\bW^{(1)},\bW^{(2)})-\btau\big)\\
= &\sqrt{\rho_n}\cdot \sqrt{n^{(1)}}\big(\hat{\btau}_1(\bS,\bW^{(1)})-\btau\big)
+ \sqrt{1-\rho_n}\cdot 
\sqrt{n^{(2)}}\big(\hat{\btau}_2(\bS,\bW^{(2)})-\btau\big).
\end{split}\]
According to Theorem \ref{thm:distribution_fac},
the asymptotic distribution of $\sqrt{n}(\hat{\btau}(\bS,\bW^{(1)},\bW^{(2)})-\btau)$ under DA-ReO$_\F$ is the same as
\[\begin{split}
&\sqrt{\rho_n} \cdot\left[(\bV_{\btau\btau}^{||})_{Fp}^{1/2}\cdot \bm\zeta_{Fp,\alpha} +
\bvarepsilon_1\right]
+ \sqrt{1-\rho_n} \cdot\left[(\bQ_\bB^{-1}\bLambda^{1/2}\bm\eta \mid \bm\eta'\bomega\bm\eta \leq \xi_{\bomega,\alpha}) +
\bvarepsilon_2\right]\\
\sim & \sqrt{\rho_n}(\bV_{\btau\btau}^{||})_{Fp}^{1/2}\cdot \bm\zeta_{Fp,\alpha} +
\sqrt{1-\rho_n}(\bQ_\bB^{-1}\bLambda^{1/2}\bm\eta \mid \bm\eta'\bomega\bm\eta \leq \xi_{\bomega,\alpha})
+
(\bV_{\btau\btau}^{\perp})^{1/2}\cdot \bvarepsilon.
\end{split}\]
For DA-ReO$_\FE$, the asymptotic distribution is 
\[\begin{split}
&\sqrt{\rho_n} \cdot\left[(\bV_{\btau\btau}^{||})_{Fp}^{1/2}\cdot \bm\zeta_{Fp,\alpha} +
\bvarepsilon_1\right]
+ \sqrt{1-\rho_n} \cdot\left[(\bV_{\btau\btau}^{||})^{1/2}\cdot \bm\zeta_{p,\alpha} +
\bvarepsilon_2\right]\\
\sim & \sqrt{\rho_n}(\bV_{\btau\btau}^{||})^{1/2}_{Fp}\cdot \bm\zeta_{Fp,\alpha} +
\sqrt{1-\rho_n}(\bV_{\btau\btau}^{||})^{1/2}\cdot \bm\zeta_{p,\alpha} +
(\bV_{\btau\btau}^{\perp})^{1/2}\cdot \bvarepsilon.
\end{split}\]
\end{proof}

\subsection{Proof of Corollary \ref{cor:PRIASV2_fac}}
Corollary~\ref{cor:distribution2_fac} and the proof of  Theorem~\ref{thm:PRIASV} show that, under DA-ReO$_\F$, the asymptotic sampling covariance matrix of $\sqrt{n}(\hat{\btau}-\btau)$ is given by the limit of the following expression:
\[\begin{split}
&\bV_{\btau\btau}^{\perp} + \rho_n \bV_{\btau\btau}^{||}\cdot v_{Fp,\alpha} + (1-\rho_n)\sum_{f=1}^F c_f \bV_{\btau\btau}^{||}[f]
\\=& \bV_{\btau\btau}^{\perp} + \sum_{f=1}^F \left[\rho_n v_{Fp,\alpha} + (1-\rho_n)c_f\right]
\bV_{\btau\btau}^{||}[f].
\end{split}\]
Therefore, the reduction compared with CRFE is 
\[\begin{split}
\lim_{n\to\infty} \sum_{f=1}^F \left[1-\rho_n v_{Fp,\alpha} - (1-\rho_n)c_f\right]
\bV_{\btau\btau}^{||}[f].
\end{split}\]
Under DA-ReO$_\FE$, $c_f = v_{F,\alpha}$ for all $f = 1,\ldots,F$.
Hence, the reduction in the asymptotic covariance is 
\[\begin{split}
\lim_{n\to\infty}\left[1-\rho_n v_{Fp,\alpha} - (1-\rho_n)v_{F,\alpha}\right]
\bV_{\btau\btau}^{||}.
\end{split}\]

Under DA-ReO$_\F$,
the reduction in asymptotic sampling variance of $\sqrt{n}(\hat{\tau}_f-\tau_f)$ is given by the limit of the following expression:
\[\begin{split}
&\be_f'\left(\sum_{j=1}^F \left[1-\rho_n v_{Fp,\alpha} - (1-\rho_n)c_j\right]
\bV_{\btau\btau}^{||}[j]\right)\be_f\\
=& \sum_{j=1}^F \left[1-\rho_n v_{Fp,\alpha} - (1-\rho_n)c_j\right]
\be_f'\bV_{\btau\btau}^{||}[j]\be_f\\
=& \sum_{j=1}^F \left[1-\rho_n v_{Fp,\alpha} - (1-\rho_n)c_j\right]
R^2_f[j]\cdot \bbV(\sqrt{n}(\hat{\tau}_f-\tau_f)).
\end{split}\]
By Lemma \ref{lem:R2}, $R^2_f[j] = 0$ for $j > f$, the PRIASV of $\sqrt{n}(\hat{\tau}_f-\tau_f)$ is 
\[\begin{split}
&\lim_{n\to\infty}100\times\sum_{j=1}^F \left[1-\rho_n v_{Fp,\alpha} - (1-\rho_n)c_j\right]
R^2_f[j] \\=& \lim_{n\to\infty}100\times\sum_{j=1}^f \left[1-\rho_n v_{Fp,\alpha} - (1-\rho_n)c_j\right]
R^2_f[j].
\end{split}\]
Under DA-ReO$_\FE$, the PRIASV is 
\[
\begin{split}
&\lim_{n\to\infty}100\times \sum_{j=1}^f \left[1-\rho_n v_{Fp,\alpha}-(1-\rho_n)v_{F,\alpha}\right]R_f^2[j] \\
=&\lim_{n\to\infty}100\times \left[1-\rho_n v_{Fp,\alpha}-(1-\rho_n)v_{F,\alpha}\right]R_f^2.
\end{split}
\]
Moreover, for $\omega_1\geq \ldots \geq \omega_f$,
we have $c_1\leq \ldots \leq c_F$, and thus, the lower bound of the PRIASV under DA-ReO$_\F$ is 
\[
\begin{split}
&\lim_{n\to\infty}100\times \sum_{j=1}^f \left[1-\rho_n v_{Fp,\alpha}-(1-\rho_n)c_f\right]R_f^2[j] \\
=&\lim_{n\to\infty}100\times \left[1-\rho_n v_{Fp,\alpha}-(1-\rho_n)c_f\right]R_f^2.
\end{split}
\]
We complete the proof.

\subsection{Proof of Lemma \ref{lem:consistent_fac}}

\begin{proof}
Under S-CRFE, the distribution of $\bZ(\bS,\bZ^{(1)},\bZ^{(2)})$ under S-CRFE is identical to that of $\bZ$ under CRFE, where exactly $n_q$ units are assigned to treatment group $q$ for $q=1,\ldots,Q$.

Let $A$ and $B$ be any finite population quantities that can be the  potential outcome under any treatment arm, or any coordinate of covariates $\bX$, then for any finite population $\{(A_i,B_i):i=1,\ldots,n\}$, let the finite population averages be $\bar{A} = n^{-1}\sum_{i=1}^n A_i$ and $\bar{B} = n^{-1}\sum_{i=1}^n B_i$, let the sample average in the treatment group $q$ be $\bar{A}_q(\bZ) = \bar{A}_q(\bS,\bZ^{(1)},\bZ^{(2)}) = n_q^{-1}\sum_{i:Z_i=q} A_i$ and $\bar{B}_q(\bZ) = \bar{B}_q(\bS,\bZ^{(1)},\bZ^{(2)}) = n_q^{-1}\sum_{i:Z_i=q} B_i$.
Define the sample covariance under treatment $q$ as $s_{AB,q}(\bZ) = s_{AB,q}(\bS,\bZ^{(1)},\bZ^{(2)}) = (n_q-1)^{-1}\sum_{i:Z_i=q} (A_i - \bar{A}_q(\bZ))(B_i - \bar{B}_q(\bZ))$.
Then Lemma B3 in \cite{yang2023rejective} implies that under S-CRFE,
\[\begin{split}
\bbV(s_{AB,q}(\bS,\bZ^{(1)},\bZ^{(2)})) &\leq \frac{4 n_q}{(n_q-1)^2}\max_{1\leq j\leq n}(A_j - \bar{A})^2 \cdot \frac{1}{n-1}\sum_{i=1}^n (B_i-\bar{B})^2.
\end{split}\]
Under Assumption \ref{assump:regularity_fac_full}, we have 
$\bbV(s_{AB,q}(\bS,\bZ^{(1)},\bZ^{(2)}))\rightarrow 0$.

Define procedure $\widetilde{\text{DA-ReO}}_\F$ (or $\widetilde{\text{DA-ReO}}_\FE$) as a single-stage procedure, under which sampling and treatment assignment vectors $(\bS,\bZ^{(1)},\bZ^{(2)})$ are acceptable if and only if $\phi_M(\sqrt{n^{(1)}}\hat{\btau}_{\bx,1}(\bS,\bZ^{(1)}),\bV_{\bx\bx,1}(\bS))=1$ and $\phi_\bomega(\sqrt{n^{(2)}}\hat{\btau}_{\bx,2}(\bS,\bZ^{(2)}),\bV_{\bx\bx,2}(\bS),
\hat{\bB}(\bS,\bZ^{(1)}))=1$ hold simultaneously.
This is a hypothetical procedure as it assumes that $\hat{\bB}(\bS,\bZ^{(1)})$ is available before the first-stage experiment is actually conducted.

By \citet{cochran1977sampling}, under S-CRFE, $s_{AB,q}(\bZ)$ is unbiased for $S_{AB}$.
Under Assumption \ref{assump:regularity_fac_full},
Eq.~\eqref{eq:prob1} implies that for $q=1,\ldots,Q$,
\[\begin{split}
&\bbE[(s_{AB,q}(\bS,\bZ^{(1)},\bZ^{(2)})-S_{AB})^2 \mid \widetilde{\text{DA-ReO}}_\F] \\
= & \bbE[(s_{AB,q}(\bS,\bZ^{(1)},\bZ^{(2)})-S_{AB})^2 \mid \phi_M(\sqrt{n^{(1)}}\hat{\btau}_{\bx,1}(\bS,\bZ^{(1)}),\bV_{\bx\bx,1}(\bS))=1,\\
&\ \ \ \ \ \ \ \phi_\bomega(\sqrt{n^{(2)}}\hat{\btau}_{\bx,2}(\bS,\bZ^{(2)}),\bV_{\bx\bx,2}(\bS),
\hat{\bB}(\bS,\bZ^{(1)}))=1]
\\
\leq &
\frac{\bbV\{s_{AB,q}(\bS,\bZ^{(1)},\bZ^{(2)})\}}{\bbP(\phi_M(\sqrt{n^{(1)}}\hat{\btau}_{\bx,1}(\bS,\bZ^{(1)}),\bV_{\bx\bx,1}(\bS))=1,
\phi_\bomega(\sqrt{n^{(2)}}\hat{\btau}_{\bx,2}(\bS,\bZ^{(2)}),\bV_{\bx\bx,2}(\bS),
\hat{\bB}(\bS,\bZ^{(1)}))=1)}\\
\rightarrow& 0,
\end{split}\]
i.e., $s_{AB,q}(\bS,\bZ^{(1)},\bZ^{(2)})-S_{AB} = o_p(1)$ under $\widetilde{\text{DA-ReO}}_\F$.
According to Lemma \ref{lem:equiv2},
for any $c>0$,
\[\begin{split}
&|\bbP(|s_{AB,q}(\bS,\bW^{(1)},\bW^{(2)})-S_{AB}| \geq c)-\bbP(|s_{AB,q}(\bS,\bZ^{(1)},\bZ^{(2)})-S_{AB}| \geq c\mid \widetilde{\text{DA-ReO}}_\F)|\\
=& |\bbP(|s_{AB,q}(\bS,\bW^{(1)},\bW^{(2)})-S_{AB}| \geq c)-\bbP(|s_{AB,q}(\bS,\bZ^{(1)},\bZ^{(2)})-S_{AB}| \geq c\mid \\& \ \ \ \ \ \ \ \ \ 
\phi_M(\sqrt{n^{(1)}}\hat{\btau}_{\bx,1}(\bS,\bZ^{(1)}),\bV_{\bx\bx,1}(\bS))=1,
\phi_\bomega(\sqrt{n^{(2)}}\hat{\btau}_{\bx,2}(\bS,\bZ^{(2)}),\bV_{\bx\bx,2}(\bS),
\hat{\bB}(\bS,\bZ^{(1)}))=1)|\\
\rightarrow & 0.
\end{split}\]
Hence, $s_{AB,q}(\bW) = s_{AB,q}(\bS,\bW^{(1)},\bW^{(2)}) =S_{AB} + o_p(1)$ under DA-ReO$_\F$ (or DA-ReO$_\FE$).
Therefore, for $q=1,\ldots,Q$ and under DA-ReO$_\F$ (or DA-ReO$_\FE$), we have 
\[
\begin{split}
s_{qq}-S_{qq}=o_p(1),\quad
\bms_{q,\bx}-\bS_{q,\bx}=o_p(1),\quad
\bms_{\bx\bx,q}-\bS_{\bx\bx}=o_p(1).
\end{split}
\]
We complete the proof.
\end{proof}

\subsection{Proof of Theorem \ref{thm:CI2}}
\begin{proof}
By the proof of Theorem~1 in \citet{li2020rerandomization}, 
$\bV_{\btau\btau}^{||} 
=\bV_{\btau\bx}\bV_{\bx\bx}^{-1}\bV_{\bx\btau} = \sum_{q=1}^Q \frac{1}{r_q}\bA_q\bA_q' \cdot
\bS_{q,\bx} \bS_{\bx\bx}^{-1} \bS_{\bx,q}-\bS_{\btau,\bx} \bS_{\bx\bx}^{-1} \bS_{\bx,\btau}$,
where $\bS_{\btau,\bx} = \sum_{q=1}^Q \bA_q \bS_{q,\bx}$.
Since $\bS_{\btau\btau}
\geq
\bS_{\btau,\bx}\bS_{\bx\bx}^{-1}\bS_{\bx,\btau}$, we have 
$\bV_{\btau\btau} = \sum_{q=1}^Q r_q^{-1}\bA_q \bA_q' S_{qq} - \bS_{\btau\btau} \leq \sum_{q=1}^Q r_q^{-1}\bA_q \bA_q' S_{qq} - \bS_{\btau,\bx} \bS_{\bx\bx}^{-1} \bS_{\bx,\btau}.$
Therefore, Lemma~\ref{lem:consistent_fac} yields the following asymptotically conservative estimator of $\bV_{\btau\btau}^{\perp}=\bV_{\btau\btau}-\bV_{\btau\btau}^{||}$:
$$\hat\bV_{\btau\btau}^{\perp} = \sum_{q=1}^Q r_q^{-1}\bA_q \bA_q' (s_{qq} - \bms_{q,\bx} \bms_{\bx\bx,q}^{-1} \bms_{\bx,q}).$$
If $\bms_{\bx\bx,q}$ is singular, we define $\bms_{\bx\bx,q}^{-1}=\bZero$.
By Lemma \ref{lem:consistent_fac}, $\hat{\bV}_{\btau\bx}-\bV_{\btau\bx}=o_p(1)$, and thus,
$\hat{\bB}-\bB = o_p(1)$, 
$\hat{\bV}_{\btau\btau}^{||}-{\bV}_{\btau\btau}^{||} = o_p(1)$, and 
$\hat{\bV}_{\btau\btau}^\perp-{\bV}_{\btau\btau}^\perp - \bS_{\btau\backslash\bx} = o_p(1)$, where $\bS_{\btau\backslash\bx}=\bS_{\btau\btau}-\bS_{\btau,\bx} \bS_{\bx\bx}^{-1} \bS_{\bx,\btau}$.
By Lemma \ref{lem:inverse_lemma}, we have 
$\bQ_{\hat{\bB}} - \bQ_{{\bB}} = o_p(1)$, $\hat{\bLambda} - \bLambda = o_p(1)$, and for $f=1,\ldots,F$, 
$\hat{\bV}_{\btau\btau}^{||}[f] - \bV_{\btau\btau}^{||}[f] = o_p(1)$. 
Therefore, the covariance estimator Eq.~\eqref{eq:variance1} is consistent for 
\[\begin{split}
\bC\bS_{\btau\backslash\bx}\bC' + \bC\bV_{\btau\btau}^{\perp}\bC' + 
\sum_{f=1}^F \left(\rho_nv_{Fp,\alpha} + (1-\rho_n)c_f\right)\bC\bV_{\btau\btau}^{||}[f]\bC',
\end{split}\]
which is
larger than or equal to $ \bC\bV_{\btau\btau}^{\perp}\bC' + 
\sum_{f=1}^F \left(\rho_nv_{Fp,\alpha} + (1-\rho_n)c_f\right)\bC\bV_{\btau\btau}^{||}[f]\bC'$, whose limit is the true asymptotic sampling variance of $\sqrt{n}\bC\hat{\btau}$ under DA-ReO$_\F$.
Similarly, the covariance estimator Eq.~\eqref{eq:variance2} consistently estimates
\[\begin{split}
\bC\bS_{\btau\backslash\bx}\bC' + \bC\bV_{\btau\btau}^{\perp}\bC' + 
\left(\rho_nv_{Fp,\alpha} + (1-\rho_n)v_{F,\alpha}\right)\bC\bV_{\btau\btau}^{||}\bC',
\end{split}\]
which is
larger than or equal to 
$\bC\bV_{\btau\btau}^{\perp}\bC' + 
\left(\rho_nv_{Fp,\alpha} + (1-\rho_n)v_{F,\alpha}\right)\bC\bV_{\btau\btau}^{||}\bC'$, whose limit is the true asymptotic sampling variance of $\sqrt{n}\bC\hat{\btau}$ under DA-ReO$_\FE$.

Let $\tilde{\bV}_{\btau\btau}^{\perp} = {\bV}_{\btau\btau}^{\perp} + \bS_{\btau\backslash\bx}$.
To prove the asymptotic conservativeness of confidence sets, we define random variables 
\[\begin{split}
\mathcal{\bm L} &\sim \sqrt{\rho_n}(\bV_{\btau\btau}^{||})^{1/2}_{Fp}\cdot \bm\zeta_{Fp,\alpha} + \sqrt{1-\rho_n}(\bQ_{\bB}^{-1}\bLambda^{1/2}\bm\eta \mid \bm\eta'\bomega\bm\eta \leq \xi_{\bomega,\alpha}) + (\bV_{\btau\btau}^{\perp})^{1/2}\cdot \bvarepsilon,\\
\tilde{\mathcal{\bm L}} &\sim \sqrt{\rho_n}(\bV_{\btau\btau}^{||})^{1/2}_{Fp}\cdot \bm\zeta_{Fp,\alpha} + \sqrt{1-\rho_n}(\bQ_\bB^{-1}\bLambda^{1/2}\bm\eta \mid \bm\eta'\bomega\bm\eta \leq \xi_{\bomega,\alpha}) + (\tilde\bV_{\btau\btau}^{\perp})^{1/2}\cdot \bvarepsilon,\\
\hat{\mathcal{\bm L}} &\sim \sqrt{\rho_n}(\hat\bV_{\btau\btau}^{||})^{1/2}_{Fp}\cdot \bm\zeta_{Fp,\alpha} + \sqrt{1-\rho_n}(\bQ_{\hat\bB}^{-1}\hat\bLambda^{1/2}\bm\eta \mid \bm\eta'\bomega\bm\eta \leq \xi_{\bomega,\alpha}) + (\hat\bV_{\btau\btau}^{\perp})^{1/2}\cdot \bvarepsilon,\\
\end{split}\]
where $\bvarepsilon \sim \cN(0,\bI_F)$, $\bzeta_{Fp,\alpha} \sim \tilde{\bEta}|\tilde{\bEta}'\tilde{\bEta} \leq \xi_{Fp,\alpha}$, $\tilde\bEta \sim \cN(\bZero,\bI_{Fp})$, and $\bEta\sim \cN(\bZero,\bI_{F})$.
Moreover, $\bvarepsilon$, $\tilde\bEta$, and $\bEta$
are mutually independent.
It is straightforward to verify that the density of $\bEta \mid \bEta'\bomega\bEta \le \xi_{\bomega,\alpha}$ is log-concave. Therefore, by Lemma A8 of \citet{li2020rerandomization}, $\bEta \mid \bEta'\bomega\bEta \le \xi_{\bomega,\alpha}$ is central convex unimodal.
Proposition 2, Lemma A6 and Lemma A22 in \cite{li2020rerandomization} imply that $\mathcal{\bm L}$, $\tilde{\mathcal{\bm L}}$ and $\hat{\mathcal{\bm L}}$ are central convex unimodal, and ${\mathcal{\bm L}} \succ \tilde{\mathcal{\bm L}}$.
Under Assumption \ref{assump:regularity_fac_full}, by Slutsky's theorem, $\hat{\mathcal{\bm L}} \overset{\cdot}{\sim} \tilde{\mathcal{\bm L}}$, and 
\[\begin{split}
\big(\bC \hat\bV_{\btau\btau}^{\perp}\bC'\big)^{-1/2}\bC\hat{\mathcal{\bm L}}
\overset{\cdot}{\sim}
\big(\bC \tilde\bV_{\btau\btau}^{\perp}\bC'\big)^{-1/2}\bC\tilde{\mathcal{\bm L}}.
\end{split}\]
From the continuous mapping theorem,
\[\begin{split}
(\bC\hat{\mathcal{\bm L}})'
\big(\bC \hat\bV_{\btau\btau}^{\perp}\bC'\big)^{-1}\bC\hat{\mathcal{\bm L}}
\overset{\cdot}{\sim}
(\bC\tilde{\mathcal{\bm L}})'
\big(\bC \tilde\bV_{\btau\btau}^{\perp}\bC'\big)^{-1}\bC\tilde{\mathcal{\bm L}}.
\end{split}\]
Thus, the $1-\delta$ quantile of $(\bC\hat{\mathcal{\bm L}})'
\big(\bC \hat\bV_{\btau\btau}^{\perp}\bC'\big)^{-1}\bC\hat{\mathcal{\bm L}}$, $\hat{c}_{1-\delta}$, is consistent for the $1-\delta$ quantile of $(\bC\tilde{\mathcal{\bm L}})'
\big(\bC \tilde\bV_{\btau\btau}^{\perp}\bC'\big)^{-1}\bC\tilde{\mathcal{\bm L}}$, $\tilde{c}_{1-\delta}$.
Since ${\mathcal{\bm L}} \succ \tilde{\mathcal{\bm L}}$, 
and the set of form 
$\{\bmu:(\bC\bmu)'
\big(\bC \tilde\bV_{\btau\btau}^{\perp}\bC'\big)^{-1}\bC\bmu \leq c\}$ is symmetric convex, then the $1-\delta$ quantile of $(\bC{\mathcal{\bm L}})'
\big(\bC \tilde\bV_{\btau\btau}^{\perp}\bC'\big)^{-1}\bC{\mathcal{\bm L}}$, ${c}_{1-\delta}$, is smaller than or equal to $\tilde{c}_{1-\delta}$.
Above all, $\hat{c}_{1-\delta}$ is consistent for $\tilde{c}_{1-\delta} \geq c_{1-\delta}$.
Therefore, the $1-\delta$ confidence set for $\sqrt{n}\bC{\btau}$ is asymptotically conservative.
When $\bS_{\btau\backslash\bx}= o(1)$, we have  $\bV^\perp_{\btau\btau}- \tilde\bV^\perp_{\btau\btau} = o(1)$, which implies ${\mathcal{\bm L}} \overset{\cdot}{\sim} \tilde{\mathcal{\bm L}}$. Thus, $\tilde{c}_{1-\delta} - c_{1-\delta} = o(1)$, and the $1-\delta$ confidence set for $\sqrt{n}\bC\btau$ becomes asymptotically exact.
\end{proof}

\newpage 
\section{Additional simulation results}\label{sec:simulations}
\setcounter{figure}{0}
\renewcommand{\thefigure}{B\arabic{figure}}
\setcounter{table}{0}
\renewcommand{\thetable}{B\arabic{table}}

{
\begin{longtable}[]{@{}ccccccccc@{}}
\caption{Average empirical coverage probability of the resulting confidence intervals (\%).}\label{tab:simulation_CP}
\tabularnewline
\toprule\noalign{}
& & \multicolumn{3}{c}{Main}& \multicolumn{4}{c}{Interaction} \\
\cmidrule(lr){3-5}
\cmidrule(lr){6-9}
Setting & Scheme & 1 & 2 & 3 & 12 & 13 & 23 & 123 \\
\midrule\noalign{}
\endfirsthead
\toprule\noalign{}
& & \multicolumn{3}{c}{Main}& \multicolumn{4}{c}{Interaction} \\
\cmidrule(lr){3-5}
\cmidrule(lr){6-9}
Setting & Scheme & 1 & 2 & 3 & 12 & 13 & 23 & 123 \\
\midrule\noalign{}
\endhead
\bottomrule\noalign{}
\endlastfoot
\multirow{5}{*}{\shortstack{Linear\\
\hspace{\fill}\\
(same\\
\hspace{\fill}\\
importance)}} &  ReO$_\FE$ &  95.97 & 95.86 & 95.91 & 95.92 & 95.96 & 95.89 & 95.90 \\ 
& DA-ReO$_\FE$ ($\rho_n=0.2$) & 96.70 & 96.65 & 96.63 & 96.62 & 96.70 & 96.63 & 96.63 \\ 
& DA-ReO$_\FE$ ($\rho_n=0.3$) & 96.99 & 96.92 & 96.97 & 96.96 & 96.99 & 96.94 & 96.92 \\ 
& DA-ReO$_\FE$ ($\rho_n=0.4$) & 97.26 & 97.25 & 97.29 & 97.22 & 97.26 & 97.23 & 97.24 \\ 
& ReFM & 95.83 & 95.91 & 95.90 & 95.92 & 95.81 & 95.94 & 95.80 \\ 
\cmidrule(lr){1-9}
\multirow{5}{*}{\shortstack{Linear\\
\hspace{\fill}\\
(varying\\
\hspace{\fill}\\
importance)}} &  ReO$_\F$ & 96.01 & 95.94 & 95.86 & 95.71 & 95.67 & 95.95 & 95.71 \\ 
& DA-ReO$_\F$ ($\rho_n=0.2$) & 96.96 & 97.02 & 96.84 & 96.04 & 96.03 & 96.30 & 96.16 \\ 
& DA-ReO$_\F$ ($\rho_n=0.3$) & 97.36 & 97.39 & 97.25 & 96.18 & 96.18 & 96.49 & 96.34 \\ 
& DA-ReO$_\F$ ($\rho_n=0.4$) & 97.71 & 97.72 & 97.61 & 96.32 & 96.31 & 96.63 & 96.51 \\ 
& ReFMT$_\F$ & 96.01 & 95.98 & 95.93 & 95.79 & 95.69 & 95.88 & 95.79 \\ 
\cmidrule(lr){1-9}
\multirow{5}{*}{\shortstack{Nonlinear\\
\hspace{\fill}\\
(same\\
\hspace{\fill}\\
importance)}} &  ReO$_\FE$ & 95.61 & 95.70 & 95.82 & 95.80 & 95.98 & 95.74 & 95.73 \\ 
& DA-ReO$_\FE$ ($\rho_n=0.2$) & 96.41 & 96.44 & 96.52 & 96.56 & 96.70 & 96.47 & 96.49 \\ 
& DA-ReO$_\FE$ ($\rho_n=0.3$) & 96.75 & 96.77 & 96.83 & 96.86 & 97.00 & 96.81 & 96.81 \\ 
& DA-ReO$_\FE$ ($\rho_n=0.4$) & 97.06 & 97.03 & 97.14 & 97.15 & 97.25 & 97.09 & 97.09 \\ 
& ReFM & 95.75 & 95.66 & 95.79 & 95.80 & 95.95 & 95.75 & 95.80 \\ 
\cmidrule(lr){1-9}
\multirow{5}{*}{\shortstack{Nonlinear\\
\hspace{\fill}\\
(varying\\
\hspace{\fill}\\
importance)}} &  ReO$_\F$ & 95.49 & 95.65 & 95.83 & 95.91 & 95.89 & 95.65 & 95.67 \\ 
& DA-ReO$_\F$ ($\rho_n=0.2$) & 96.63 & 96.81 & 96.94 & 96.31 & 96.23 & 96.03 & 95.95 \\ 
& DA-ReO$_\F$ ($\rho_n=0.3$) & 97.00 & 97.23 & 97.38 & 96.49 & 96.34 & 96.18 & 96.12 \\ 
& DA-ReO$_\F$ ($\rho_n=0.4$) & 97.33 & 97.59 & 97.75 & 96.65 & 96.46 & 96.34 & 96.24 \\ 
& ReFMT$_\F$ & 95.89 & 95.80 & 95.88 & 95.80 & 95.93 & 95.78 & 95.74 \\ 
\end{longtable}
}

{
\begin{longtable}[]{@{}ccccccccc@{}}
\caption{Average percentage
reduction in confidence interval length relative to ReFM or ReFMT$_\F$ (\%).}\label{tab:simulation_Len}
\tabularnewline
\toprule\noalign{}
& & \multicolumn{3}{c}{Main}& \multicolumn{4}{c}{Interaction} \\
\cmidrule(lr){3-5}
\cmidrule(lr){6-9}
Setting & Scheme & 1 & 2 & 3 & 12 & 13 & 23 & 123 \\
\midrule\noalign{}
\endfirsthead
\toprule\noalign{}
& & \multicolumn{3}{c}{Main}& \multicolumn{4}{c}{Interaction} \\
\cmidrule(lr){3-5}
\cmidrule(lr){6-9}
Setting & Scheme & 1 & 2 & 3 & 12 & 13 & 23 & 123 \\
\midrule\noalign{}
\endhead
\bottomrule\noalign{}
\endlastfoot
\multirow{4}{*}{\shortstack{Linear\\
\hspace{\fill}\\
(same\\
\hspace{\fill}\\
importance)}} &  ReO$_\FE$ & 15.19 & 15.36 & 15.26 & 15.33 & 15.10 & 15.21 & 15.35 \\ 
& DA-ReO$_\FE$ ($\rho_n=0.2$) & 11.85 & 11.97 & 11.95 & 12.00 & 11.81 & 11.89 & 11.99 \\ 
& DA-ReO$_\FE$ ($\rho_n=0.3$) & 10.24 & 10.39 & 10.36 & 10.36 & 10.23 & 10.30 & 10.39 \\ 
& DA-ReO$_\FE$ ($\rho_n=0.4$) & 8.65 & 8.79 & 8.78 & 8.80 & 8.66 & 8.73 & 8.83 \\ 
\cmidrule(lr){1-9}
\multirow{4}{*}{\shortstack{Linear\\
\hspace{\fill}\\
(varying\\
\hspace{\fill}\\
importance)}} &  ReO$_\F$ & 18.38 & 18.87 & 18.29 & 12.71 & 12.51 & 12.76 & 13.55 \\ 
& DA-ReO$_\F$ ($\rho_n=0.2$) & 13.58 & 14.04 & 13.46 & 11.23 & 10.97 & 11.25 & 11.77 \\ 
& DA-ReO$_\F$ ($\rho_n=0.3$) & 11.27 & 11.76 & 11.20 & 10.53 & 10.24 & 10.47 & 10.96 \\ 
& DA-ReO$_\F$ ($\rho_n=0.4$) & 9.04 & 9.50 & 8.98 & 9.77 & 9.55 & 9.81 & 10.14 \\ 
\cmidrule(lr){1-9}
\multirow{4}{*}{\shortstack{Nonlinear\\
\hspace{\fill}\\
(same\\
\hspace{\fill}\\
importance)}} &  ReO$_\FE$ & 15.00 & 15.15 & 15.25 & 15.00 & 15.32 & 15.25 & 15.24 \\ 
& DA-ReO$_\FE$ ($\rho_n=0.2$) & 11.76 & 11.85 & 11.94 & 11.71 & 11.97 & 11.93 & 11.96 \\ 
& DA-ReO$_\FE$ ($\rho_n=0.3$) & 10.15 & 10.27 & 10.38 & 10.11 & 10.39 & 10.35 & 10.33 \\ 
& DA-ReO$_\FE$ ($\rho_n=0.4$) & 8.60 & 8.74 & 8.78 & 8.58 & 8.79 & 8.77 & 8.78 \\ 
\cmidrule(lr){1-9}
\multirow{4}{*}{\shortstack{Nonlinear\\
\hspace{\fill}\\
(varying\\
\hspace{\fill}\\
importance)}} &  ReO$_\F$ & 18.01 & 18.34 & 18.76 & 12.98 & 12.98 & 12.89 & 12.63 \\ 
& DA-ReO$_\F$ ($\rho_n=0.2$) & 13.31 & 13.56 & 14.02 & 11.38 & 11.40 & 11.29 & 11.11 \\ 
& DA-ReO$_\F$ ($\rho_n=0.3$) & 11.08 & 11.27 & 11.74 & 10.53 & 10.65 & 10.56 & 10.31 \\ 
& DA-ReO$_\F$ ($\rho_n=0.4$) & 8.87 & 9.08 & 9.53 & 9.74 & 9.82 & 9.84 & 9.63 \\ 
\end{longtable}
}

\end{document}